\documentclass[11pt]{article}
\usepackage{mathrsfs}% 用于产生一种数学用的花体字
\usepackage{CJK}
\usepackage{bm}
\usepackage{graphicx}
\usepackage{subfigure}
\usepackage{amsmath,amssymb}% AMSLaTeX宏包 用来排出更加漂亮的公式
\usepackage{amsfonts}
\usepackage{amsthm}
\usepackage{verbatim}
\usepackage{color}% 支持彩色
\usepackage{graphicx}
\usepackage{amssymb}
\usepackage{latexsym}
\allowdisplaybreaks
\large\normalsize

\newcommand{\be}{\begin{equation}}
\newcommand{\ee}{\end{equation}}
\newcommand{\bes}{\begin{equation}\begin{aligned}}
\newcommand{\ees}{\end{aligned}\end{equation}}
\newcommand{\ben}{\begin{equation}\nonumber\begin{aligned}}

 \newcommand{\R}{\mathbb{R}}

\newcommand{\Z}{\mathbb{Z}}

\renewcommand{\leq}{\leqslant}
\renewcommand{\geq}{\geqslant}

\newtheorem{thm}{Theorem}[section]
\newtheorem{lem}[thm]{Lemma}

\newtheorem{pro}[thm]{Proposition}
\newtheorem{rmk}[thm]{Remark}
\newtheorem{defi}[thm]{Definition}
\newtheorem*{main thm}{Main Theorem}

\numberwithin{equation}{section}
\begin{document}

\title{Recovery  of signals  by a weighted $\ell_2/\ell_1$ minimization under  arbitrary prior support information}

\author{Wengu~Chen
\thanks{W. Chen is with Institute of Applied Physics and Computational Mathematics,
Beijing, 100088, China, e-mail: chenwg@iapcm.ac.cn.}% <-this % stops a space
\thanks{This work was supported by the NSF of China (Nos.11271050, 11371183)
.} }
\author{Wengu~Chen,  Huanmin~Ge
\thanks{W. Chen is with Institute of Applied Physics and Computational Mathematics,
Beijing, 100088, China, e-mail: chenwg@iapcm.ac.cn.}% <-this % stops a space
\thanks{H. Ge is with Graduate School, China Academy of Engineering Physics,
Beijing, 100088, China, e-mail: gehuanmin@163.com.}% <-this % stops a space
\thanks{This work was supported by the NSF of China (Nos.11271050, 11371183)
.} }

\maketitle

\begin{abstract}
 In this paper,  we introduce a  weighted $\ell_2/\ell_1$ minimization to  recover block sparse signals
 with arbitrary prior support information.
When partial prior support information is available, a sufficient condition based on the high order block RIP is derived
 to  guarantee stable and robust recovery of block sparse signals via the weighted $\ell_2/\ell_1$ minimization.
  We then show if the accuracy of
arbitrary prior block support estimate is at least $50\%$,  the sufficient
recovery condition by the weighted $\ell_2/\ell_{1}$ minimization is weaker
than that by the  $\ell_2/\ell_{1}$ minimization, and the
 weighted $\ell_2/\ell_{1}$ minimization  provides better upper
bounds on the recovery error in terms of the measurement noise and
the compressibility of the signal. Moreover, we illustrate the
advantages of the weighted $\ell_2/\ell_1$ minimization approach  in
the recovery performance of block sparse signals under uniform and
non-uniform prior information by extensive numerical
 experiments. The significance of the results lies in the facts that making
 explicit use of block sparsity  and partial support information of block sparse
 signals can achieve  better recovery performance than handling the signals as being in the conventional sense,
thereby ignoring the additional structure  and  prior  support information in the problem.
\end{abstract}


\section{Introduction}

Compressed sensing, a new type of sampling
theory,  aims at recovering an unknown high dimensional sparse signal $x\in\R^N$,
through  the following linear measurement
\begin{align}\label{mod1}
y=Ax+z
\end{align}
where  $A\in\R^{n\times N}$ ($n\ll N$ ) is a sensing matrix,  $y\in
\mathbb{R}^{n}$ is a vector of measurements and $z \in
\mathbb{R}^{n}$ is the measurement error.  In last decade,
compressed sensing has been a fast growing field of research. A
multitude of  different recovery algorithms including the $\ell_1$
minimization \cite{CWX1}-\cite{CRT1}, \cite{ML},
 greedy  algorithm \cite{CG,DM,D1,DDTS,T,WKS},
 \cite{M}-\cite{NV}
 and iterative threshold algorithm \cite{BD1,BD2,DDD,DDFG,D2,LXY} have been
used to recover the sparse signal $x$ under a variety of different conditions on the sensing matrix $A$.

In this paper,  the unknown sparse signal $x$ of the model \eqref{mod1} has additional structure, whose nonzero coefficients occur  in  blocks (or clusters).
Such signal is called block sparse signal \cite{EM}, \cite{EKB}.
The recovery of block sparse signals naturally arise in practical examples such as equalization of sparse communication channels \cite{CR},
DNA  microarrays \cite{PVMH}, multiple measurement vector (MMV) problem \cite{CH}, \cite{CREKD}, \cite{ME}.
A block  signal $x\in\R^N$ over $\mathcal{I}=\{d_1,d_2,\ldots,d_M\}$ is a concatenation of $M$ blocks of length $d_i\ (i=1,2,\cdots,M)$, that is,
\begin{eqnarray}\label{m1}
x=(\underbrace{x_1, \ldots, x_{d_1}}_{x'[1]}, \underbrace{x_{d_1+1}, \ldots,  x_{d_1+d_2}}_{x'[2]}, \ldots, \underbrace{x_{N-d_M+1}, \ldots, x_N}_{x'[M]})^{'}
\end{eqnarray}
where $x[i]$ denotes the $i$th block of $x$ and $N=\sum\limits_{i=1}^Md_i$. Let the block index  set  $[M]=\{1,2,\ldots,M\}$.
The block signal $x$ is referred to  block $k$-sparse  if $x[i]$ has nonzero $\ell_2$ norm for at most $k$ indices $i\in[M]$,
 i.e., $\sum\limits_{i=1}^MI(\|x[i]\|_2>0)\leq k$, where $I(\cdot)$ is an indicator function.
Denote $\|x\|_{2,0}=\sum_{i=1}^MI(\|x[i]\|_2>0)$ and $T=\textmd{b-supp}(x)=\{i\in[M]:\|x[i]\|_2>0\}$,
then a block $k$-sparse signal $x$ satisfies  $\|x\|_{2,0}\leq k$ and $|T|\leq k$. If $d_i=1$ for all $i\in [M]$,
the block sparse signal $x$ reduces to the conventional sparse signal \cite{CRT2}, \cite{D3}.
Similar to \eqref{m1},
sensing matrix $A$ can be expressed
as a concatenation of $M$ column blocks over $\mathcal{I}=\{d_1,d_2,\ldots,d_M\}$
\begin{eqnarray*}
 A=[\underbrace{A_1\ldots A_{d_1}}_{A[1]} \underbrace{A_{d_1+1}\ldots A_{d_1+d_2}}_{A[2]} \ldots \underbrace{A_{N-d_M+1}\ldots A_N}_{A[M]}],
\end{eqnarray*}
where $A_i$ is the $i$th column of $A$ for $i=1,2,\cdots,N$.

To reconstruct the block sparse  signal $x$ in \eqref{m1},  researchers  explicitly take this block structure into account.
One of the efficient methods is  the following  $\ell_{2}/\ell_{1}$ minimization
\begin{align}\label{pro3}
  \min_{x\in \R^N} \quad\|x\|_{2,1} \ \ \ {\rm subject\quad to}\ \ \
  \|y-Ax\|_2\leq\epsilon
\end{align}
where $\|x\|_{2,1}=\sum\limits_{i=1}^M\|x[i]\|_2$.  To study
 the uniqueness and stability of the   $\ell_{2}/\ell_{1}$ minimization method, Eldar and Mishali introduced the notion
of block restricted isometry property in \cite{EM}, which is a generalization of the standard RIP \cite{CT}.
\begin{defi}
Let $A\in \mathbb{R}^{n\times N}$ be a matrix. Then $A$ has the $k$ order  block  restricted isometry property (block RIP)
 over $\mathcal{I}=\{d_1,d_2,\ldots,d_M\}$ with parameter $\delta^\mathcal{I} \in [0,1)$
if for all block $k$-sparse vector $x\in\R^N$ over $\mathcal{I}$
it holds that
\begin{eqnarray}\label{brip}
(1-\delta^\mathcal{I})\|x\|_2^2\leq \|Ax\|_2^2\leq(1+\delta^\mathcal{I})\|x\|^2_2.
\end{eqnarray}
The smallest constant $\delta^\mathcal{I}$ is called block  restricted isometry  constant (block RIC) $\delta_{k}^\mathcal{I}$.
When $k$ is not an integer, we define $\delta_{k}^\mathcal{I}$ as $\delta_{\lceil k\rceil}^\mathcal{I}$.
\end{defi}

For block sparse signal recovery, sufficient conditions in term of the block RIP  have been introduced
 and studied in the literatures. For example, Eldar and Mishali proved that if the sensing matrix $A$
satisfies $\delta_{2k}^\mathcal{I}<\sqrt{2}-1$ then the  $\ell_2/\ell_1$ minimization can recover perfectly block
sparse signals in noiseless case and can well approximate the best block $k$-sparse approximation in \cite{EM}.
Later,  Lin and Li improved the bound to $\delta_{2k}^\mathcal{I}<0.4931$ and also obtained another sufficient
condition on the block RIP with $\delta_k^\mathcal{I}<0.307$ \cite{LL}. Recently, Chen and Li \cite{CL2} have shown  a sharp sufficient condition
based on the high order block RIC $\delta_{tk}^\mathcal{I}<\sqrt{\frac{t-1}{t}}$ for any $t\geq\frac{4}{3}$.

The  $\ell_2/\ell_{1}$ minimization method \eqref{pro3} is itself nonadaptive since it dose not use any
 prior  information about the block sparse signal $x$.
  However, the estimate of the support of the signal
$x$  or of its largest coefficients may be possible to be drawn in
many applications (see \cite{FMSY}).
 Incorporating prior block support information of signals,
 we introduce a method by replacing the $\ell_2/\ell_{1}$ minimization \eqref{pro3}
with the following  weighted $\ell_2/\ell_{1}$ minimization with $L\ (1\leq L\leq M)$  weights $\omega_1,\omega_2,\ldots,\omega_L\in[0,1]$
\begin{align}\label{pro4}
\min_{x\in \mathbb{R}^{N}}&\|x_\mathrm{w}\|_{2,1} \ \ \
 {\rm subject\quad to}\ \ \  \|y-Ax\|_2\leq\varepsilon
 \end{align}
where $\mathrm{w} \in [0,1]^M$ and $\|x_\mathrm{w}\|_{2,1}=\sum\limits_{i=1}^{M}\mathrm{w}_{i}\|x[i]\|_2$ is the weighted $\ell_2/\ell_1$ norm of $x$.
The main idea of the weighted $\ell_2/\ell_{1}$ minimization approach \eqref{pro4}
is to choose  appropriately $\mathrm{w}$ such that in the weighted objective function,  the blocks of $x$ which are expected to be large are penalized  less.
 Throughout the article, given $L$ disjoint block support estimates of the block signal $x$ over $\mathcal{I}$ by
$\widetilde{T_j}\subseteq [M]$, where $j=1,2,\ldots,L$ and $\cup_{j=1}^L\widetilde{T}_j=\widetilde{T}$,  we set
\begin{align}\label{pro7}
\mathrm{w}_{i}=\left\{\begin{array}{cc}
                     1, & i\in \widetilde{T}^{c} \\
                     \omega_j, &  i\in \widetilde{T}_j
                   \end{array}
                   \right.
\end{align}
for all $i\in[M]$. Note that when $\mathcal{I}=\{d_1=1,d_2=1,\ldots,d_M=1\}$, the weighted $\ell_2/\ell_{1}$ minimization  meets with  the weighted $\ell_{1}$
 minimization  \cite{BMP,CL,FMSY,J,KXAH,LV,NSW, VL2, VL1}.

 In this paper, we establish the high order block RIP condition  to ensure the stable
and robust recovery of block signals $x$ through the weighted $\ell_2/\ell_{1}$ minimization  \eqref{pro4} and derive an error bound between
the unknown original block sparse signal $x$ and the minimizer of \eqref{pro4}. And we also show that
 when all of  the accuracy of $L$ disjoint prior block support
estimates  are at least  $50\%$, the recovery by the
weighted $\ell_2/\ell_{1}$ minimization method \eqref{pro4} is stable and robust
under weaker sufficient conditions compared to the  $\ell_2/\ell_{1}$
minimization \eqref{pro3}. Moreover, we analyze how many random measurements of some random measurements matrices $A$  are sufficient to satisfy
 the block RIP condition with high probability. Last, we present an algorithm used to solve  the weighted $\ell_2/\ell_{1}$
 minimization \eqref{pro4} with $0<\omega_i\leq 1$ ($i=1,2,\ldots, L$) and  illustrate the advantages of weighted $\ell_2/\ell_1$ minimization
approach  in the recovery performance of block sparse signals under
uniform and non-uniform prior information by extensive numerical
experiments.

The rest of this paper is organized as follows. In Section \ref{2},
we will introduce some notations and some basic results that will be
used. The main results and  their proofs
 are presented in Section \ref{3}.  Section \ref{4} discusses the measurement number of  some  random matrices
 satisfying the block RIP condition with high probability.
 In Section \ref{5},
we demonstrate the benefits of the weighted $\ell_2/\ell_1$-minimization allowing uniform and non-uniform weights in the reconstruction of
 block sparse signals by numerical experiments. A conclusion is included in Section \ref{6}.

\section{Preliminaries}\label{2}
Let us begin with some notations. Define a mixed $\ell_2/\ell_p$
norm with $p=1,\ 2,\ \infty$
 as $\|x\|_{2,p}=(\sum_{i=1}^M\|x[i]\|_2^p)^{\frac{1}{p}}$. Note that $\|x\|_{2,2}=\|x\|_2$.
 Let $\Gamma\subset [M]$ be a block index set and $\Gamma^c\subset [M]$ be its complement set. For arbitrary block signal
$x\in\mathbb{R}^{N}$ over $\mathcal{I}$,  let $x_{k}$ over $\mathcal{I}$ be its best block $k$-sparse
approximation such that $x_{k}$ is block
$k$-sparse  supported on $T \subseteq [M]$ with $|T|\leq k$ and minimizes $\|x-s\|_{2,1}$
over all block $k$-sparse vectors $s$ over  $\mathcal{I}$.
Then $T$ is the block support of $x_{k}$,
i.e., $T={\rm b}$-$\rm{supp}$$(x_{k})$.
Let $x[\Gamma] \in \R^N $ over $\mathcal{I}$ be a vector which equals to $x$ on block indices $\Gamma$ and $0$ otherwise.
$x[\Gamma][i]$ denotes $i$th block of $x[\Gamma]$.
$x[\max(k)]$ over $\mathcal{I}$ is defined as $x$ with all but the
largest $k$  blocks in $\ell_2$ norm set to zero, and
$x[-\max(k)]=x-x[\max(k)]$.  For any  $i\in\{1,2,\cdots,L\}$, let $\widetilde{T}_i\subseteq
[M]$ be the support estimate of $x$ with
$|\widetilde{T}_i|=\rho_i k$ ($\rho_i \geq 0$) and $\widetilde{T}=\cup_{i=1}^L\widetilde{T}_i$,
where $\widetilde{T}_i\cap \widetilde{T}_j=\emptyset\ (i\neq j)$, $|\widetilde{T}|=\rho k$ and $\rho=\sum\limits_{i=1}^L\rho_i \geq 0$ represents the ratio
of the size of all estimated block support to the size of the actual block support $T$.  $\delta_k$ denotes the  $k$ order standard  restricted isometry constant \cite{CT}.

The following lemma is a key technical tool for analysing the sharp restricted isometry conditions of block sparse signal recovery.
It is an extension of Lemma 1.1 \cite{CZ} in the block case, which represents block signals in a block polytope by convex combination of block sparse
signals.
 \begin{lem}(\cite{CL2}, Lemma 2.2)\label{l1}
For a positive number $\beta$ and a positive integer $s$,
define the block polytope $T(\beta, s)\subset\mathbb{R}^{N}$ by
$$T(\beta, s)=\{v\in\mathbb{R}^{N}: \|v\|_{2,\infty}\leq \beta, \|v\|_{2,1}\leq s\beta\}.$$
For any $v\in \mathbb{R}^{N}$,
define the set of block  sparse vectors $U(\beta, s, v)\subset\mathbb{R}^{N}$ by
\begin{eqnarray*}
U(\beta, s, v)=\{u\in\mathbb{R}^{N}: \mathrm{b}\text{-}\mathrm{supp}(u)\subseteq \mathrm{b}\text{-}\mathrm{supp}(v), \|u\|_{2,0}\leq s, \|u\|_{2,1}=\|v\|_{2,1}, \|u\|_{2,\infty}\leq\beta\}.
\end{eqnarray*}
Then any $v\in T(\beta, s)$ can be expressed
as
$$v=\sum\limits_{i=1}^{J}\lambda_{i}u_{i},$$
where $u_{i}\in U(\beta, s, v)$ and $0\leq \lambda_{i}\leq 1,\  \sum\limits_{i=1}^{J}\lambda_{i}=1.$
\end{lem}
 Cai and Zhang established  a useful elementary inequality in Lemma 5.3 \cite{CZ1}. Applying  the inequality,
we can perform finer estimation on mixed $\ell_2/\ell_1,\ \ell_2/\ell_2$ norms  in  the  proof of Theorem 3.1.
 \begin{lem}(\cite{CZ1}, Lemma 5.3)\label{l2}
Assume $m\geq l$, $a_{1}\geq a_{2}\geq\cdots\geq a_{m}\geq 0$,
$\sum\limits_{i=1}^{l}a_{i}\geq \sum\limits_{i=l+1}^{m}a_{i},$
then for all $\theta\geq 1$,
$$\sum\limits_{j=l+1}^{m}a_{j}^{\theta}\leq \sum\limits_{i=1}^{l}a_{i}^{\theta}.$$
More generally, assume $a_{1}\geq a_{2}\geq\cdots\geq a_{m}\geq 0$, $\lambda\geq 0$
and $\sum\limits_{i=1}^{l}a_{i}+\lambda\geq \sum\limits_{i=l+1}^{m}a_{i},$ then for all $\theta\geq 1$,
$$\sum\limits_{j=l+1}^{m}a_{j}^{\theta}\leq l\Big(\sqrt[\theta]{\frac{\sum_{i=1}^{l}a_{i}^{\theta}}{l}}+\frac{\lambda}{l}\Big)^{\theta}.$$
\end{lem}

As we mentioned in the introduction, Chen and Li have obtained  a high order sufficient condition based on the block RIP to ensure the recovery of block sparse signals in \cite{CL2}.
The main result on the sufficient condition is stated as below.
\begin{thm}(\cite{CL2}, Theorem 3.1)\label{t1}
Let $x\in\R^N$ be an arbitrary vector consistent with \eqref{mod1}  and $\|z\|_{2}\leq \varepsilon$.
 If the measurement matrix $A$ satisfies the block RIP with
\begin{eqnarray}\label{g1}
 \delta_{tk}^\mathcal{I}<\sqrt{\frac{t-1}{t}}
\end{eqnarray}
for  $t>1$, the solution $\hat{x}$ to \eqref{pro3} obeys
\begin{eqnarray}\label{g2}
  \|\hat{x}-x\|_{2}\leq 2C_{0}\varepsilon
  +C_{1}\frac{2\|x[T^c]\|_{2,1}}{\sqrt{k}},
\end{eqnarray}
where
\begin{eqnarray}\label{g21}
  &&C_{0}=\frac{\sqrt{2t(t-1)(1+\delta_{tk}^\mathcal{I})}}{t(\sqrt{(t-1)/t}-\delta_{tk}^\mathcal{I})}, \notag\\
   &&C_{1}=\frac{\sqrt{2}\delta_{tk}^\mathcal{I}+\sqrt{t(\sqrt{(t-1)/t}-\delta_{tk}^\mathcal{I})
\delta_{tk}^\mathcal{I}}}{t(\sqrt{(t-1)/t}-\delta_{tk}^\mathcal{I})}+1.
\end{eqnarray}
\end{thm}
Note that if $t\geq 4/3$,   the condition \eqref{g1} is sharp in Theorem 3.2 of  \cite{CL2}.

It is clear that the weighted $\ell_2/\ell_1$ minimization problem \eqref{pro4} is equivalent to the weighted $\ell_1$ minimization problem when
$\mathcal{I}=\{d_1=1,d_2=1,\cdots,d_M=1\}$, i.e, $M=N$. In the case,
Theorem \ref{theor1} below states the main result of \cite{NSW} for the weighted $\ell_1$ minimization with $L$ ($1\leq L\leq N$) weights.
\begin{thm}(\cite{NSW},Theorem 2)\label{theor1}
Let $x\in\R^N$,  $x_k$ denote its best $k$-sparse approximation, and denote the support of
$x_k$ by $T\subseteq\{1,2,\ldots,N\}$. Let $\widetilde{T}_i\subseteq\{1,2,\cdots,N\}$ for $i=1,\ldots,L$, where $1\leq L\leq N$, be arbitrary disjoint sets
and denote $\widetilde{T}=\cup_{i=1}^L\widetilde{T}_i$. Without loss of generality, assume that the weights in \eqref{pro7} are ordered so that
$1\geq\omega_1\geq\omega_2\geq\cdots\geq\omega_L\geq0$. For each $i$, define the relative size $\rho_i$ and $\alpha_i$ via
$|\widetilde{T}_i|=\rho_i k$ and $\frac{|\widetilde{T}_i\cap T|}{|\widetilde{T}_i|}=\alpha_i$. Suppose that there
exists $a>1,\ a\in\frac{1}{k}\Z$ with $\sum_{i=1}^L\rho_i(1-\alpha_i)\leq a$, and that the measurement matrix $A$ has the standard RIP with
\begin{align}\label{ghm1}
\delta_{ak}+\frac{a}{K_L^2}\delta_{(a+1)k}<\frac{a}{K_L^2}-1,
\end{align}
where $K_L=\omega_L+(1-\omega_1)\sqrt{1+\sum\limits_{i=1}^L(\rho_i-2\alpha_i\rho_i)}+\sum\limits_{j=2}^L\bigg((\omega_{j-1}-\omega_j)\sqrt{1+\sum\limits_{i=j}^L(\rho_i-2\alpha_i\rho_i)}\bigg)$.
Then  the minimizer $\hat{x}$ to \eqref{pro4} obeys
\begin{align*}
\|\hat{x}-x\|_2\leq 2\varepsilon C_0^{'}+2C_1^{'}k^{-\frac{1}{2}}\bigg(\|x-x_k\|_1\sum_{i=1}^L\omega_i+(1-\sum_{i=1}^L\omega_i)\|x_{\widetilde{T}^c\cap T^c}\|_1-\sum_{i=1}^L\sum_{j=1,j\neq i}^L\omega_j\|x_{\widetilde{T}_i\cap T^c}\|_1\bigg)
\end{align*}
where
the constants
\begin{align}\label{g6}
C_0^{'}=\frac{1+\frac{ K_L}{\sqrt{a}}}{\sqrt{1-\delta_{(a+1)k}}-\frac{K_L}{\sqrt{a}}\sqrt{1+\delta_{ak}}},
\ \ \ \ C_1^{'}=\frac{a^{-1/2}(\sqrt{1-\delta_{(a+1)k}}
+\sqrt{1+\delta_{ak}})}{\sqrt{1-\delta_{(a+1)k}}-\frac{K_L}{\sqrt{a}}\sqrt{1+\delta_{ak}}}.
\end{align}
\end{thm}
\begin{rmk}\label{rmk3}
Since $\delta_{ak}\leq\delta_{(a+1)k}$, the sufficient condition for \eqref{ghm1} to hold is
\begin{align}\label{g7}
\delta_{(a+1)k}<\frac{a-K_L^2}{a+K_L^2}\triangleq\delta({a+1},{K_L}).
\end{align}
\end{rmk}

From now on, let $h=\hat{x}-x$, where $\hat{x}$ over $\mathcal{I}$
is the minimizer of the weighted $\ell_2/\ell_1$ minimization
problem \eqref{pro4} and $x$ be the original block signal over
$\mathcal{I}$. For any block index set, one establishes a cone
constraint to prove our results (in Section \ref{3}) as following.
\begin{lem}(Block cone constraint)\label{lem2}
For any block index set $\Gamma\subseteq [M]$, it holds that
\begin{align}\label{cgl9}
  &\|h[\Gamma^{c}]\|_{2,1}\leq \omega_L\|h[\Gamma]\|_{2,1}+(1-\omega_1)\|h[\Gamma\cup\cup_{i=1}^L\widetilde{T}_i \backslash \cup_{i=1}^L(\widetilde{T}_i\cap\Gamma)]\|_{2,1}\nonumber\\
 &\ \ \ \ +\sum\limits_{i=2}^L(\omega_{i-1}-\omega_{i})\|h[\Gamma \cup \cup_{j=i}^L\widetilde{T}_j\backslash \cup_{j=i}^L(\widetilde{T}_j\cap \Gamma)]\|_{2,1}
+2\left( \sum_{i=1}^L\omega_i\|x[\Gamma^c]\|_{2,1}\right.\nonumber\\
&\ \ \ \  \left.+(1-\sum_{i=1}^L\omega_i)\|x[\widetilde{T}^c\cap\Gamma^c]\|_{2,1}-
\sum_{i=1}^L(\sum_{j=1}^L\omega_j-\omega_i)\|x[\widetilde{T}_i\cap\Gamma^c]\|_{2,1}\right).
\end{align}
\end{lem}
\begin{rmk}
When $\omega_i=1$ for all $i\in\{1,2,\ldots,L\}$,  the result of  Lemma \ref{lem2} is identical to that of Lemma 2.3 in \cite{CL2}.
Suppose  $d_i=1$ for all $i\in[M]$,  the result of  Lemma \ref{lem2} meets with that of Lemma 1 in \cite{NSW}. Under the above assumption, if
$0\leq \omega_1=\omega_2=\cdots=\omega_L=\omega<1$, the inequality \eqref{cgl9} is
\begin{align*}
 \|h[\Gamma^{c}]\|_{2,1}\leq& \omega\|h[\Gamma]\|_{2,1}+(1-\omega)\|h[\Gamma\cup\widetilde{T}\backslash\widetilde{T}\cap\Gamma]\|_{2,1}\nonumber\\
 &+2( \omega\|x[\Gamma^c]\|_{2,1}+(1-\omega)\|x[\widetilde{T}^c\cap\Gamma^c]\|_{2,1}),
\end{align*}
 which is (21) of \cite{FMSY}.
\end{rmk}
\begin{proof}
Using the fact that $\hat{x}=x+h$ is                                                                                                                                                                                                  a minimizer of the weighted  $\ell_2/\ell_1$ minimization  problem \eqref{pro4}, we have
\begin{eqnarray*}
\|\hat{x}_\mathrm{w}\|_{2,1}=\|(x+h)_\mathrm{w}\|_{2,1}\leq\|x_\mathrm{w}\|_{2,1}.
\end{eqnarray*}
We then obtain that
\begin{eqnarray*}
\sum_{i=1}^L\omega_i\|x[\widetilde{T}_i]+h[\widetilde{T}_i]\|_{2,1}+\|x[\widetilde{T}^c]+h[\widetilde{T}^c]\|_{2,1}\leq
\sum_{i=1}^L\omega_i\|x[\widetilde{T}_i]\|_{2,1}+\|x[\widetilde{T}^c]\|_{2,1},
\end{eqnarray*}
since $\widetilde{T}_i\cap \widetilde{T}_j=\emptyset\ (i\neq j)$.
Therefore,
\begin{eqnarray*}
&&\sum_{i=1}^L\omega_i\|x[\widetilde{T}_i\cap\Gamma]+h[\widetilde{T}_i\cap\Gamma]\|_{2,1}+\sum_{i=1}^L\omega_i\|x[\widetilde{T}_i\cap\Gamma^c]
+h[\widetilde{T}_i\cap\Gamma^c]\|_{2,1}\\
&&\ +\|x[\widetilde{T}^c\cap\Gamma]+h[\widetilde{T}^c\cap\Gamma]\|_{2,1}+\|x[\widetilde{T}^c\cap\Gamma^c]+h[\widetilde{T}^c\cap\Gamma^c]\|_{2,1}\\
&&\leq\sum_{i=1}^L\omega_i\|x[\widetilde{T}_i\cap\Gamma]\|_{2,1}+\sum_{i=1}^L\omega_i\|x[\widetilde{T}_i\cap\Gamma^c]\|_{2,1}
+\|x[\widetilde{T}^c\cap\Gamma]\|_{2,1}+\|x[\widetilde{T}^c\cap\Gamma^c]\|_{2,1}.
\end{eqnarray*}
From the triangle inequality, it follows that
\begin{eqnarray*}
&&\sum_{i=1}^L\omega_i\|x[\widetilde{T}_i\cap\Gamma]\|_{2,1}-\sum_{i=1}^L\omega_i\|h[\widetilde{T}_i\cap\Gamma]\|_{2,1}
+\sum_{i=1}^L\omega_i\|h[\widetilde{T}_i\cap\Gamma^c]\|_{2,1}-\sum_{i=1}^L\omega_i\|x[\widetilde{T}_i\cap\Gamma^c]\|_{2,1}\\
&&\ +\|x[\widetilde{T}^c\cap\Gamma]\|_{2,1}-\|h[\widetilde{T}^c\cap\Gamma]\|_{2,1}
+\|h[\widetilde{T}^c\cap\Gamma^c]\|_{2,1}-\|x[\widetilde{T}^c\cap\Gamma^c]\|_{2,1}\\
&&\leq\sum_{i=1}^L\omega_i\|x[\widetilde{T}_i\cap\Gamma]\|_{2,1}+\sum_{i=1}^L\omega_i\|x[\widetilde{T}_i\cap\Gamma^c]\|_{2,1}
+\|x[\widetilde{T}^c\cap\Gamma]\|_{2,1}+\|x[\widetilde{T}^c\cap\Gamma^c]\|_{2,1},
\end{eqnarray*}
i.e.,
\begin{eqnarray*}
\sum_{i=1}^L\omega_i\|h[\widetilde{T}_i\cap\Gamma^c]\|_{2,1}
+\|h[\widetilde{T}^c\cap\Gamma^c]\|_{2,1}&\leq&\sum_{i=1}^L\omega_i\|h[\widetilde{T}_i\cap\Gamma]\|_{2,1}
+\|h[\widetilde{T}^c\cap\Gamma]\|_{2,1}\\
&\ & +2\bigg(\|x[\widetilde{T}^c\cap\Gamma^c]\|_{2,1}+\sum_{i=1}^L\omega_i\|x[\widetilde{T}_i\cap\Gamma^c]\|_{2,1}\bigg).
\end{eqnarray*}

Adding and subtracting $\sum\limits_{i=1}^L\omega_i\|h[\widetilde{T}_i^c\cap\Gamma^c]\|_{2,1}$ on the left hand side, and $\sum\limits_{i=1}^L\omega_i\|h[\widetilde{T}_i^c\cap\Gamma]\|_{2,1}$, $2\sum\limits_{i=1}^L\omega_i\|x[\widetilde{T}_i^c\cap\Gamma^c]\|_{2,1}$ on the right hand side respectively, we obtain
\begin{eqnarray}\label{cgl7}
&&\sum_{i=1}^L\omega_i\|h[\Gamma^c]\|_{2,1}
+\|h[\widetilde{T}^c\cap\Gamma^c]\|_{2,1}-\sum_{i=1}^L\omega_i\|h[\tilde{T}_i^c\cap\Gamma^c]\|_{2,1}\nonumber\\
&&\leq\sum_{i=1}^L\omega_i\|h[\Gamma]\|_{2,1}+
\|h[\widetilde{T}^c\cap\Gamma]\|_{2,1}-\sum_{i=1}^L\omega_i\|h[\widetilde{T}_i^c\cap\Gamma]\|_{2,1}\nonumber\\
&\ &
+2\left(\sum_{i=1}^L\omega_i\|x[\Gamma^c]\|_{2,1}
+x[\widetilde{T}^c\cap\Gamma^c]\|_{2,1}-\sum_{i=1}^L\omega_i\|x[\widetilde{T}_i^c\cap\Gamma^c]\|_{2,1}\right).
\end{eqnarray}
Note that $\|h[\widetilde{T}^c\cap\Gamma^c]\|_{2,1}$ and $\sum\limits_{i=1}^L\omega_i\|h[\tilde{T}_i^c\cap\Gamma^c]\|_{2,1}$ are written as
\begin{eqnarray*}
&&\|h[\widetilde{T}^c\cap\Gamma^c]\|_{2,1}
=(1-\sum\limits_{i=1}^L\omega_i)\|h[\widetilde{T}^c\cap\Gamma^c]\|_{2,1}+\sum\limits_{i=1}^L\omega_i\|h[\widetilde{T}^c\cap\Gamma^c]\|_{2,1}\\
&&=(1-\sum\limits_{i=1}^L\omega_i)\|h[\widetilde{T}^c\cap\Gamma^c]\|_{2,1}+\sum\limits_{i=1}^L\omega_i\bigg(\|h[\Gamma^c]\|_{2,1}
-\|h[\widetilde{T}\cap\Gamma^c]\|_{2,1}\bigg)\\
&&=(1-\sum\limits_{i=1}^L\omega_i)\|h[\widetilde{T}^c\cap\Gamma^c]\|_{2,1}
+\sum\limits_{i=1}^L\omega_i\bigg(\|h[\Gamma^c]\|_{2,1}-\sum_{j=1}^L\|h[\widetilde{T}_j\cap\Gamma^c]\|_{2,1}\bigg)
\end{eqnarray*}
and
\begin{eqnarray*}
\sum_{i=1}^L\omega_i\|h[\widetilde{T}_i^c\cap\Gamma^c]\|_{2,1}=\sum_{i=1}^L\omega_i\bigg(\|h[\Gamma^c]\|_{2,1}-\|h[\widetilde{T}_i\cap\Gamma^c]\|_{2,1}\bigg).
\end{eqnarray*}
Then it is clear that
\begin{align*}
&\|h[\widetilde{T}^c\cap\Gamma^c]\|_{2,1}-\sum_{i=1}^L\omega_i\|h[\widetilde{T}_i^c\cap\Gamma^c]\|_{2,1}\\
&=(1-\sum\limits_{i=1}^L\omega_i)\|h[\widetilde{T}^c\cap\Gamma^c]\|_{2,1}-\sum_{i=1}^L(\sum_{j=1}^L\omega_j-\omega_i)\|h[\widetilde{T}_i\cap\Gamma^c]\|_{2,1}.
\end{align*}
Similarly, there are
\begin{align*}
&\|h[\widetilde{T}^c\cap\Gamma]\|_{2,1}-\sum_{i=1}^L\omega_i\|h[\widetilde{T}_i^c\cap\Gamma]\|_{2,1}\\
&=(1-\sum\limits_{i=1}^L\omega_i)\|h[\widetilde{T}^c\cap\Gamma]\|_{2,1}
-\sum_{i=1}^L(\sum_{j=1}^L\omega_j-\omega_i)\|h[\widetilde{T}_i\cap\Gamma]\|_{2,1}
\end{align*}
and
\begin{align*}
&\|x[\widetilde{T}^c\cap\Gamma^c]\|_{2,1}-\sum_{i=1}^L\omega_i\|x[\widetilde{T}_i^c\cap\Gamma^c]\|_{2,1}\\
&=(1-\sum\limits_{i=1}^L\omega_i)\|x[\widetilde{T}^c\cap\Gamma^c]\|_{2,1}
-\sum_{i=1}^L(\sum_{j=1}^L\omega_j-\omega_i)\|x[\widetilde{T}_i\cap\Gamma^c]\|_{2,1}.
\end{align*}
 Combining \eqref{cgl7} with the above equalities, one easily deduces  that
\begin{eqnarray*}
&&\sum_{i=1}^L\omega_i\|h[\Gamma^c]\|_{2,1}
\leq\sum_{i=1}^L\omega_i\|h[\Gamma]\|_{2,1}+(1-\sum_{i=1}^L\omega_i)\bigg(\|h[\widetilde{T}^c\cap\Gamma]\|_{2,1}
-\|h[\widetilde{T}^c\cap\Gamma^c]\|_{2,1}\bigg)\\
&&-\sum_{i=1}^L(\sum_{j=1}^L\omega_j-\omega_i)\bigg(\|h[\widetilde{T}_i\cap\Gamma]\|_{2,1}
-\|h[\widetilde{T}_i\cap\Gamma^c]\|_{2,1}\bigg)+2\bigg(\sum_{i=1}^L\omega_i\|x[\Gamma^c]\|_{2,1}\bigg.\\
&&\bigg.+(1-\sum_{i=1}^L\omega_i)\|x[\widetilde{T}^c\cap\Gamma^c]\|_{2,1}-
\sum_{i=1}^L(\sum_{j=1}^L\omega_j-\omega_i)\|x[\widetilde{T}_i\cap\Gamma^c]\|_{2,1}\bigg).
\end{eqnarray*}

In the remainder of the proof, denote
\begin{eqnarray*}
Z=\sum_{i=1}^L\omega_i\|x[\Gamma^c]\|_{2,1}+(1-\sum_{i=1}^L\omega_i)\|x[\widetilde{T}^c\cap\Gamma^c]\|_{2,1}-
\sum_{i=1}^L(\sum_{j=1}^L\omega_j-\omega_i)\|x[\widetilde{T}_i\cap\Gamma^c]\|_{2,1}.
\end{eqnarray*} Then the above inequality can be expressed as
\begin{eqnarray*}
\sum_{i=1}^L\omega_i\|h[\Gamma^c]\|_{2,1}
&\leq&\sum_{i=1}^L\omega_i\|h[\Gamma]\|_{2,1}+(1-\sum_{i=1}^L\omega_i)\bigg(\|h[\widetilde{T}^c\cap\Gamma]\|_{2,1}
-\|h[\widetilde{T}^c\cap\Gamma^c]\|_{2,1}\bigg)\nonumber\\
&\ &-\sum_{i=1}^L(\sum_{j=1}^L\omega_j-\omega_i)\bigg(\|h[\widetilde{T}_i\cap\Gamma]\|_{2,1}
-\|h[\widetilde{T}_i\cap\Gamma^c]\|_{2,1}\bigg)
+2Z.
\end{eqnarray*}
Since  $\|h[\Gamma^c]\|_{2,1}=\sum\limits_{i=1}^L\omega_i\|h[\Gamma^c]\|_{2,1}+(1-\sum\limits_{i=1}^L\omega_i)\bigg(\|h[\widetilde{T}^c\cup\Gamma^c]\|_{2,1}
+\|h[\widetilde{T}\cup \Gamma^c]\|_{2,1}\bigg)$ and $(\widetilde{T}^c\cap\Gamma)\cup(\widetilde{T}\cap\Gamma^c)=(\widetilde{T}\cup\Gamma)\backslash(\widetilde{T}\cap\Gamma)
=\Gamma\cup\cup_{i=1}^L\widetilde{T}_i\backslash\cup_{i=1}^L(\widetilde{T}_i\cap\Gamma)$, we deduce that
\begin{eqnarray}\label{cgl8}
\|h[\Gamma^c]\|_{2,1}
&\leq&\sum_{i=1}^L\omega_i\|h[\Gamma]\|_{2,1}+(1-\sum_{i=1}^L\omega_i)\bigg(\|h[\widetilde{T}^c\cap\Gamma]\|_{2,1}+\|h[\widetilde{T}\cup \Gamma^c]\|_{2,1}\bigg)
\nonumber\\
&\ &
-\sum_{i=1}^L(\sum_{j=1}^L\omega_j-\omega_i)\bigg(\|h[\widetilde{T}_i\cap\Gamma]\|_{2,1}-\|h[\widetilde{T}_i\cap\Gamma^c]\|_{2,1}\bigg)
+2Z\nonumber\\
&=&\omega_L\|h[\Gamma]\|_{2,1}+(1-\omega_1)\|h[\Gamma\cup\cup_{i=1}^L\widetilde{T}_i\backslash\cup_{i=1}^L(\widetilde{T}_i\cap\Gamma)]\|_{2,1}
+\sum_{i=1}^{L-1}\omega_i\|h[\Gamma]\|_{2,1}\nonumber\\
&\ &+(\omega_1-\sum_{i=1}^L\omega_i)\|h[\widetilde{T}^c\cap\Gamma]\|_{2,1}-\sum_{i=1}^L(\sum_{j=1}^L\omega_j-\omega_i)\|h[\widetilde{T}_i\cap\Gamma]\|_{2,1}
\nonumber\\
&\ &+(\omega_1-\sum_{i=1}^L\omega_i)\|h[\widetilde{T}\cup \Gamma^c]\|_{2,1}+\sum_{i=1}^L(\sum_{j=1}^L\omega_j-\omega_i)\|h[\widetilde{T}_i\cap\Gamma^c]\|_{2,1}+2Z\nonumber\\
&=&\omega_L\|h[\Gamma]\|_{2,1}+(1-\omega_1)\|h[\Gamma\cup\cup_{i=1}^L\widetilde{T}_i\backslash\cup_{i=1}^L(\widetilde{T}_i\cap\Gamma)]\|_{2,1}
+\sum_{i=1}^{L-1}\omega_i\|h[\Gamma]\|_{2,1}\nonumber\\
&\ &+(\omega_1-\sum_{i=1}^L\omega_i)\bigg(\|h[\Gamma]\|_{2,1}-\sum_{i=1}^L\|h[\widetilde{T}_i\cap\Gamma]\|_{2,1}\bigg)-\sum_{i=1}^L(\sum_{j=1}^L\omega_j-\omega_i)\|h[\widetilde{T}_i\cap\Gamma]\|_{2,1}
\nonumber\\
&\ &+\sum_{i=1}^L(\omega_1-\sum_{j=1}^L\omega_j)\|h[\widetilde{T}_i\cup \Gamma^c]\|_{2,1}+\sum_{i=1}^L(\sum_{j=1}^L\omega_j-\omega_i)\|h[\widetilde{T}_i\cap\Gamma^c]\|_{2,1}+2Z\nonumber\\
&=&\omega_L\|h[\Gamma]\|_{2,1}+(1-\omega_1)\|h[\Gamma\cup\cup_{i=1}^L\widetilde{T}_i\backslash\cup_{i=1}^L(\widetilde{T}_i\cap\Gamma)]\|_{2,1}
\nonumber\\
&\ &+(\omega_1-\omega_L)\|h[\Gamma]\|_{2,1}+\sum_{i=2}^L(\omega_1-\omega_i)\bigg(\|h[\widetilde{T}_i^c\cap\Gamma]\|_{2,1}-\|h[\Gamma]\|_{2,1}\bigg)\nonumber\\
&\ &+\sum_{i=2}^L(\omega_1-\omega_i)\|h[\widetilde{T}_i\cup \Gamma^c]\|_{2,1}+2Z\nonumber\\
&=&\omega_L\|h[\Gamma]\|_{2,1}+(1-\omega_1)\|h[\Gamma\cup\cup_{i=1}^L\widetilde{T}_i\backslash\cup_{i=1}^L(\widetilde{T}_i\cap\Gamma)]\|_{2,1}
\nonumber\\
&\ &+\sum_{i=2}^{L-1}(\omega_i-\omega_1)\|h[\Gamma]\|_{2,1}
+\sum_{i=2}^L(\omega_1-\omega_i)\bigg(\|h[\widetilde{T}_i^c\cap\Gamma]\|_{2,1}\bigg.\nonumber\\
&\ &\bigg.+\|h[\widetilde{T}_i\cup \Gamma^c]\|_{2,1}\bigg)+2Z.
\end{eqnarray}

Using the facts that
\begin{eqnarray*}
\|h[\widetilde{T}_j^c\cap\Gamma]\|_{2,1}&=&\|h[\Gamma]\|_{2,1}-\|h[\widetilde{T}_{j}\cap \Gamma]\|_{2,1}\\
&=&\sum_{i=1,i\neq j}^{L}\|h[\widetilde{T}_j\cap \Gamma]\|_{2,1}+\|h[\Gamma\cap\cap_{i=1}^L\widetilde{T}_{i}^c]\|_{2,1}
\end{eqnarray*}
and
\begin{eqnarray*}
\sum_{j=i}^L\|h[\widetilde{T}_{j}\cap \Gamma^c]\|_{2,1}+\|h[\Gamma \cap \cap_{j=i}^L\widetilde{T}_j^c]\|_{2,1} =\|h[\Gamma \cup \cup_{j=i}^L\widetilde{T}_j\backslash \cup_{j=i}^L(\widetilde{T}_j\cap \Gamma)]\|_{2,1},
\end{eqnarray*}
we obtain
\begin{eqnarray*}
&&\sum\limits_{i=2}^L(\omega_1-\omega_i)\bigg(\|h[\widetilde{T}_i^c\cap\Gamma]\|_{2,1}+\|h[\widetilde{T}_i\cup \Gamma^c]\|_{2,1}\bigg)\\
&&=\sum\limits_{i=2}^L(\omega_{i-1}-\omega_{i})\sum\limits_{j=i}^L\bigg(\|h[\widetilde{T}_j^c\cap\Gamma]\|_{2,1}+\|h[\widetilde{T}_j\cup \Gamma^c]\|_{2,1}\bigg)\\
&&=\sum\limits_{i=2}^L(\omega_{i-1}-\omega_{i})\bigg(\sum_{j=i}^L(\|h[\Gamma]\|_{2,1}-\|h[\widetilde{T}_{j}\cap \Gamma]\|_{2,1})-\|h[\Gamma \cap \cap_{j=i}^L\widetilde{T}_j^c]\|_{2,1}\bigg.\\
&\ \ &\bigg.+\|h[\Gamma \cap \cap_{j=i}^L\widetilde{T}_j^c]\|_{2,1}+\sum_{j=i}^L\|h[\widetilde{T}_{j}\cap \Gamma^c]\|_{2,1}\bigg)\\
&&=\sum\limits_{i=2}^L(\omega_{i-1}-\omega_{i})(L-i)\|h[\Gamma]\|_{2,1}+\sum\limits_{i=2}^L(\omega_{i-1}-\omega_{i})\|h[\Gamma \cup \cup_{j=i}^L\widetilde{T}_j\backslash \cup_{j=i}^L(\widetilde{T}_j\cap \Gamma)]\|_{2,1}\\
&&=-\sum_{i=2}^{L-1}(\omega_i-\omega_1)\|h[\Gamma]\|_{2,1}+\sum\limits_{i=2}^L(\omega_{i-1}-\omega_{i})\|h[\Gamma \cup \cup_{j=i}^L\widetilde{T}_j\backslash \cup_{j=i}^L(\widetilde{T}_j\cap \Gamma)]\|_{2,1}.
\end{eqnarray*}
Substituting the above equality into \eqref{cgl8}, we get the result \eqref{cgl9}.
\end{proof}

\section{ Main results }\label{3}
In this section, we present the main results. First, we consider the signal
  recovery  model \eqref{mod1} in the setting where the error vector $z\neq 0$ and the block signal $x$ is not exactly
 block $k$-sparse and establish the sufficient condition based on the high order block RIP.
The result implies  that the condition  guarantees
 the exact recovery in the noiseless setting and stable recovery in noisy setting when the block signal $x$ is  block $k$-sparse.
\begin{thm}\label{t3}
Consider the  signal recovery model \eqref{mod1} with $\|z\|_{2}\leq\varepsilon$, where $x\in\mathbb{R}^{N}$ over $\mathcal{I}$ is an arbitrary block signal.
Suppose that $x_{k}$ over $\mathcal{I}$ is the  best block $k$-sparse approximation and $\hat{x}$ is the minimizer of \eqref{pro4}.
Let $\widetilde{T}_i\subseteq [M]$ with $i=1,2,\ldots, L$ be disjoint block index  sets and denote $\widetilde{T}=\cup_{i=1}^L\widetilde{T}_i$
 where $L$ is a positive integer, such that $|\widetilde{T}_i|=\rho_i k$ and
  $|\widetilde{T}_i\cap T|=\alpha_i \rho_i k$ where $T$ is the support of $x_k$,\, $\rho_i \geq 0$ and $0\leq \alpha_i \leq 1$.
 Without loss of generality, assume that the weights in \eqref{pro7} are
ordered so that $0\leq\omega_L\leq\omega_{L-1}\leq\cdots\leq\omega_1\leq 1$.
If  $A$ satisfies the block RIP with
\begin{align}\label{g23}
  \delta_{tk}^\mathcal{I}<\sqrt{\frac{t-d}{t-d+\Upsilon_L^{2}}}\triangleq \delta^\mathcal{I}({t},{\Upsilon_L})
\end{align}
 for $t>d$, where $$\Upsilon_L=\omega_L+(1-\omega_1)\sqrt{1+\sum\limits_{i=1}^L\rho_i-2\sum\limits_{i=1}^L\alpha_i\rho_i}
 +\sum_{i=2}^L(\omega_{i-1}-\omega_i)\sqrt{1+\sum\limits_{j=i}^L\rho_j-2\sum\limits_{j=i}^L\alpha_j\rho_j}$$
 and
 \begin{align*}
   d=\left\{
 \begin{array}{cc}
   1, & \prod\limits_{i=1}^L\omega_i=1 \\
  \max\limits_{i\in\{1,2,\ldots,L\}}\{b_i(1-\sum\limits_{j=i}^L\alpha_j\rho_j+a_i)\}, & 0\leq\prod\limits_{i=1}^L\omega_i<1
 \end{array}
 \right.
 \end{align*}
 with $a_i=\max{\{\sum\limits_{j=i}^L\alpha_j\rho_j, \sum\limits_{j=i}^L(1-\alpha_j)\rho_j\}}$
 and
 \begin{align*}
 b_i=\left\{
       \begin{array}{ll}
         1, & \hbox{$i=1$} \\
        sgn (\omega_{i-1}-\omega_i), & \hbox{$i=2,\ldots,L.$}
       \end{array}
     \right.
 \end{align*}
Then
\begin{align}\label{g9}
\|\hat{x}-x\|_{2}&\leq 2D_{0}\varepsilon+\frac{2 D_{1}}{\sqrt{k}}\Big( \sum\limits_{i=1}^L\omega_i\|x-x_k\|_{2,1}
+(1-\sum\limits_{i=1}^L\omega_i)\|x[\widetilde{T}^c\cap T^c]\|_{2,1}\Big.\nonumber\\
&\ \ \Big.-\sum\limits_{i=1}^L(\sum\limits_{j=1}^L\omega_j-\omega_i)\|x[\widetilde{T}_i\cap T^c]\|_{2,1}\Big),
\end{align}
where
\begin{align}\label{g24}
D_{0}=& \frac{\sqrt{2(t-d)(t-d+\Upsilon_L^{2})(1+\delta_{tk}^\mathcal{I})}}{(t-d+\Upsilon_L^{2})(\sqrt{\frac{t-d}{t-d+\Upsilon_L^{2}}}-\delta_{tk}^\mathcal{I})},\notag\\
D_{1}
=&\frac{\sqrt{2}\delta_{tk}^\mathcal{I}\Upsilon_L+\sqrt{(t-d+\Upsilon_L^{2})(\sqrt{\frac{t-d}{t-d+\Upsilon_L^{2}}}
-\delta_{tk}^\mathcal{I})\delta_{tk}^\mathcal{I}}}
{(t-d+\Upsilon_L^{2})(\sqrt{\frac{t-d}{t-d+\Upsilon_L^{2}}}-\delta_{tk}^\mathcal{I})}
+\frac{1}{\sqrt{d}}.
\end{align}
\end{thm}
\begin{proof}

We  will prove  the associated recovery guarantees \eqref{g9} of the weighted $\ell_2/\ell_1$ minimization. To this end,
assume that $tk$ is an integer and $\hat{x}=x+h$,
where $x$ is the original block signal over $\mathcal{I}$ and $\hat{x}$ over $\mathcal{I}$ is a solution of the
  weighted $\ell_2/\ell_1$ minimization  problem \eqref{pro4}. From Lemma \ref{lem2}
  and the block support $T\subseteq[M]$, it follows that
  \begin{align}\label{g11}
&\|h[T^{c}]\|_{2,1}\leq \omega_L\|h[T]\|_{2,1}+(1-\omega_1)\|h[T\cup\cup_{i=1}^L\widetilde{T}_i\backslash\cup_{i=1}^L(\widetilde{T}_i\cap T)]\|_{2,1}\nonumber\\
 &\ \ \ \ +\sum\limits_{i=2}^L(\omega_{i-1}-\omega_{i})\|h[T \cup \cup_{j=i}^L\widetilde{T}_j\backslash \cup_{j=i}^L(\widetilde{T}_j\cap T)]\|_{2,1}
+2\bigg( \sum_{i=1}^L\omega_i\|x[T^c]\|_{2,1}\bigg.\nonumber\\
&\ \ \ \  \bigg.+(1-\sum_{i=1}^L\omega_i)\|x[\widetilde{T}^c\cap T^c]\|_{2,1}-
\sum_{i=1}^L(\sum_{j=1}^L\omega_j-\omega_i)\|x[\widetilde{T}_i\cap T^c]\|_{2,1}\bigg).
\end{align}
Based on  \begin{align*}
   d=\left\{
 \begin{array}{cc}
   1, & \prod\limits_{i=1}^L\omega_i=1 \\
  \max\limits_{i\in\{1,2,\ldots,L\}}\{b_i(1-\sum\limits_{j=i}^L\alpha_j\rho_j+a_i)\}, & 0\leq\prod\limits_{i=1}^L\omega_i<1
 \end{array}
 \right.
 \end{align*}
 with  $a_i=\max{\{\sum\limits_{j=i}^L\alpha_j\rho_j, \sum\limits_{j=i}^L(1-\alpha_j)\rho_j\}}$
 and
 \begin{align*}
 b_i=\left\{
       \begin{array}{ll}
         1, & \hbox{$i=1$} \\
        sgn (\omega_{i-1}-\omega_i), & \hbox{$i=2,\ldots,L$},
       \end{array}
     \right.
 \end{align*}
it is clear that $d$ is an  integer and   $d\geq 1$. Recall  $h[\max(dk)]$ as the block $dk$-sparse vector $h$ over $\mathcal{I}$ with all but the largest $dk$ blocks in $\ell_2$ norm set to zero.
From \eqref{g11} and $d\geq 1$, we have
\begin{align}\label{g12}
&\|h[-\max(dk)]\|_{2,1}\leq \|h[T^c]\|_{2,1}\leq \omega_L\|h[T]\|_{2,1}+(1-\omega_1)\|h[T\cup\cup_{i=1}^L\widetilde{T}_i\backslash\cup_{i=1}^L(\widetilde{T}_i\cap T)]\|_{2,1}\nonumber\\
 &\ \ \ \ +\sum\limits_{i=2}^L(\omega_{i-1}-\omega_{i})\|h[T \cup \cup_{j=i}^L\widetilde{T}_j\backslash \cup_{j=i}^L(\widetilde{T}_j\cap T)]\|_{2,1}
+2\bigg( \sum_{i=1}^L\omega_i\|x[T^c]\|_{2,1}\bigg.\nonumber\\
&\ \ \ \  \bigg.+(1-\sum_{i=1}^L\omega_i)\|x[\widetilde{T}^c\cap T^c]\|_{2,1}-
\sum_{i=1}^L(\sum_{j=1}^L\omega_j-\omega_i)\|x[\widetilde{T}_i\cap T^c]\|_{2,1}\bigg).
\end{align}
Let
\begin{align*}
r&=\frac{1}{k}\left[ \omega_L\|h[T]\|_{2,1}+(1-\omega_1)\|h[T\cup\cup_{i=1}^L\widetilde{T}_i\backslash\cup_{i=1}^L(\widetilde{T}_i\cap T)]\|_{2,1}\right.\nonumber\\
 &\ \ \ \ +\sum\limits_{i=2}^L(\omega_{i-1}-\omega_{i})\|h[T \cup \cup_{j=i}^L\widetilde{T}_j\backslash \cup_{j=i}^L(\widetilde{T}_j\cap T)]\|_{2,1}
+2\bigg( \sum_{i=1}^L\omega_i\|x[T^c]\|_{2,1}\bigg.\nonumber\\
&\ \ \ \  \bigg.\bigg.+(1-\sum_{i=1}^L\omega_i)\|x[\widetilde{T}^c\cap T^c]\|_{2,1}-
\sum_{i=1}^L(\sum_{j=1}^L\omega_j-\omega_i)\|x[\widetilde{T}_i\cap T^c]\|_{2,1}\bigg)\bigg],
\end{align*}
then $r\geq 0$. If $r=0$, then $\|h[T^c]\|_{2,1}=0$ and $h$ is a block $k$-sparse vector. From the definition of the block RIP and $t>d\geq 1$, it follows that $(1-\delta_{tk}^\mathcal{I})\|h\|_2^2\leq\|Ah\|_2^2=\|A\hat{x}-Ax\|_{2}^2
\leq(\|y-A\hat{x}\|_{2}+\|Ax-y\|_{2})^2\leq 4\varepsilon^2$, that is, $\|h\|_2\leq \frac{2\varepsilon}{\sqrt{1-\delta_{tk}^\mathcal{I}}}$. By means of a series of calculation, we obtain
$D_0\geq\frac{1}{\sqrt{1-\delta_{tk}^\mathcal{I}}}$. Therefore, we get the associated recovery guarantees \eqref{g9} of the weighted $\ell_2/\ell_1$ minimization under $r=0$.

From now on, we only consider $r>0$. Decompose $h[-\max(dk)]$ over $\mathcal{I}$ into two parts,
 $h[-\max(dk)]=h^{(1)}+h^{(2)}$, where for all $i\in [M]$, the $i$th block of the block vectors $h^{(1)}$ and $h^{(2)}$   satisfies respectively
 \begin{eqnarray*}
 h^{(1)}[i]=\left\{
             \begin{array}{ll}
               h[-\max(dk)][i], & \hbox{$\|h[-\max(dk)][i]\|_2>\frac{r}{t-d}$} \\
               \mathbf{0}\in\R^{d_i}, & \hbox{else}
             \end{array}
           \right.
 \end{eqnarray*}
and
\begin{eqnarray*}
 h^{(2)}[i]=\left\{
             \begin{array}{ll}
               h[-\max(dk)][i], & \hbox{$\|h[-\max(dk)][i]\|_2\leq\frac{r}{t-d}$} \\
               \mathbf{0}\in\R^{d_i}, & \hbox{else.}
             \end{array}
           \right.
 \end{eqnarray*}

In view of the  definition of the block vector $h^{(1)}$ and  \eqref{g12}, we obtain
 $$\|h^{(1)}\|_{2,1}\leq \|h[-\max(dk)]\|_{2,1}\leq kr.$$
Let $\|h^{(1)}\|_{2,0}=m.$
Because the $\ell_2$ norm of every non-zero blocks of  $h^{(1)}$ is larger than $\frac{r}{t-d}$ ($t>d,\ r>0$), we have
\begin{align*}
kr\geq \|h^{(1)}\|_{2,1}
=\sum\limits_{i\in \mathrm{b}\text{-}\mathrm{supp}(h^{(1)})}\|h^{(1)}[i]\|_{2}
\geq \sum\limits_{i\in \mathrm{b}\text{-}\mathrm{supp}(h^{(1)})}\frac{r}{t-d}
=\frac{mr}{t-d}.
\end{align*}
Namely $m\leq k(t-d)$. In addition, we have
\begin{align}\label{cgl11}
\|h[\max(dk)]+h^{(1)}\|_{2,0}=dk+m\leq dk+k(t-d)=tk,
\end{align}
\begin{align*}
\|h^{(2)}\|_{2,1}=\|h[-\max(dk)]\|_{2,1}-\|h^{(1)}\|_{2,1}
\leq kr-\frac{mr}{t-d}
=(k(t-d)-m)\cdot\frac{r}{t-d}
\end{align*}
and
\begin{eqnarray*}
\|h^{(2)}\|_{2,\infty}\leq \frac{r}{t-d},
\end{eqnarray*}
where the last inequality follows from all non-zero blocks of $h^{(2)}$ having $\ell_2$ norm smaller  than $\frac{r}{t-d}$.

Now, using Lemma \ref{l1} with $s=k(t-d)-m$ and $\beta=\frac{r}{t-d}$, then $h^{(2)}$ can be expressed as a convex combination of block-sparse vectors, i.e.,
$h^{(2)}=\sum\limits_{i=1}^{J}\lambda_{i}u_{i},$
 where $\sum\limits_{i=1}^{J}\lambda_{i}=1$ and
  $u_i\in U(\frac{r}{t-d},k(t-d)-m, h^{(2)})$. In the remainder of the proof, one considers the following
  two case.

  Case 1: $\Upsilon_L\neq 0$

 In the case, denote $X=\|h[\max(dk)]+h^{(1)}\|_{2}$ and
\begin{align*}
P=\frac{2\bigg(\sum\limits_{i=1}^L\omega_i\|x[T^c]\|_{2,1}+(1-\sum\limits_{i=1}^L\omega_i)\|x[\widetilde{T}^c\cap T^c]\|_{2,1}-
\sum\limits_{i=1}^L(\sum\limits_{j=1}^L\omega_j-\omega_i)\|x[\widetilde{T}_i\cap T^c]\|_{2,1}\bigg)}{\sqrt{k}\Upsilon_L}.
\end{align*}
Thus, we have the upper bound
\begin{align}\label{cgl10}
\|u_{i}\|_{2}&=\|u_{i}\|_{2,2}\leq \sqrt{\|u_{i}\|_{2,0}}\|u_{i}\|_{2,\infty}\leq \sqrt{k(t-d)-m}\|u_{i}\|_{2,\infty}\leq \sqrt{k(t-d)}\cdot\frac{r}{t-d} \nonumber\\
 &\leq \sqrt{\frac{k}{t-d}}r
 =\sqrt{\frac{k}{t-d}}\cdot\frac{1}{k}\Big[\omega_L\|h[T]\|_{2,1}
 +(1-\omega_1)\|h[T\cup\cup_{i=1}^L\widetilde{T}_i\backslash\cup_{i=1}^L(\widetilde{T}_i\cap T)]\|_{2,1}\Big.\nonumber\\
 &\ \ \ \ +\sum\limits_{i=2}^L(\omega_{i-1}-\omega_{i})\|h[T \cup \cup_{j=i}^L\widetilde{T}_j\backslash \cup_{j=i}^L(\widetilde{T}_j\cap T)]\|_{2,1}
+2\Big( \sum_{i=1}^L\omega_i\|x[T^c]\|_{2,1}\Big.\nonumber\\
&\ \ \ \  \Big.\Big.+(1-\sum_{i=1}^L\omega_i)\|x[\widetilde{T}^c\cap T^c]\|_{2,1}-
\sum_{i=1}^L(\sum_{j=1}^L\omega_j-\omega_i)\|x[\widetilde{T}_i\cap T^c]\|_{2,1}\Big)\Big] \nonumber\\
=&\frac{1}{\sqrt{k(t-d)}}\Big[\omega_L\|h[T]\|_{2,1}+(1-\omega_1)\|h[T\cup\cup_{i=1}^L\widetilde{T}_i\backslash\cup_{i=1}^L(\widetilde{T}_i\cap T)]\|_{2,1}
\Big.\nonumber\\
 &\ \ \ \ \Big.+\sum\limits_{i=2}^L(\omega_{i-1}-\omega_{i})\|h[T \cup \cup_{j=i}^L\widetilde{T}_j\backslash \cup_{j=i}^L(\widetilde{T}_j\cap T)]\|_{2,1}\Big]
+\frac{\Upsilon_L P}
{\sqrt{t-d}}\nonumber\\
\leq&\frac{1}{\sqrt{t-d}}\Big[\omega_L\|h[T]\|_{2}+(1-\omega_1)\sqrt{1+\sum_{i=1}^L\rho_i-2\sum_{i=1}^L\alpha_i\rho_i}
\|h[T\cup\cup_{i=1}^L\widetilde{T}_i\backslash\cup_{i=1}^L(\widetilde{T}_i\cap T)]\|_{2}
\Big.\nonumber\\
 &\ \ \ \ \Big.+\sum\limits_{i=2}^L(\omega_{i-1}-\omega_{i})\sqrt{1+\sum_{j=i}^L\rho_j-2\sum_{j=i}^L\alpha_j\rho_j}\|h[T \cup \cup_{j=i}^L\widetilde{T}_j\backslash \cup_{j=i}^L(\widetilde{T}_j\cap T)]\|_{2}\Big]\nonumber\\
 &\ \ \ \
+\frac{\Upsilon_L P}
{\sqrt{t-d}}\nonumber\\
\leq&\frac{\|h[\max(dk)]\|_{2}}{\sqrt{t-d}}\Big[\omega_L+(1-\omega_1)\sqrt{1+\sum_{i=1}^L\rho_i-2\sum_{i=1}^L\alpha_i\rho_i}\Big.\nonumber\\
 &\ \ \ \
\Big.+\sum\limits_{i=2}^L(\omega_{i-1}-\omega_{i})\sqrt{1+\sum_{j=i}^L\rho_i-2\sum_{j=i}^L\alpha_i\rho_i}\Big]
+\frac{\Upsilon_L P}
{\sqrt{t-d}}\nonumber\\
\leq&\frac{\Upsilon_L \|h[\max(dk)]+h^{(1)}\|_{2}}{\sqrt{t-d}}+\frac{\Upsilon_L P}
{\sqrt{t-d}}\nonumber \\
=&\frac{\Upsilon_L}{\sqrt{t-d}}(X+P),
\end{align}
 where we use the facts $|T\cup\cup_{i=1}^L\widetilde{T}_i\backslash\cup_{i=1}^L(\widetilde{T}_i\cap T)|=(1+\sum\limits_{i=1}^L\rho_i-2\sum\limits_{i=1}^L\alpha_i\rho_i)k\leq dk$
  and  for $i=2,\ldots,L$, $|T \cup \cup_{j=i}^L\widetilde{T}_j\backslash \cup_{j=i}^L(\widetilde{T}_j\cap T)|=(1+\sum\limits_{j=i}^L\rho_i-2\sum\limits_{j=i}^L\alpha_i\rho_i)k\leq dk$  when $sgn(\omega_{i-1}-\omega_i)=1$.

Let $\beta_{i}=h[\max(dk)]+h^{(1)}+\mu u_{i}$ where $0\leq \mu
\leq 1$, then we get
\begin{align}\label{g16}
  \sum\limits_{j=1}^{J}\lambda_{j}\beta_{j}-\frac{1}{2}\beta_{i}
  =&h[\max(dk)]+h^{(1)}+\mu h^{(2)}-\frac{1}{2}\beta_{i} \notag\\
  =&(\frac{1}{2}-\mu)(h[\max(dk)]+h^{(1)})-\frac{1}{2}\mu u_{i}+\mu h,
\end{align}
where $\sum\limits_{i=1}^J\lambda_i=1$ and $h^{(1)}+h^{(2)}=h-h[\max(dk)]$.
Since $h[\max(dk)]$, $h^{(1)}$, $u_{i}$ are block $dk$-, $m$-, $((t-d)k-m)$-sparse vectors respectively, $\beta_{i}$ and $\sum\limits_{j=1}^{N}\lambda_{j}\beta_{j}-\frac{1}{2}\beta_{i}-\mu h=(\frac{1}{2}-\mu)(h[\max(dk)]+h^{(1)})-\frac{1}{2}\mu u_{i}$ are block $tk$-sparse vectors.

Next, we compute an upper bound of $X=\|h[\max(dk)]+h^{(1)}\|_{2}$.
We shall use the facts that
\begin{align}\label{g14}
  \|Ah\|_{2}=\|A\hat{x}-Ax\|_{2}
\leq\|y-A\hat{x}\|_{2}+\|Ax-y\|_{2}\leq 2\varepsilon
\end{align}
and the following identity (see (25) in \cite{CZ})
\begin{align}\label{g18}
 \sum\limits_{i=1}^{J}\lambda_{i}\Big\|A\Big(\sum\limits_{j=1}^{J}\lambda_{j}\beta_{j}-\frac{1}{2}\beta_{i}\Big)\Big\|_{2}^{2}
 =\sum\limits_{i=1}^{J}\frac{\lambda_{i}}{4}\|A\beta_{i}\|_{2}^{2}.
\end{align}
Besides,
\begin{align}\label{g15}
 \langle A(h[\max(dk)]+h^{(1)}), Ah\rangle
  &\leq \|A(h[\max(dk)]+h^{(1)})\|_{2}\|Ah\|_{2}\nonumber\\
  &\leq\sqrt{1+\delta_{tk}^{\mathcal{I}}}\|h[\max(dk)]+h^{(1)}\|_{2}\cdot(2\varepsilon),
\end{align}
where the last inequality uses  the definition of block RIP with $\delta_{tk}^\mathcal{I}$, \eqref{cgl11} and \eqref{g14}.
Combining \eqref{g15} and \eqref{g16}, we estimate the left hand side of \eqref{g18}
\begin{align*}
&\sum\limits_{i=1}^{J}\lambda_{i}\Big\|A\Big(\sum\limits_{j=1}^{J}\lambda_{j}\beta_{j}-\frac{1}{2}\beta_{i}\Big)\Big\|_{2}^{2}\\
&=\sum\limits_{i=1}^{J}\lambda_{i}\Big\|A\Big((\frac{1}{2}-\mu)(h[\max(dk)]+h^{(1)})-\frac{1}{2}\mu u_{i}+\mu h\Big)\Big\|_{2}^{2} \\
&=\sum\limits_{i=1}^{J}\lambda_{i}\Big \|A
 \Big((\frac{1}{2}-\mu)(h[\max(dk)]+h^{(1)})-\frac{1}{2}\mu u_{i}\Big)\Big\|_{2}^{2}+\mu \|Ah\|_2^2\\
 &\ \ +2\sum\limits_{i=1}^{J}\lambda_{i}\langle A(\frac{1}{2}-\mu)(h[\max(dk)]+h^{(1)})-\frac{1}{2}\mu u_{i},\mu h\rangle\\
&=\sum\limits_{i=1}^{J}\lambda_{i}\Big \|A
 \Big((\frac{1}{2}-\mu)(h[\max(dk)]+h^{(1)})-\frac{1}{2}\mu u_{i}\Big)\Big\|_{2}^{2}
+\mu(1-\mu)\langle A(h[\max(dk)]+h^{(1)}), Ah\rangle  \\
&\leq(1+\delta_{tk}^{\mathcal{I}})\sum\limits_{i=1}^{J}\lambda_{i}\Big\|(\frac{1}{2}-\mu)(h[\max(dk)]+h^{(1)})-\frac{1}{2}\mu u_{i}\Big\|_{2}^{2}  \\
&\ \ +2\varepsilon\mu(1-\mu)\sqrt{1+\delta_{tk}^{\mathcal{I}}}\|h[\max(dk)]+h^{(1)}\|_{2} \\
&=(1+\delta_{tk}^{\mathcal{I}})\Big[(\frac{1}{2}-\mu)^{2}\|h[\max(dk)]+h^{(1)}\|_{2}^{2}
   +\frac{\mu^{2}}{4}
   \sum\limits_{i=1}^{J}\lambda_{i}\|u_{i}\|_{2}^{2}\Big]\\
&\ \ +2\varepsilon\mu(1-\mu)\sqrt{1+\delta_{tk}^{\mathcal{I}}}\|h[\max(dk)]+h^{(1)}\|_{2}\\
&=(1+\delta_{tk}^{\mathcal{I}})(\frac{1}{2}-\mu)^{2}X^{2}
   +\frac{\mu^{2}(1+\delta_{tk}^{\mathcal{I}})}{4}
   \sum\limits_{i=1}^{J}\lambda_{i}\|u_{i}\|_{2}^{2}+2\varepsilon\mu(1-\mu)\sqrt{1+\delta_{tk}^{\mathcal{I}}}X,
\end{align*}
where the last but one equality applies $\sum\limits_{i=1}^J\lambda_i=1$ and
\begin{align*}
\langle\lambda_iu_i,h[\max(dk)]+h^{(1)}\rangle=0.
\end{align*}

For the right hand side of \eqref{g18}, from the expression of $\beta_{i}$ and  the definition of the block RIP with $\delta_{tk}^{\mathcal{I}}$ we have
\begin{align*}
\sum\limits_{i=1}^{J}\frac{\lambda_{i}}{4}\|A\beta_{i}\|_{2}^{2}
&=\sum\limits_{i=1}^{J}\frac{\lambda_{i}}{4}\|A(h[\max(dk)]+h^{(1)}+\mu u_{i})\|_{2}^{2} \\
&\geq\sum\limits_{i=1}^{J}\frac{\lambda_{i}}{4}(1-\delta_{tk}^{\mathcal{I}})\|h[\max(dk)]+h^{(1)}+\mu u_{i}\|_{2}^{2}\\
&= (1-\delta_{tk}^{\mathcal{I}})\sum\limits_{i=1}^{J}\frac{\lambda_{i}}{4}
\Big(\|h[\max(dk)]+h^{(1)}\|_{2}^{2}+\mu^{2}\|u_{i}\|_{2}^{2}\Big)\\
&=\frac{1-\delta_{tk}^{\mathcal{I}}}{4}X^2+\frac{\mu^2(1-\delta_{tk}^{\mathcal{I}})}{4}\sum_{i=1}^{J}\lambda_i \|u_{i}\|_{2}^{2}.
\end{align*}
In consideration of the above two inequalities and \eqref{g18} we
have
\begin{align}\label{gg1}
 0=&\sum\limits_{i=1}^{J}\lambda_{i}\Big\|A\Big(\sum\limits_{j=1}^{J}\lambda_{j}\beta_{j}-\frac{1}{2}\beta_{i}\Big)\Big\|_{2}^{2}
-\sum\limits_{i=1}^{J}\frac{\lambda_{i}}{4}\|A\beta_{i}\|_{2}^{2} \nonumber \\
\leq&\bigg((1+\delta_{tk}^{\mathcal{I}})(\frac{1}{2}-\mu)^{2}-\frac{1}{4}(1-\delta_{tk}^{\mathcal{I}})\bigg)
X^{2}+\frac{1}{2}\delta_{tk}^{\mathcal{I}}\mu^{2}\sum\limits_{i=1}^{J}\lambda_{i}\|u_{i}\|_{2}^{2}\nonumber\\
&+2\varepsilon\mu(1-\mu)\sqrt{1+\delta_{tk}^{\mathcal{I}}}X\nonumber\\
\leq&\bigg[(1+\delta_{tk}^{\mathcal{I}})(\frac{1}{2}-\mu)^{2}-\frac{1}{4}(1-\delta_{tk}^{\mathcal{I}})
+\frac{\delta_{tk}^{\mathcal{I}}\mu^{2}\Upsilon_L^{2}}{2(t-d)}\bigg]
X^{2} \nonumber \\
&+\bigg[\mu(1-\mu)\sqrt{1+\delta_{tk}^{\mathcal{I}}}\cdot(2\varepsilon)+\frac{\delta_{tk}^{\mathcal{I}}\mu^{2}\Upsilon_L^{2}P}{t-d}\bigg]X
+\frac{\delta_{tk}^{\mathcal{I}}\mu^{2}\Upsilon_L^{2}P^{2}}{2(t-d)}\nonumber\\
=&\Big[(\mu^{2}-\mu)+\Big(\frac{1}{2}-\mu+(1+\frac{\Upsilon_L^{2}}{2(t-d)})\mu^{2}\Big)\delta_{tk}^{\mathcal{I}}\Big]X^{2} \nonumber\\
  &+\Big[2\varepsilon\mu(1-\mu)\sqrt{1+\delta_{tk}^{\mathcal{I}}}+\frac{\delta_{tk}^{\mathcal{I}}\mu^{2}\Upsilon_L^{2}P}{t-d}\Big]X
+\frac{\delta_{tk}^{\mathcal{I}}\mu^{2}\Upsilon_L^{2}P^{2}}{2(t-d)},
\end{align}
where we apply the estimate of $\|u_i\|_2$ in \eqref{cgl10}.
Substituting $\mu=\frac{\sqrt{(t-d)(t-d+\Upsilon_L^{2})}-(t-d)}{\Upsilon_L^{2}}\in (0,1)$ into \eqref{gg1} yields
\begin{align*}
&-\frac{t-d+\Upsilon_L^{2}}{t-d}\mu^{2}\bigg(\sqrt{\frac{t-d}{t-d+\Upsilon_L^{2}}}-\delta_{tk}^{\mathcal{I}}\bigg)X^{2}
+\bigg(2\varepsilon\mu^{2}\frac{t-d+\Upsilon_L^{2}}{t-d}\sqrt{\frac{(1+\delta_{tk}^{\mathcal{I}})(t-d)}{t-d+\Upsilon_L^{2}}}\\
&\ +\frac{\delta_{tk}^{\mathcal{I}}\mu^{2}\Upsilon_L^{2}P}{t-d}\bigg)X
+\frac{\delta_{tk}^{\mathcal{I}}\mu^{2}\Upsilon_L^{2}P^{2}}{2(t-d)}\geq0,
\end{align*}
i.e.,
\begin{eqnarray*}
&&\frac{\mu^{2}}{t-d}\Big[-(t-d+\Upsilon_L^{2})\Big(\sqrt{\frac{t-d}{t-d+\Upsilon_L^{2}}}-\delta_{tk}^{\mathcal{I}}\Big)X^{2}\\
&&+\Big(2\varepsilon\sqrt{(t-d)(t-d+\Upsilon_L^{2})(1+\delta_{tk}^{\mathcal{I}})}
+\delta_{tk}^{\mathcal{I}}\Upsilon_L^{2}P\Big)X
+\frac{\delta_{tk}^{\mathcal{I}}\Upsilon_L^{2}P^{2}}{2}\Big]\geq 0,
\end{eqnarray*}
which is a second-order inequality for $X$. Hence, under the conditions $\delta_{tk}^\mathcal{I}<\sqrt{\frac{t-d}{t-d+\Upsilon_L^2}}$ and $t>d$ we have
\begin{align*}
 &X\leq\bigg\{\Big( 2\varepsilon\sqrt{(t-d)(t-d+\Upsilon_L^{2})(1+\delta_{tk}^{\mathcal{I}})}+\delta_{tk}^{\mathcal{I}}\Upsilon_L^{2}P \Big)\\
 &+\Big[\Big( 2\varepsilon\sqrt{(t-d)(t-d+\Upsilon_L^{2})(1+\delta_{tk}^{\mathcal{I}})}+\delta_{tk}^{\mathcal{I}}\Upsilon_L^{2}P \Big)^{2}\\
 &+2(t-d+\Upsilon_L^{2})\Big( \sqrt{\frac{t-d}{t-d+\Upsilon_L^{2}}}-\delta_{tk}^{\mathcal{I}} \Big)\delta_{tk}^{\mathcal{I}}\Upsilon_L^{2}P^{2} \Big]^{1/2}\bigg\} \\
&\cdot \Big(2(t-d+\Upsilon_L^{2})(\sqrt{\frac{t-d}{t-d+\Upsilon_L^{2}}}-\delta_{tk}^{\mathcal{I}})\Big)^{-1}\\
&\leq\frac{2\varepsilon\sqrt{(t-d)(t-d+\Upsilon_L^{2})(1+\delta_{tk}^{\mathcal{I}})}}{(t-d+\Upsilon_L^{2})(\sqrt{\frac{t-d}{t-d+\Upsilon_L^{2}}}-\delta_{tk}^{\mathcal{I}})}
 \\
&+\frac{2\delta_{tk}^{\mathcal{I}}\Upsilon_L^{2}+\sqrt{2(t-d+\Upsilon_L^{2})(\sqrt{\frac{t-d}{t-d+\Upsilon_L^{2}}}-\delta_{tk}^{\mathcal{I}})
\delta_{tk}^{\mathcal{I}}\Upsilon_L^{2}}}
{2(t-d+\Upsilon_L^{2})(\sqrt{\frac{t-d}{t-d+\Upsilon_L^{2}}}-\delta_{tk}^{\mathcal{I}})}P,
\end{align*}
which is an upper bound of $X=\|h[\max(dk)]+h^{(1)}\|_{2}$.

Last, it remains to develop an upper bound on $\|h\|_2$.
To this end, we express $\|h\|_2^2=\|h[\max(dk)]\|_2^2+\|h[-\max(dk)]\|_2^2$.

 Considering the inequality \eqref{g12} and the definition of $P$, we have
\begin{align*}
\|h[-\max(dk)]\|_{2,1}&\leq \big(\omega_L+(1-\omega_1)+
\sum\limits_{i=2}^L(\omega_{i-1}-\omega_{i})\big)\|h[\max(dk)]\|_{2,1}
+P\sqrt{k}\Upsilon_L\\
&=\|h[\max(dk)]\|_{2,1}+P\sqrt{k}\Upsilon_L,
\end{align*}
where  we use that $|T|\leq dk \ (d\geq1)$, $|T\cup\cup_{i=1}^L\widetilde{T}_i\backslash\cup_{i=1}^L(\widetilde{T}_i\cap T)|=(1+\sum\limits_{i=1}^L\rho_i-2\sum\limits_{i=1}^L\alpha_i\rho_i)k\leq dk$ and for all $i\in\{2,\ldots,L\}$ as $sgn(\omega_{i-1}-\omega_i)=1$
\begin{align*}
|T\cup\cup_{j=i}^L\widetilde{T}_j\backslash\cup_{j=i}^L(\widetilde{T}_j\cap T)|=(1+\sum\limits_{j=i}^L\rho_j-2\sum\limits_{j=i}^L\alpha_j\rho_j)k\leq dk.
\end{align*}
Thanks to  Lemma \ref{l2} with $\theta=2$ , $l=dk$, and $\lambda=P\sqrt{k}\Upsilon_L$, we have
$$\|h[-\max(dk)]\|_{2} =\|h[-\max(dk)]\|_{2,2}\leq\|h[\max(dk)]\|_{2,2}+\frac{P\Upsilon_L}{\sqrt{d}}=\|h[\max(dk)]\|_{2}+\frac{P\Upsilon_L}{\sqrt{d}}.$$
Therefore,  we conclude that
\begin{align*}
\|h\|_{2}
\leq&\sqrt{\|h[\max(dk)]\|_{2}^{2}+\left(\|h[\max(dk)]\|_{2}+\frac{P\Upsilon_L}{\sqrt{d}}\right)^{2}}
\leq\sqrt{2}\|h[\max(dk)]\|_{2}+\frac{P\Upsilon_L}{\sqrt{d}}
\\
&\leq\sqrt{2}\|h[\max(dk)]+h^{(1)}\|_{2}+\frac{P\Upsilon_L}{\sqrt{d}}
=\sqrt{2}X+\frac{P\Upsilon_L}{\sqrt{d}}
\\
&\leq\frac{2\varepsilon\sqrt{2(t-d)(t-d+\Upsilon_L^{2})(1+\delta_{tk}^{\mathcal{I}})}}{(t-d+\Upsilon_L^{2})(\sqrt{\frac{t-d}{t-d+\Upsilon_L^{2}}}
-\delta_{tk}^{\mathcal{I}})} \\
&\ \ +\Bigg(\frac{\sqrt{2}\delta_{tk}\Upsilon_L^{2}+\sqrt{(t-d+\Upsilon_L^{2})(\sqrt{\frac{t-d}{t-d+\Upsilon_L^{2}}}
-\delta_{tk}^{\mathcal{I}})\delta_{tk}^{\mathcal{I}}\Upsilon_L^{2}}}
{(t-d+\Upsilon_L^{2})(\sqrt{\frac{t-d}{t-d+\Upsilon_L^{2}}}-\delta_{tk}^{\mathcal{I}})}
+\frac{\Upsilon_L}{\sqrt{d}}\Bigg)P \\
=&\frac{2\varepsilon\sqrt{2(t-d)(t-d+\Upsilon_L^{2})(1+\delta_{tk}^{\mathcal{I}})}}{(t-d+\Upsilon_L^{2})(\sqrt{\frac{t-d}{t-d+\Upsilon_L^{2}}}
-\delta_{tk}^{\mathcal{I}})} \\
&+\Bigg(\frac{\sqrt{2}\delta_{tk}^{\mathcal{I}}\Upsilon_L+\sqrt{(t-d+\Upsilon_L^{2})(\sqrt{\frac{t-d}{t-d+\Upsilon_L^{2}}}-\delta_{tk}^{\mathcal{I}})
\delta_{tk}^{\mathcal{I}}}}
{(t-d+\Upsilon_L^{2})(\sqrt{\frac{t-d}{t-d+\Upsilon_L^{2}}}-\delta_{tk}^{\mathcal{I}})}
+\frac{1}{\sqrt{d}}\Bigg)\\
&\ \frac{2\left( \sum\limits_{i=1}^L\omega_i\|x[T^c]\|_{2,1}
+(1-\sum\limits_{i=1}^L\omega_i)\|x[\widetilde{T}^c\cap T^c]\|_{2,1}-
\sum\limits_{i=1}^L(\sum\limits_{j=1}^L\omega_j-\omega_i)\|x[\widetilde{T}_i\cap T^c]\|_{2,1}\right)}{\sqrt{k}}.
\end{align*}

Case 2: $\Upsilon_L=0$

Similarly, let $X=\|h[\max(dk)]+h^{(1)}\|_{2}$ and
\begin{align*}
P'=\frac{2\left(\sum\limits_{i=1}^L\omega_i\|x[T^c]\|_{2,1}+(1-\sum\limits_{i=1}^L\omega_i)\|x[\widetilde{T}^c\cap T^c]\|_{2,1}-
\sum\limits_{i=1}^L(\sum\limits_{j=1}^L\omega_j-\omega_i)\|x[\widetilde{T}_i\cap T^c]\|_{2,1}\right)}{\sqrt{k}}.
\end{align*}
 In the same way, we  have that
$\|u_{i}\|_{2}\leq\frac{P'}{\sqrt{t-d}}$, $\|h[-\max(dk)]\|_{2,1}\leq\|h[\max(dk)]\|_{2,1}+P'\sqrt{k}$
and
\begin{align}\label{cgl13}
\|h[-\max(dk)]\|_{2} \leq\|h[\max(dk)]\|_{2}+\frac{P'}{\sqrt{d}}.
\end{align}
Then
 \begin{align*}
 0=&\sum\limits_{i=1}^{J}\lambda_{i}\Big\|A\Big(\sum\limits_{j=1}^{J}\lambda_{j}\beta_{j}-\frac{1}{2}\beta_{i}\Big)\Big\|_{2}^{2}
-\sum\limits_{i=1}^{J}\frac{\lambda_{i}}{4}\|A\beta_{i}\|_{2}^{2} \nonumber \\
\leq&\left((1+\delta_{tk}^{\mathcal{I}})(\frac{1}{2}-\mu)^{2}-\frac{1}{4}(1-\delta_{tk}^{\mathcal{I}})\right)
X^2+\frac{1}{2}\delta_{tk}^{\mathcal{I}}\mu^{2}\sum\limits_{i=1}^{J}\lambda_{i}\|u_{i}\|_{2}^{2}\nonumber\\
&+2\varepsilon\mu(1-\mu)\sqrt{1+\delta_{tk}^{\mathcal{I}}}X\nonumber\\
\leq&\left[(1+\delta_{tk}^{\mathcal{I}})(\frac{1}{2}-\mu)^{2}-\frac{1}{4}(1-\delta_{tk}^{\mathcal{I}})
\right]
X^{2}+2\varepsilon\mu(1-\mu)\sqrt{1+\delta_{tk}^{\mathcal{I}}}X
+\frac{\delta_{tk}^{\mathcal{I}}\mu^{2}(P')^{2}}{2(t-d)}.
\end{align*}
Taking $\mu=\frac{1}{2}$, we have
\begin{eqnarray*}
-(1-\delta_{tk}^{\mathcal{I}})X^{2}+2\varepsilon\sqrt{1+\delta_{tk}^{\mathcal{I}}} X
+\frac{\delta_{tk}^{\mathcal{I}}(P')^{2}}{2(t-d)}\geq 0.
\end{eqnarray*}
So, under the conditions $\delta_{tk}^\mathcal{I}<\sqrt{\frac{t-d}{t-d+\Upsilon_L^2}}$, i.e., $\delta_{tk}^\mathcal{I}<1$ and $t>d$ we obtain
\begin{align*}
 X\leq\frac{2\varepsilon\sqrt{1+\delta_{tk}^{\mathcal{I}}}}{1-\delta_{tk}^{\mathcal{I}}}
 +\sqrt{\frac{\delta_{tk}^{\mathcal{I}}}{2(t-d)
(1-\delta_{tk}^{\mathcal{I}})}}P'.
\end{align*}
From the above inequality and \eqref{cgl13}, it follows that
\begin{align*}
\|h\|_{2}
\leq&\sqrt{\|h[\max(dk)]\|_{2}^{2}+\left(\|h[\max(dk)]\|_{2}+\frac{P'}{\sqrt{d}}\right)^{2}}\\
\leq&\sqrt{2}X+\frac{P'}{\sqrt{d}}\\
\leq& \frac{2\varepsilon\sqrt{2(1+\delta_{tk}^{\mathcal{I}})}}{1-\delta_{tk}^{\mathcal{I}}}
 +\bigg(\sqrt{\frac{\delta_{tk}^{\mathcal{I}}}{2(t-d)
(1-\delta_{tk}^{\mathcal{I}})}}+\frac{1}{\sqrt{d}}\bigg)P'\\
=&\frac{2\varepsilon\sqrt{2(t-d)(t-d+\Upsilon_L^{2})(1+\delta_{tk}^{\mathcal{I}})}}{(t-d+\Upsilon_L^{2})(\sqrt{\frac{t-d}{t-d+\Upsilon_L^{2}}}
-\delta_{tk}^{\mathcal{I}})} \\
&+\Bigg(\frac{\sqrt{2}\delta_{tk}^{\mathcal{I}}\Upsilon_L+\sqrt{(t-d+\Upsilon_L^{2})(\sqrt{\frac{t-d}{t-d+\Upsilon_L^{2}}}-\delta_{tk}^{\mathcal{I}})
\delta_{tk}^{\mathcal{I}}}}
{(t-d+\Upsilon_L^{2})(\sqrt{\frac{t-d}{t-d+\Upsilon_L^{2}}}-\delta_{tk}^{\mathcal{I}})}
+\frac{1}{\sqrt{d}}\Bigg)\\
&\ \frac{2\left( \sum\limits_{i=1}^L\omega_i\|x[T^c]\|_{2,1}
+(1-\sum\limits_{i=1}^L\omega_i)\|x[\widetilde{T}^c\cap T^c]\|_{2,1}-
\sum\limits_{i=1}^L(\sum\limits_{j=1}^L\omega_j-\omega_i)\|x[\widetilde{T}_i\cap T^c]\|_{2,1}\right)}{\sqrt{k}}
\end{align*}
where in last equality $\Upsilon_L=0$.

When $tk$ is not an integer, take $t'=\lceil tk\rceil/k$, then $t'k$ is an integer, $t< t'$ and
$$\delta_{t'k}^{\mathcal{I}}=\delta_{tk}^{\mathcal{I}}<\sqrt{\frac{t-d}{t-d+\gamma^{2}}}<\sqrt{\frac{t'-d}{t'-d+\gamma^{2}}}$$
which implies that the case can be deduced to the former case ($tk$ is an integer).
To sum up, we complete the proof of Theorem \ref{t3}.
\end{proof}

\begin{rmk}\label{rmk1}
From Theorem \ref{t3}, it is clear that  the block signal $x$ can be   recovered exactly and stably
 from $y$ and $A$ in the noiseless and noisy cases as  $x$ is a block
$k$-sparse over $\mathcal{I}$.
\end{rmk}
Now, we present the sufficient condition and  associated constants in Theorem \ref{t3} for some special cases.
 As well as, we compare them with the sufficient condition and associated constants mentioned in previous works.
 The following results are easy to verify.
 \begin{pro}\label{pro6}
\begin{description}
 \item[{\rm (1)}]
 If $\omega_1=\omega_2=\cdots=\omega_L=\omega\in[0,1]$, then
\begin{align}\label{cgl12}
\Upsilon_L=\omega+(1-\omega)\sqrt{1+\sum\limits_{i=1}^L\rho_i-2\sum\limits_{i=1}^L\alpha_i\rho_i},\ \
 d=\left\{
             \begin{array}{ll}
               1, & \hbox{$\omega=1$} \\
               1-\sum\limits_{i=1}^L\alpha_i\rho_i+a_1, & \hbox{$0\leq\omega<1$}
             \end{array}
           \right.
\end{align}
and $\|\hat{x}-x\|_2\leq2 D_0\varepsilon
+\frac{2D_1}{\sqrt{k}}(\omega\|x[T^c]\|_{2,1}+(1-\omega)\|x[\widetilde{T}^c\cap
T^c]\|_{2,1})$, which can be regarded as an extension of  Theorem
3.1 \cite{CL} to block signals. In this case, denote $\Upsilon_L$ by
$\Upsilon_L^\omega$. For $d_1=d_2=\cdots=d_M=1$, the above result of
Theorem \ref{t3} is identical to that of Theorem 3.1 in \cite{CL}
with
 $\rho=\sum\limits_{i=1}^L\rho_i$ and $\alpha=\sum\limits_{i=1}^L\alpha_i\rho_i/\sum\limits_{i=1}^L\rho_i$.

\item[{\rm (2)}] If $\omega_i=1$ for all $i\in\{1,2,\ldots,L\}$,
then $\Upsilon_L=1$ and $d=1$. The result reduces to that of Theorem \ref{t1}. That is,
$D_0=C_0$, $D_1=C_1$ and
 the sufficient condition  for Theorem \ref{t3} given in \eqref{g23} is identical to \eqref{g1} in Theorem \ref{t1}.

\item[{\rm (3)}]If $\alpha_i=\frac{1}{2}$  for all $i\in\{1,2,\ldots,L\}$, then $\Upsilon_L=1$, $d=1$, $D_0=C_0$, $D_1=C_1$ and
 the sufficient condition  for Theorem \ref{t3} given in \eqref{g23} is identical to \eqref{g1} in Theorem \ref{t1}.

 \item[{\rm (4)}] Suppose that $0\leq\prod_{i=1}^L\omega_i<1$ and $\alpha_i>\frac{1}{2}$ for all $i=1,\ldots,L$, then $d=1$,
  $\Upsilon_L<1$, $D_0<C_0$, $D_1<C_1$ and the sufficient condition \eqref{g23} is weaker than
\eqref{g1} in Theorem \ref{t1}. For $d_1=\cdots=d_M=1$ then $D_0<C_0'$, $D_1<C_1'$ and
 the sufficient condition \eqref{g23} is weaker than the sufficient condition \eqref{g7} in Remark \ref{rmk3}.
\end{description}
 \end{pro}

Furthermore, we compare  the sufficient condition \eqref{g23} used the single weight with that used the combination of weights when all accuracies $\alpha_i$ are greater than $\frac{1}{2}$.

\begin{pro}\label{pro5}
Let $\omega_1\geq\omega_2\geq\cdots\geq \omega_L$, $\sum\limits_{i=1}^L \rho_i=\rho$ and $\alpha_1=\alpha_2=\cdots=\alpha_L=\alpha$. Then
$\delta^\mathcal{I}(t,\Upsilon_L^{\omega_1})\leq\delta^\mathcal{I}(t,\Upsilon_L)\leq\delta^\mathcal{I}(t,\Upsilon_L^{\omega_L})$ if and only if
$\alpha\geq\frac{1}{2}$.

\end{pro}
\begin{proof}
Needell  et.al have shown  $\Upsilon_L^{\omega_L}\leq \Upsilon_L\leq
\Upsilon_L^{\omega_1}$ if and only if $\alpha\geq\frac{1}{2}$ in the
proof of Proposition 1 in \cite{NSW}. From  the definition of
$\delta^\mathcal{I}(t,\Upsilon_L)$ in \eqref{g23}, it is clear that
$\delta^\mathcal{I}(t,\Upsilon_L^{\omega_1})\leq\delta^\mathcal{I}(t,\Upsilon_L)\leq\delta^\mathcal{I}(t,\Upsilon_L^{\omega_L})$
if and only if $\alpha\geq\frac{1}{2}$.
\end{proof}

\begin{figure*}[htbp!]
\centering
\begin{minipage}[htbp]{0.46\linewidth}
\centering
\includegraphics[width=9cm,height=7cm]{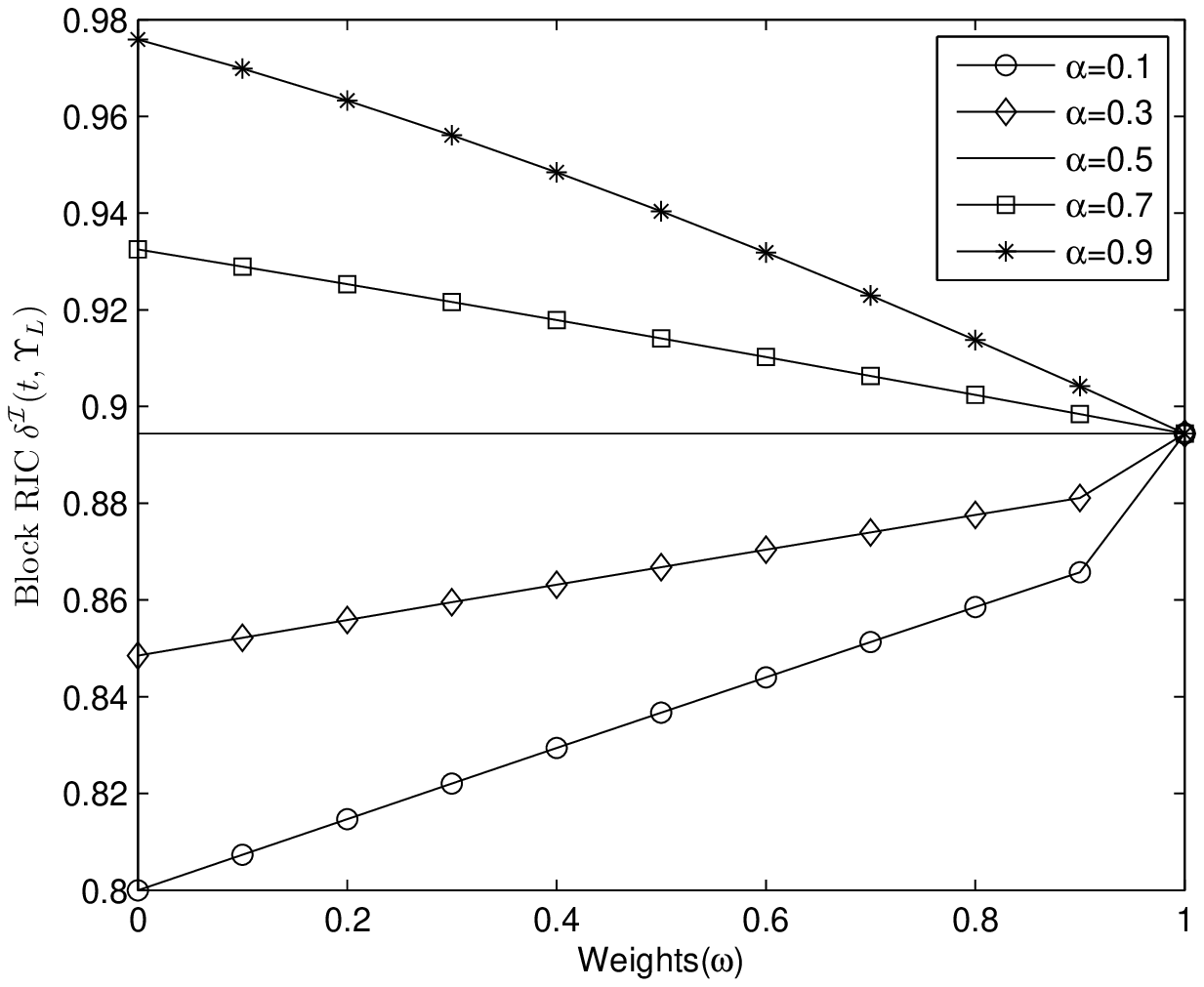}\\[-0.3cm]
{(a) $\delta^\mathcal{I}(t,\Upsilon_L^{\omega})$ versus $\omega$}
\end{minipage}
\\
\centering
\begin{minipage}[htbp]{0.46\linewidth}
\centering
\includegraphics[width=3.0in]{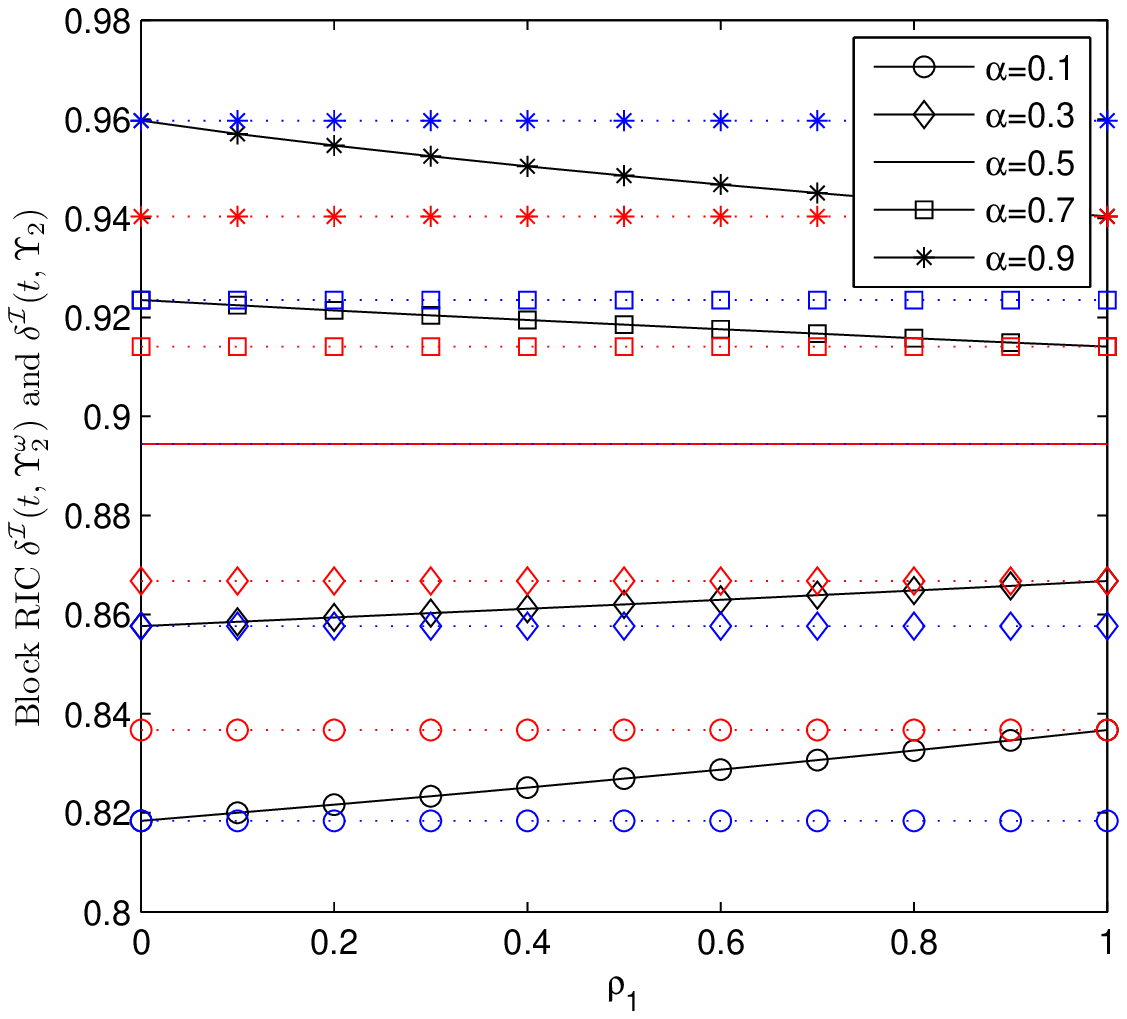}\\[-0.3cm]
{(b) $\delta^\mathcal{I}(t,\Upsilon_2^\omega)$ and $\delta^\mathcal{I}(t,\Upsilon_2)$ versus $\rho_1$}
\end{minipage}
\centering
\begin{minipage}[htbp]{0.46\linewidth}
\centering
\includegraphics[width=3.0in]{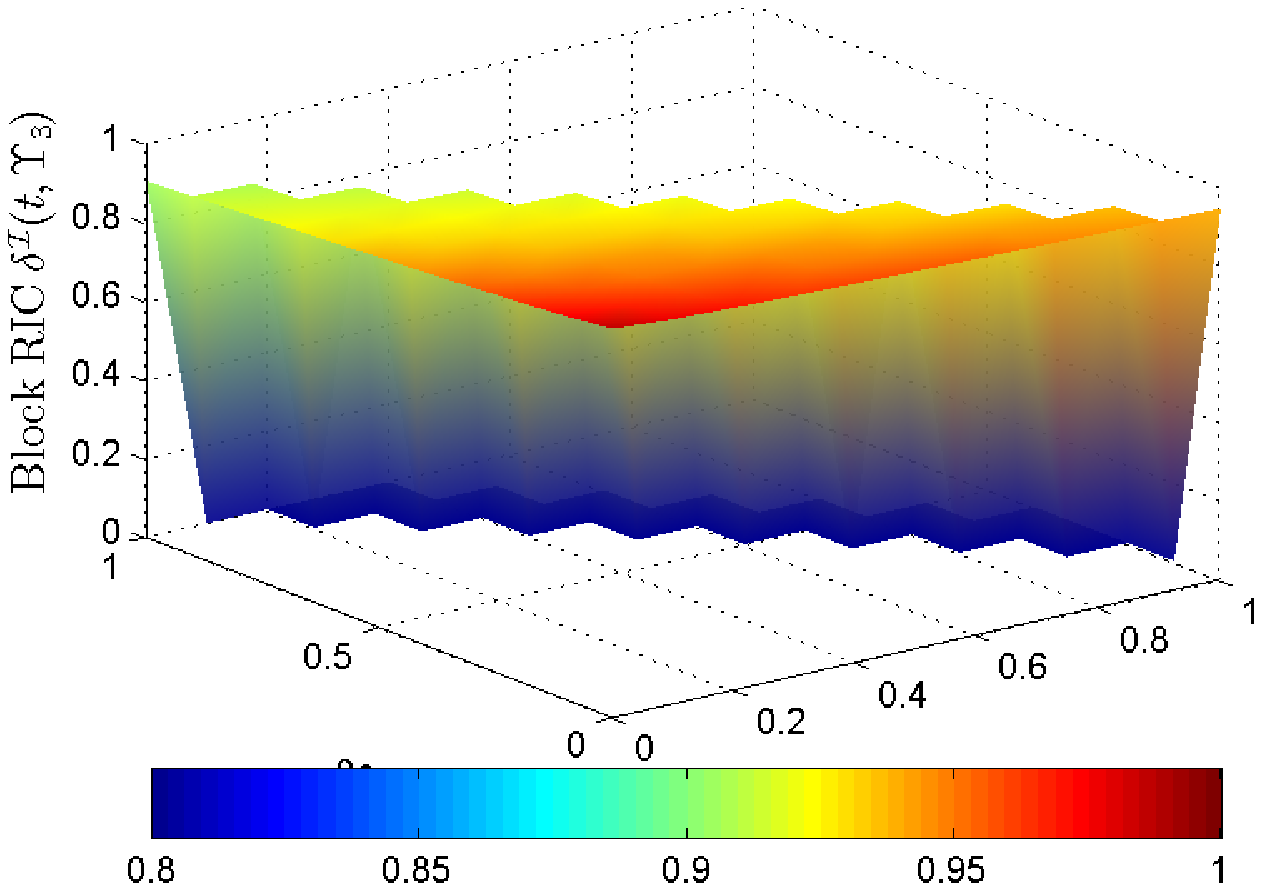}\\[-0.3cm]
{(c) $\delta^\mathcal{I}(t,\Upsilon_3)$ versus $\rho_1$ and $\rho_2$}
\end{minipage}
\\
\centering
\begin{minipage}[htbp]{0.46\linewidth}
\centering
\includegraphics[width=3.0in]{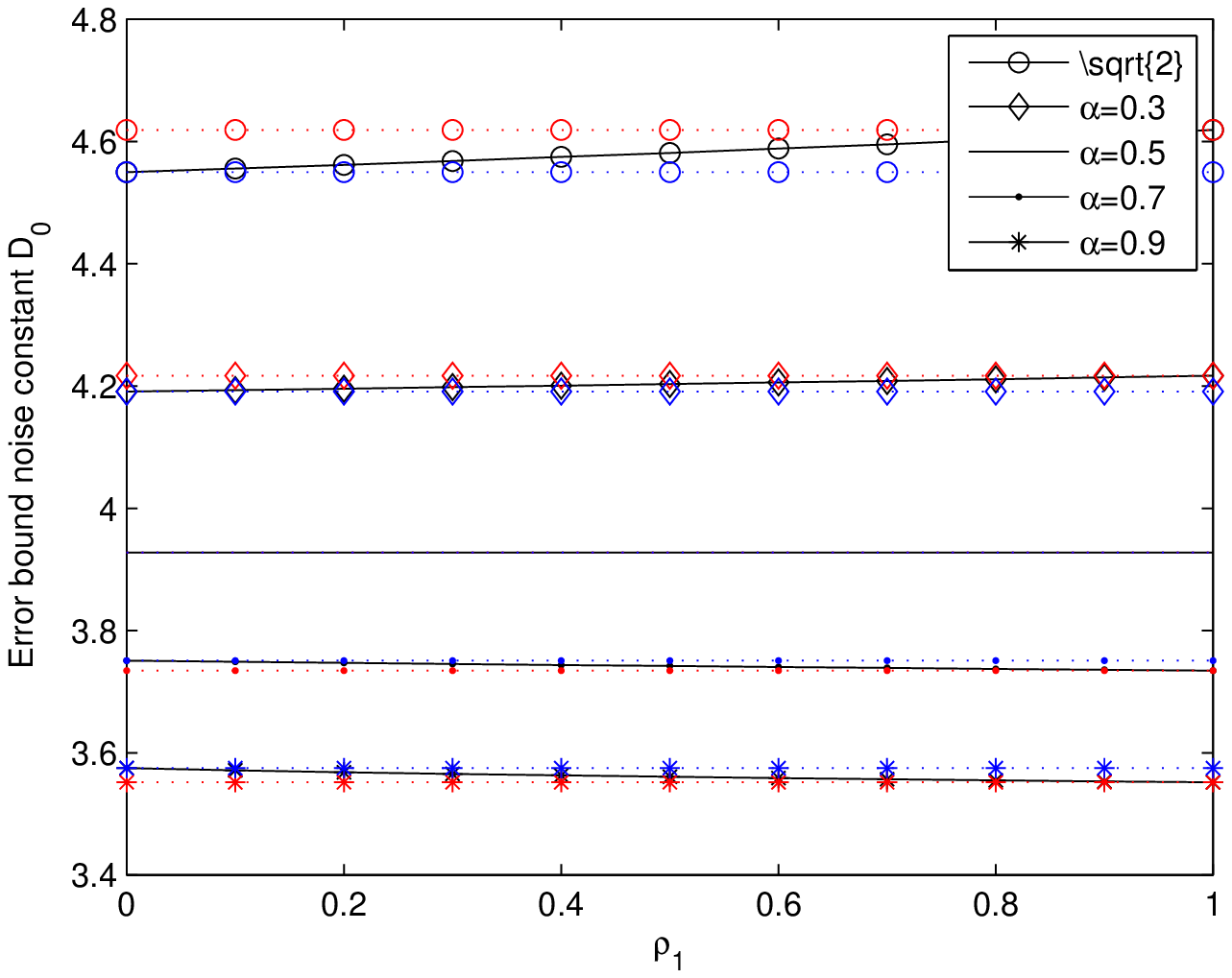}\\[-0.3cm]
{($\mathrm{d}$) $D_{0}$ versus $\rho_1$}
\end{minipage}
\begin{minipage}[htbp]{0.46\linewidth}
\centering
\includegraphics[width=3.0in]{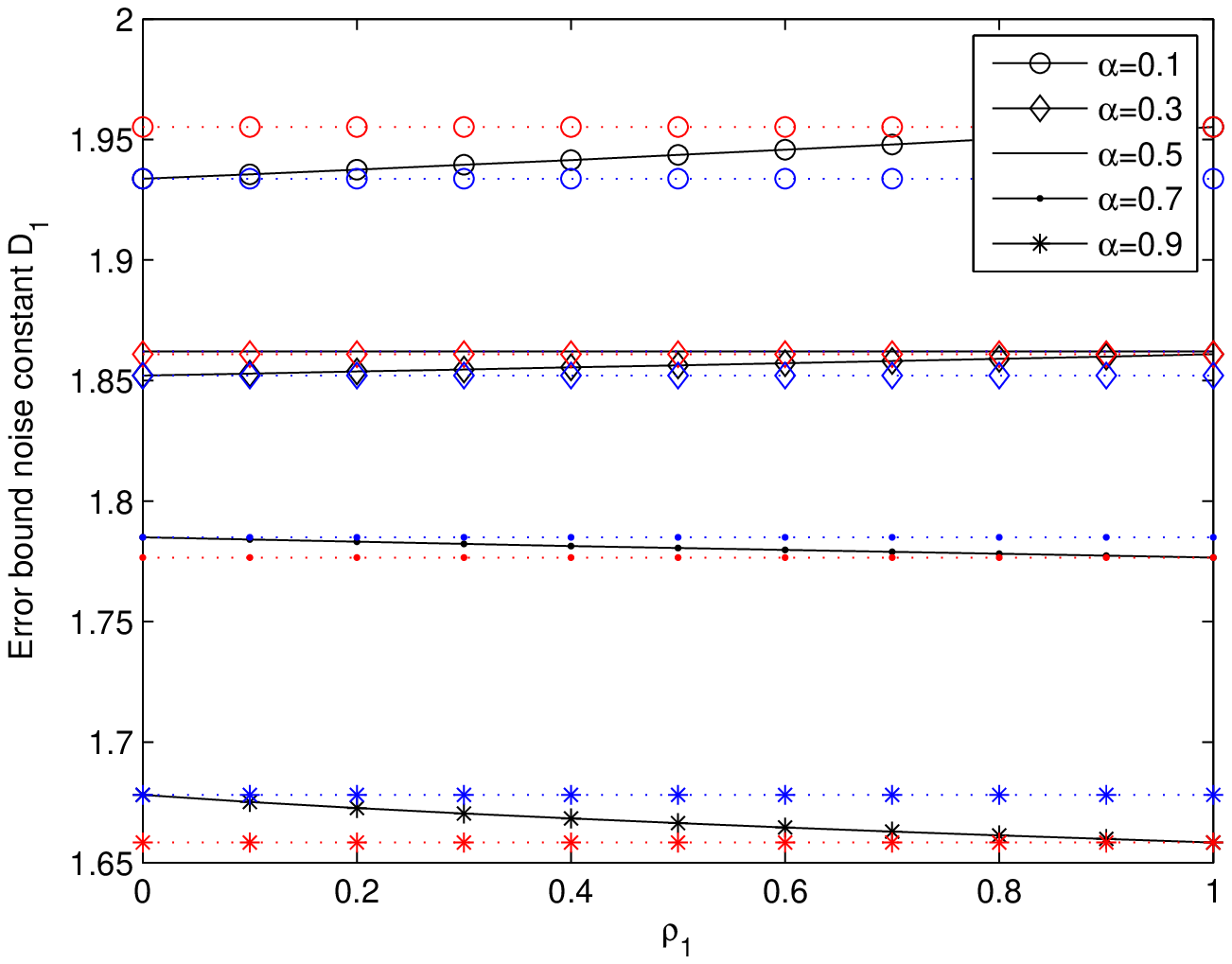}\\[-0.3cm]
{($\mathrm{d}'$) $D_{1}$ versus $\rho_1$}
\end{minipage}
\\
\caption{\label{fig1}Comparison of the sufficient conditions for
recovery and stability constants for the weighted $\ell_2/\ell_{1}$
 reconstruction. In all figures, we set $t=5$.
 In ($\mathrm{b}$),  ($\mathrm{d}$) and ($\mathrm{d}'$), the red dotted lines and the blue dotted lines
  indicate respectively the
  cases of $\omega=0.5$ and $\omega=0.25$ while the two weights case uses the solid lines.
 In ($\mathrm{d}$) and ($\mathrm{d}'$), we fix $\delta_{tk}^{\mathcal{I}}=0.5$.  }
\end{figure*}
\begin{figure*}[htbp!]
\centering
\begin{minipage}[htbp]{0.46\linewidth}
\centering
\includegraphics[width=4in]{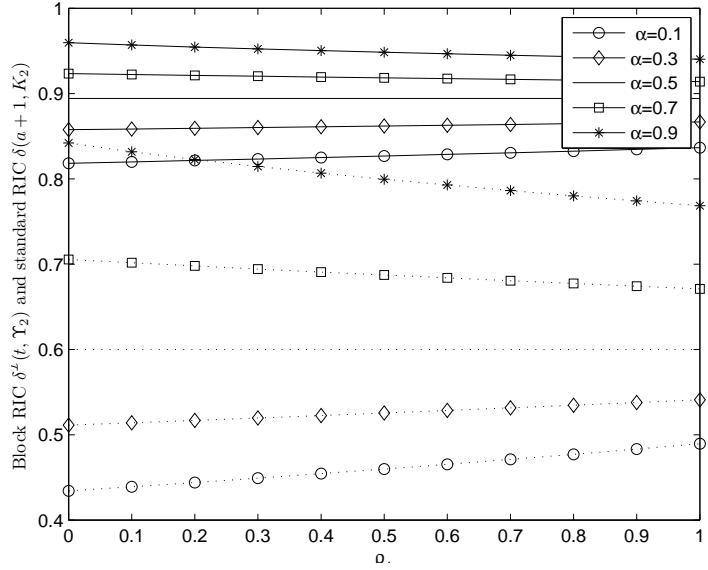}\\[-0.3cm]
{(a) $\delta^\mathcal{I}(t,\Upsilon_2)$ and $\delta(a+1,K_2)$ versus $\rho_1$}
\end{minipage}
\\
\centering
\begin{minipage}[htbp]{0.46\linewidth}
\centering
\includegraphics[width=3.0in]{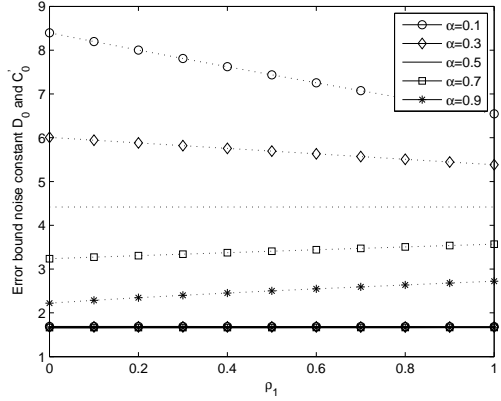}\\[-0.3cm]
{(b) $D_0$ and $C_0^{'}$ versus $\rho_1$}
\end{minipage}
\centering
\begin{minipage}[htbp]{0.46\linewidth}
\centering
\includegraphics[width=3.0in]{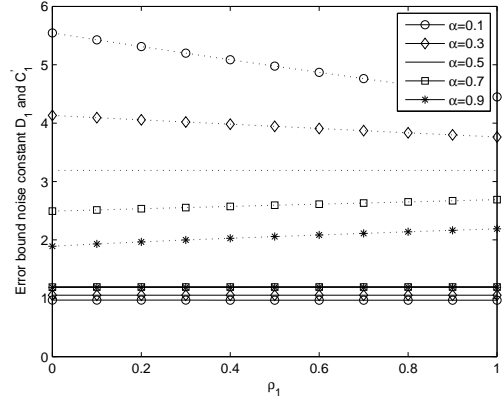}\\[-0.3cm]
{(c) $D_1$ and $C_1^{'}$ versus $\rho_1$}
\end{minipage}
\\
\caption{\label{fig2}Comparison of
$\delta^\mathcal{I}(t,\Upsilon_2)$ (solid lines) and
$\delta(a+1,K_2)$ (dotted lines) and stability constants. In all the
figures, we set $t=5$.
 In ($\mathrm{b}$) and ($\mathrm{c}$), comparing  stability constants $D_0$ (solid lines) and $C_0^{'}$ (dotted lines) as well as  $D_1$ (solid lines) and $C_1^{'}$ (dotted lines)
  we fix $\delta_{tk}^\mathcal{I}=0.2$ and $\delta_{ak}=0.15$.}
\end{figure*}

Fig.\ref{fig1} illustrates how  the bound $\delta^\mathcal{I}(t, \Upsilon_L)$ of the block RIP
constant $\delta_{tk}^{\mathcal{I}}$ given in \eqref{g23} and the stability constants given in
\eqref{g24} change with  weights  and  the prior block support estimate sizes for the different accuracy of prior block support estimate in
the case of weighted $\ell_2/\ell_{1}$ when $t=5$. And
 we also compare  the bound $\delta^\mathcal{I}(t, \Upsilon_L)$  and the stability constants when the
single weight is used with that using two or three distinct weights as a function of the block support estimate sizes
for the different accuracy of prior support estimate.

In Fig.\ref{fig1}(a), we set $\omega_1=\omega_2=\cdots=\omega_L=\omega\in[0,1]$, $\alpha_1=\alpha_2=\cdots=\alpha_L=\alpha$
and $\rho_1+\rho_2+\cdots+\rho_L=\rho=1$. That is, we only consider the weighted $\ell_2/\ell_1$ minimization with the single weight $\omega$,
the block support estimate size $\rho$ and the  accuracy $\alpha$.
We plot the bound $\delta^\mathcal{I}(t, \Upsilon_L^\omega)$  versus $\omega$ with different values of $\alpha$. We
observe that the bound $\delta^\mathcal{I}(t, \Upsilon_L^\omega)$ gets larger as $\alpha$
increases, which implies  the sufficient condition on the block RIP
constant becomes weaker as $\alpha$ increases in the case of the weighted $\ell_2/\ell_1$ minimization with the single weight.
 And when $\omega=1$ or $\alpha=\frac{1}{2}$, $\delta^\mathcal{I}(t, \Upsilon_L^\omega)$ is a constant (see Proposition \ref{pro6}). In addition,
 $\delta^\mathcal{I}(t, \Upsilon_L^\omega)$ decreases as $\omega$  increases with $\alpha>\frac{1}{2}$, which means the condition \eqref{g23} is weaker for smaller weight $\omega$.

In Fig.\ref{fig1}(b), we compare the bound $\delta^\mathcal{I}(t, \Upsilon_L)$ ($L=2$) when using either two disjoint prior block support estimates
 $\widetilde{T}_1$ and $\widetilde{T}_2$ or a single prior block support estimate $\widetilde{T}=\widetilde{T}_1\cup \widetilde{T}_2$,
which implies $\rho=\rho_1+\rho_2$.
Let  $\omega_1=0.5$ (applied on $\widetilde{T}_1$), $\omega_2=0.25$ (applied on $\widetilde{T}_2$), the single weight $\omega=0.5$ or $0.25$
(applied on $\widetilde{T}$), $\rho=1$ and $\alpha_1=\alpha_2=\alpha$.
The figure displays  the bounds  $\delta^\mathcal{I}(t, \Upsilon_L)$   and $\delta^\mathcal{I}(t, \Upsilon_L^\omega)$  as a function of the size  $\rho_1$ for  different $\alpha$.
As expected, the bounds $\delta^\mathcal{I}(t, \Upsilon_L)$  and  $\delta^\mathcal{I}(t, \Upsilon_L^\omega)$  get larger as $\alpha$ increases   both in the single and two weights cases.
 The bound $\delta^\mathcal{I}(t, \Upsilon_L)$ lies between the bound $\delta^\mathcal{I}(t, \Upsilon_L^\omega)$
  in the single weight $\omega=0.5$  case
 and the bound $\delta^\mathcal{I}(t, \Upsilon_L^\omega)$  applied the single weight  $\omega=0.25$
when $\rho_1\in[0,1]$ and $\alpha\neq 0.5$. In addition,
 when $\alpha\geq \frac{1}{2}$, the figure demonstrates the result of Proposition \eqref{pro5},
i.e., $\delta^\mathcal{I}(t, \Upsilon_L^{0.5}) \leq \delta^\mathcal{I}(t, \Upsilon_L) \leq \delta^\mathcal{I}(t, \Upsilon_L^{0.25})$.

 Fig.\ref{fig1}(c)  displays the transition of $\delta^\mathcal{I}(t, \Upsilon_L)$ as $\rho_1$ and $\rho_2$  vary  with $L=3$,
$\rho_1+\rho_2+\rho_3=1$, $\alpha_1=\alpha_2=\alpha_3=0.9$, $\omega_1=0.9,\ \omega_2=0.5$ and $\omega_3=0.1$.

In Fig.\ref{fig1}(d) and ($\mathrm{d}'$),  set $L=2$, $\delta_{tk}^{\mathcal{I}}=0.5$ and $\alpha_1=\alpha_2=\alpha$.
And we set $\omega_1=0.5$, $\omega_2=0.25$ and the single weight $\omega=0.5$ or $0.25$. One can easily see that
  $D_{0}$ and $D_1$ in \eqref{g24} decreases as $\alpha$ increases  for the case of  two distinct  weights
 and the cases of the single weight.   We  observe that constants $D_{0}$ and $D_1$ with two distinct  weights
 lie between those with  a single weight  for the accuracy $\alpha$.  For $\alpha>\frac{1}{2}$,
the smallest weight results  in  the best (smallest) constants $D_{0}$ and $D_1$  and the largest weight results in the worst (largest) constants $D_{0}$ and $D_1$.\,

Fig.\ref{fig2} compares the bound  $\delta^\mathcal{I}(t, \Upsilon_L)$  in \eqref{g23} for $\mathcal{I}=\{d_1=1, d_2=1,\cdots,d_M=1\}$  with the bound of the standard  RIC
 $\delta(a+1, K_L)$ in
\eqref{g7} as well as stability constants in \eqref{g24} and
\eqref{g6} for various accuracy $\alpha_1=\alpha_2=\alpha$. Set
$L=2$, $\rho_1+\rho_2=1$,  $t=5, ~a=4$, $\omega_1=0.5$ and
$\omega_2=0.25$. Here we depict the bounds $\delta^\mathcal{I}(t,
\Upsilon_L)$ in \eqref{g23}, $\delta(a+1, K_L)$ in \eqref{g7} and
the constants in \eqref{g24} and \eqref{g6} versus $\rho_1$ with
various $\alpha$.

Fig.\ref{fig2}(a) illustrates $\delta^\mathcal{I}(t,\Upsilon_L)$ is
larger than $\delta(a+1,K_L)$ under the same support estimate.
Moreover, Fig.\ref{fig2}(b) and (c) describe that constants $D_{0}$
and $D_{1}$ are always smaller than $C'_{0}$ and $C'_{1}$,
respectively. Therefore,  the sufficient condition (\ref{g23}) is
weaker than (\ref{g7}), and error bound constants (\ref{g24}) in
Theorem \ref{t3}  are better than those (\ref{g6}) in Theorem
\ref{theor1}.

\section{Random matrices}\label{4}
Theorem \ref{t3} established that the  block $k$-sparse signal $x$ can be  exactly recovered
under a sufficient condition $\delta_{tk}^\mathcal{I}<\delta^\mathcal{I}(t,\Upsilon_L)=\sqrt{\frac{t-d}{t-d+\Upsilon_L^{2}}}$.
In this section, we  prove that how many random measurements are needed for $\delta_{tk}^\mathcal{I}<\delta^\mathcal{I}(t,\Upsilon_L)$ to be satisfied with high probability.
Firstly, we recall Lemma 5.1 of \cite{BDDW}, which plays an important role in the proof of Theorem \ref{t5}.
\begin{lem}\label{lem1}(\cite{BDDW}\ Lemma 5.1)
Let $\Phi(\omega),\ \omega\in \Omega^{nN}$,  be a random matrix of size $n\times N$ drawn according to any
 distribution that satisfies
the concentration inequality
\begin{eqnarray}\label{cgl5}
P(|\|\Phi(\omega)x\|_2^2-\|x\|_2^2|\geq \varepsilon \|x\|_2^2)\leq2e^{-nc_0(\varepsilon)},\ \ 0<\varepsilon<1,
\end{eqnarray}
 where $c_0(\varepsilon)$
is a constant depending on $\varepsilon$. Then, for any set
$T$ with $|T|=k<n$ and any $0<\delta<1$, we have that
\begin{eqnarray}\label{cgl3}
(1-\delta)\|x\|_2\leq\|\Phi(\omega)x\|_2\leq (1+\delta)\|x\|_2,\ \ \ \ \mathrm{for\  all}\  x\in X_T
\end{eqnarray}
with probability
\begin{eqnarray}\label{cgl4}
\geq 1-2\bigg(\frac{12}{\delta}\bigg)^k\exp(-c_0(\delta/2)n),
\end{eqnarray}
where $X_T$ denotes the set of all signals in $\R^N$ that are zero outside of $T$.
\end{lem}

In the section,
we  consider special random measurement matrices $A=(A_{ij})_{n\times N}$,
where
\begin{eqnarray}\label{matr1}
&A_{ij}\sim \mathcal{N}(0,1/n),\ \ \ \  A_{ij}=\left\{
                                                \begin{array}{ll}
                                                  1/\sqrt{n}, & \hbox{w.p. $1/2$} \\
                                                  - 1/\sqrt{n}, & \hbox{w.p. $1/2$}
                                                \end{array}
                                             \right.\nonumber \\
&\mathrm{or} \ \ \ \
A_{ij}=\left\{
         \begin{array}{ll}
           \sqrt{3/n}, & \hbox{w.p. $1/6$} \\
           0, & \hbox{w.p. $2/3$} \\
           -\sqrt{3/n}, & \hbox{w.p. $1/6$.}
         \end{array}
       \right.
\end{eqnarray}
Achlioptas \cite{A} showed  that the above  random measurement matrices \eqref{matr1} satisfy
\eqref{cgl5} with $c_0(\varepsilon)=\varepsilon^2/4-\varepsilon^3/6$.
Therefore,  for each of the $k$-dimensional spaces $X_T$, random measurement matrices \eqref{matr1}
will fail to satisfy \eqref{cgl3} with probability
\begin{eqnarray}\label{cgl6}
\leq 2\bigg(\frac{12}{\delta}\bigg)^k\exp\bigg(-n\bigg(\frac{\delta^2}{16}-\frac{\delta^3}{48}\bigg)\bigg)
\end{eqnarray}
by  Lemma \ref{lem1}.
\begin{thm}\label{t5}
For random measurement matrices \eqref{matr1}, suppose
\begin{align*}
n\geq\frac
{tk\log\frac{M}{k}}{\frac{t-d}{16(t-d+\Upsilon_L^{2})}-\frac{\left(\sqrt{(t-d)/(t-d+\Upsilon_L^{2})}\right)^3}{48}}.
\end{align*}
Then $\delta_{tk}^\mathcal{I}<\delta^\mathcal{I}(t,\Upsilon_L)=\sqrt{(t-d)/(t-d+\Upsilon_L^{2})}$ ($t>d$) holds in high probability.
\end{thm}
\begin{proof}
Without loss of generality, let $tk$ is a  positive integer.
By \eqref{cgl6}, a $n\times tk\hat{d} $ submatrix of random measurement matrices $A$ \eqref{matr1}
fails to fulfil \eqref{cgl3} with probability
$$\leq 2\bigg(\frac{12}{\delta^\mathcal{I}(t,\Upsilon_L)}\bigg)^{tk\hat{d}}\exp\bigg(-n\bigg(\frac{(\delta^\mathcal{I}(t,\Upsilon_L))^2}{16}-\frac{(\delta^\mathcal{I}(t,\Upsilon_L))^3}{48}\bigg)\bigg).$$

As discussed in \cite{EM}, we know that a block sparse signal lies in a structured
union of subspaces.
Then random measurement matrices \eqref{matr1} fail to satisfy \eqref{brip}
with probability
\begin{eqnarray*}
\leq 2\left(^M_{tk}\right)\bigg(\frac{12}{\delta^\mathcal{I}(t,\Upsilon_L)}\bigg)^{tk\hat{d}}\exp\bigg(-n\bigg(\frac{(\delta^\mathcal{I}(t,\Upsilon_L))^2}{16}-\frac{(\delta^\mathcal{I}(t,\Upsilon_L))^3}{48}\bigg)\bigg).
\end{eqnarray*}

Note that $\left(^M_{tk}\right)\leq (\frac{eM}{tk})^{tk}$. Then for
$t>d$ and
$\delta^\mathcal{I}(t,\Upsilon_L)=\sqrt{(t-d)/(t-d+\Upsilon_L^{2})}$,
we have
\begin{align*}
P(\delta_{tk}^\mathcal{I}\geq \delta^\mathcal{I}(t,\Upsilon_L))
 &\leq 2\left(^M_{tk}\right)\bigg(\frac{12}{\delta^\mathcal{I}(t,\Upsilon_L)}\bigg)^{tk\hat{d}}
 \exp\bigg(-n\bigg(\frac{(\delta^\mathcal{I}(t,\Upsilon_L))^2}{16}-\frac{(\delta^\mathcal{I}(t,\Upsilon_L))^3}{48}\bigg)\bigg)\\
&\leq 2\bigg(\frac{eM}{tk}\bigg)^{tk}\bigg(\frac{12}{\sqrt{(t-d)/(t-d+\Upsilon_L^{2})}}\bigg)^{tk\hat{d}}\\
&\ \ \ \ \times \exp\bigg(-n\bigg(\frac{t-d}{16(t-d+\Upsilon_L^{2})}-\frac{\left(\sqrt{(t-d)/(t-d+\Upsilon_L^{2})}\right)^{3}}{48}\bigg)\bigg)\\
&=2\exp\bigg(-n\bigg(\frac{t-d}{16(t-d+\Upsilon_L^{2})}-\frac{\left(\sqrt{(t-d)/(t-d+\Upsilon_L^{2})}\right)^{3}}{48}\bigg)\\
&\ \ \ \ \ \ \ \ \ \ +tk\bigg(\log\frac{eM}{tk}
+\hat{d}\log\frac{12}{\sqrt{(t-d)/(t-d+\Upsilon_L^{2})}}\bigg)\bigg)
\end{align*}
 Hence,
\begin{align*}
&P(\delta_{tk}^\mathcal{I}<\sqrt{(t-d)/(t-d+\Upsilon_L^{2})})\\
&\geq
1-2\exp\bigg(-n\bigg(\frac{t-d}{16(t-d+\Upsilon_L^{2})}-\frac{\left(\sqrt{(t-d)/(t-d+\Upsilon_L^{2})}\right)^{3}}{48}\bigg)\\
&\ \ \ \ \ \ \ \ \ \ \ \  \  +tk\bigg(\log\frac{eM}{tk}
+\hat{d}\log\frac{12}{\sqrt{(t-d)/(t-d+\Upsilon_L^{2})}}\bigg)\bigg)
\end{align*}
It is easy to see that the  random measurements $n\geq\frac
{tk\log(M/k)}{(t-d)/(16(t-d+\Upsilon_L^{2}))-((t-d)/(t-d+\Upsilon_L^{2}))^{3/2}/48}$
when $M/k\rightarrow\infty$ to sure
$\delta_{tk}^\mathcal{I}<\sqrt{(t-d)/(t-d+\Upsilon_L^{2})}$ ($t>d$)
to hold in high probability. We have completed the proof of the
theorem.
\end{proof}
\section{Numerical experiments}\label{5}
In this section, we present several numerical experiments to compare the weighted  $\ell_2/\ell_1$ minimization
method  with  the  $\ell_2/\ell_1$ minimization method
in the context of block signal recovery. By numerical experiments, we illustrate the benefits of the weighted $\ell_2/\ell_1$ minimization to
recover block sparse signals in both noiseless and noisy cases.
 In addition, we also demonstrate that non-uniform block support information  can be preferable to
 uniform  block support information.

For the solution of the  $\ell_2/\ell_1$ minimization problem, Wang et.al adopt an efficient iteratively reweighted least squares (IRLS)
algorithm \cite{WWX1}, \cite{WWX2}. Inspired by the ideas of \cite{WWX1},
 we present a generalized algorithm of the IRLS to solve the weighted $\ell_2/\ell_1$ minimization problem
\eqref{pro4} with \eqref{pro7}. First, we rewrite the problem \eqref{pro4} as the following regularized unconstrained smoothed weighted $\ell_2/\ell_1$ minimization
\begin{align}\label{cgl14}
\min_x\|x_\mathrm{w}\|_{2,1}^\varepsilon+\frac{1}{2\tau}\|y-Ax\|_2^2,
\end{align}
where $\|x_\mathrm{w}\|_{2,1}^\varepsilon=\sum\limits_{i=1}^M\mathrm{w}_i(\|x[i]\|_2^2+\varepsilon^2)^{1/2}$ and $\mathrm{w}_i \in (0,1]$.  Let
\begin{align*}
f(x,\varepsilon,\tau)=\sum_{i=1}^M\mathrm{w}_i(\|x[i]\|_2^2+\varepsilon^2)^{1/2}+\frac{1}{2\tau}\|y-Ax\|_2^2
\end{align*}
 be the objective function associated with the minimization problem \eqref{cgl14}. For the solution of $x$,
  it is known that the first-order necessary condition
 is
 \begin{align*}
\left[\frac{\mathrm{w}_ix[i]}{(\|x[i]\|_2^2+\varepsilon^2)^{1/2}}\right]_{1\leq i\leq M}+\frac{1}{\tau}A'(Ax-y)=0.
 \end{align*}
Let  the block vector $\widetilde{x}\in\R^N$ over $\mathcal{I}=\{d_1,d_2,\ldots,d_M\}$ satisfy
\begin{align*}
\widetilde{x}[i]=(\sqrt{\mathrm{w}_i}(\|x[i]\|_2^2+\varepsilon^2)^{-1/4},
\ldots ,\sqrt{\mathrm{w}_i}(\|x[i]\|_2^2+\varepsilon^2)^{-1/4})'\in\R^{d_i}
\end{align*}
 for all  $i\in[M]$. Define the diagonal weighting matrix $W=\mathrm{diag}(\widetilde{x})$,
 Therefore, we obtain the necessary optimality condition $(\tau W^2+A'A)x=A'y$. Due to the nonlinearity of the above system, we apply an iterative method to solve the above equations. That
 is, if we fix $W=W^{(t)}$  to be that determined already in the $t$-th iteration step, we set the solution of the above equations
  $x^{(t+1)}=(W^{(t)})^{-1}((A(W^{(t)})^{-1})^{'}(A(W^{(t)})^{-1})+\tau I)^{-1}(A(W^{(t)})^{-1})^{'} y$
  as the $(t+1)$-th iterate.

By the above analysis, we extend naturally the IRLS algorithm to the above problem \eqref{cgl14} denoting by  Algorithm 1 as  following:

\textbf{Input}:\ \ \ \ \  measurements $y\in\R^n$, sensing matrix $A\in\R^{n\times N}$,  estimated block-sparsity $\hat{k}$,

\ \ \ \ \ \ \ \ \ \ \ \ \ \ weighted vector $\mathrm{w}\in \R^M$.

\textbf{Step\ 1}:\ \ \ \  choose appropriate parameter $\tau>0$, set iteration count $t=0$ and $\varepsilon_0=1$,

\ \ \ \ \ \ \ \ \ \ \ \ \ \ initialize $x^{(0)}=\arg\min \|y-Ax\|_2^2$ .

\textbf{Step\ 2}: ``stopping criterion is not met'' \textbf{do}

1: $W^{(t)}=\mathrm{diag}(\sqrt{\mathrm{w}_i}(\varepsilon_t^2+\|x^{(t)}[i]\|_2^2)^{-1/4})$, $i=1,\ldots,M$;

2: $B^{(t)}=A(W^{(t)})^{-1}$;

3: $x^{(t+1)}=(W^{(t)})^{-1}((B^{(t)})'B^{(t)}+\tau I)^{-1}(B^{(t)})'y$;

4: $\varepsilon_{t+1}=\min\{\varepsilon_t,\nu r(x^{(t+1)})_{\hat{k}+1}/N\}$;

5: t=t+1.

\textbf{End}

\textbf{Output} $x^{(t+1)}$ is an approximation solution.

In the  algorithm 1, $r(x^{(t+1)})_{\hat{k}+1}$ is the $(\hat{k}+1)$-th largest $\ell_2$ norm value of the block of $x^{(t+1)}$ in the decreasing order, $\nu\in (0,1)$
satisfies $\nu r(x^{(1)})_{\hat{k}+1}/N<1$ and $\tau$ is an appropriately chosen parameter, which controls the tolerance of noise term. Note that the algorithm 1 is
the IRLS when $\mathrm{w}_i=1$ for all $i=[M]$, i.e., $\Pi_{i=1}^L\omega_i=1$.
In this paper, we don't make a detailed analysis including convergence, local convergence rate and error bound of the algorithm leaving to the interested reader.

In all of our experiments, we apply the algorithm 1 to solve the
weighted $\ell_2/\ell_1$ minimization problem  with
$0<\Pi_{i=1}^L\omega_i\leq 1$. For the algorithm 1, we set the
estimated $\hat{k}=k$ and $\nu=0.7$. If $\varepsilon_{t+1}<10^{-7}$
or $\|x^{(t+1)}-x^{(t)}\|_2<10^{-8}$, the iteration terminates and
outputs $x^{(t+1)}$; otherwise, the maximum number of iterations is
$1000$. The measurement matrix $A\in \R^{n\times N}$ was generated
randomly  with i.i.d draws from a standard Gaussian distribution and
the measurement vector $y$ was observed from $y=Ax+z$, where $z$ was
zero-mean Gaussian noise with standard deviation $\sigma$ or zero
vector. In the noise-free case ($\sigma=0$), $\tau=10^{-3}$
 and the average exact recovery frequency over 50 experiments is plotted by the following figures.
  If $\|x^{(t+1)}-x\|_2/\|x\|_2\leq 10^{-4}$, the recovery is regarded exact.
 For the presence of noise ($\sigma=0.01$), $\tau=10^{-1}\max|A'y|$ and we draw up the average reconstruction signal to noise ratio (SNR) over 50 experiments. The SNR is given by
 $\mathrm{SNR}(x^{(t+1)},x)=20\log_{10}\|x\|_2^2/\|x^{(t+1)}-x\|_2^2$, where the measure of the SNR is dB.

\subsection{The  uniform weight  case}
We first  consider   the uniform weight $\omega\in(0,1]$, i.e., $\omega_1=\cdots=\omega_L=\omega$,
 applied on $\widetilde{T}=\cup_{i=1}^L\widetilde{T}_i$ for the block $k$-sparse signal $x$ over $\mathcal{I}$  with length $N=256$ and $k=10$,
  generated  by choosing $k$ blocks uniformly at random, where $\mathcal{I}=\{d_1=\hat{d},\ldots,d_M=\hat{d}\}$.
For these $k$ blocks, we choose the nonzero values from a standard Gaussian distribution.
Let $\alpha_1=\alpha_2=\cdots=\alpha_L=\alpha$.
\begin{figure*}[htbp!]
\centering
\subfigure[Noise Free]{
\begin{minipage}{4cm}
\centering
{$\alpha=0.8$}
\includegraphics[width=4cm,height=6cm]{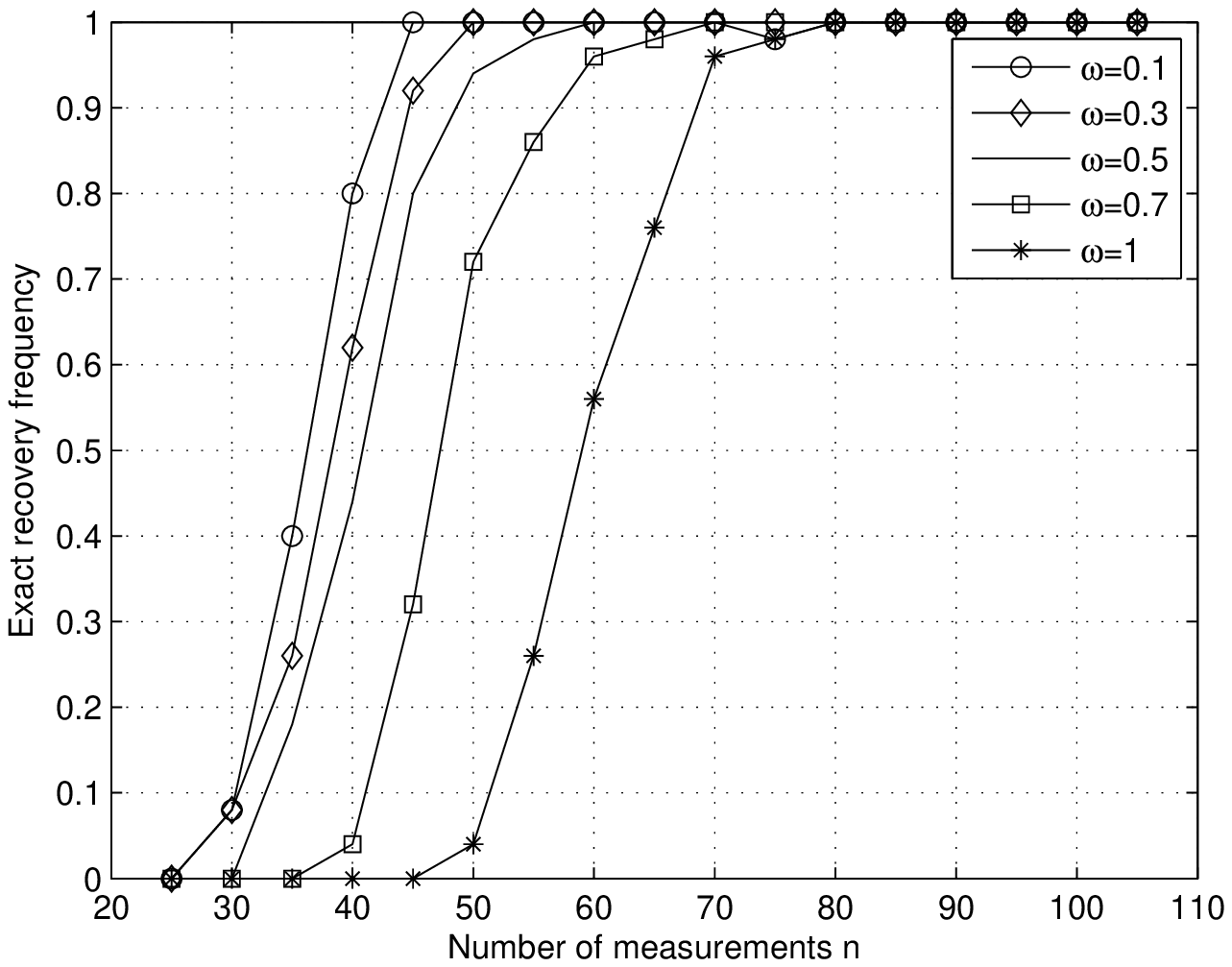}%\\[-0.3cm]
\end{minipage}
\centering
\begin{minipage}{4cm}
\centering
{$\alpha=0.5$}
\includegraphics[width=4cm,height=6cm]{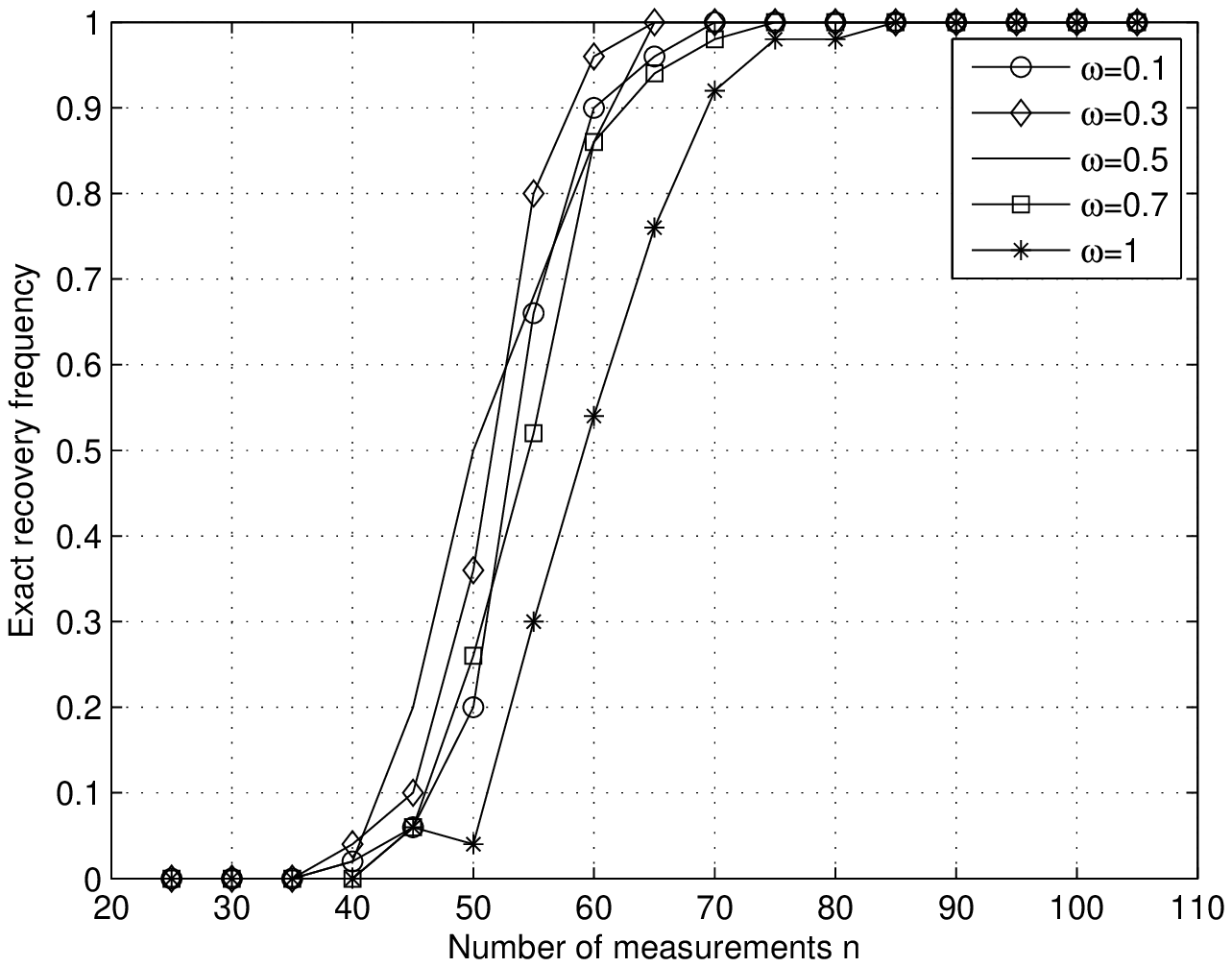}%\\[-0.3cm]
\end{minipage}
\centering
\begin{minipage}{4cm}
\centering
{$\alpha=0.2$}
\includegraphics[width=4cm,height=6cm]{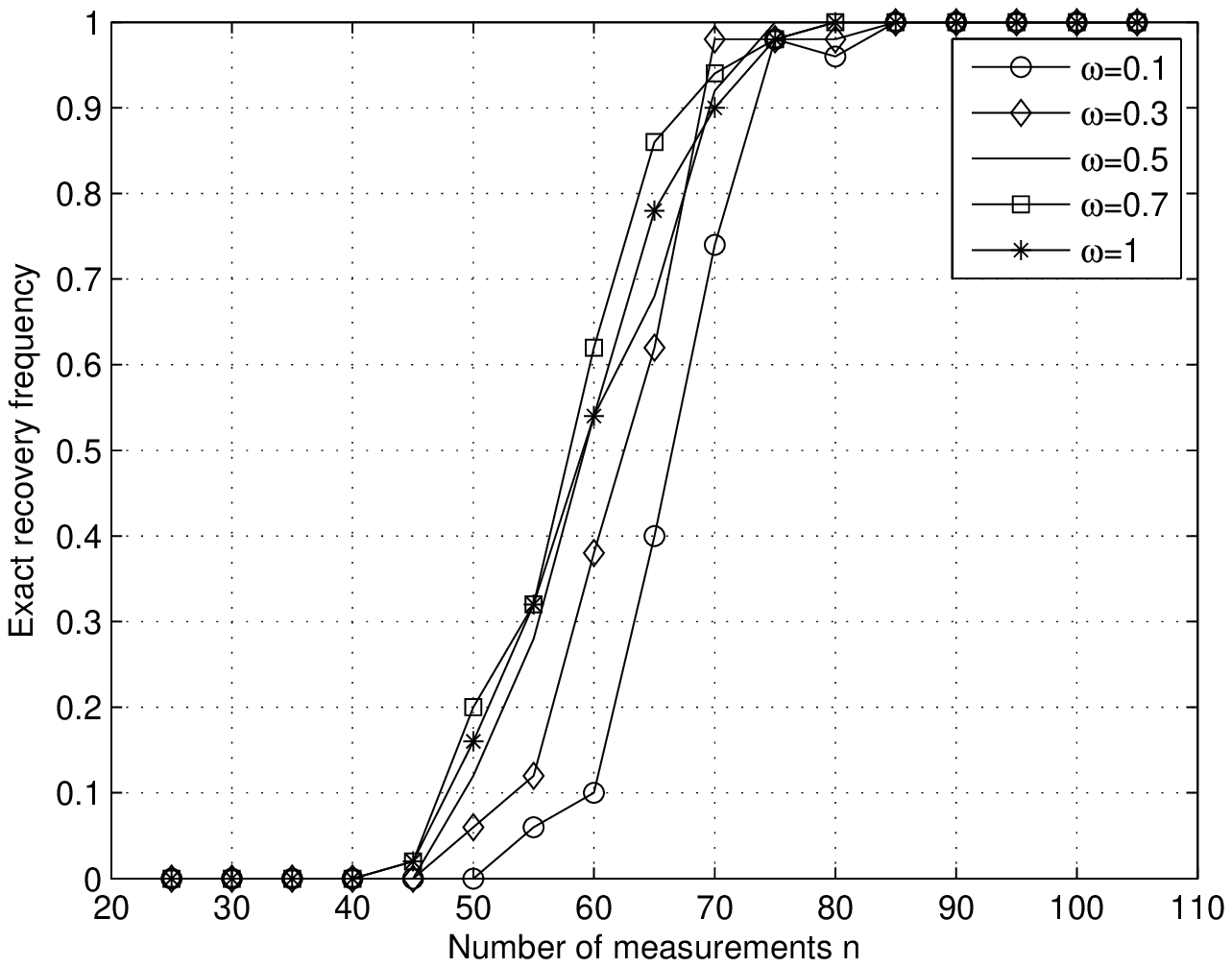}%\\[-0.3cm]
\end{minipage}}

\subfigure[Noisy Case ]{
\centering
\begin{minipage}{4cm}
\centering
{$\alpha=0.8$}
\includegraphics[width=4cm,height=6cm]{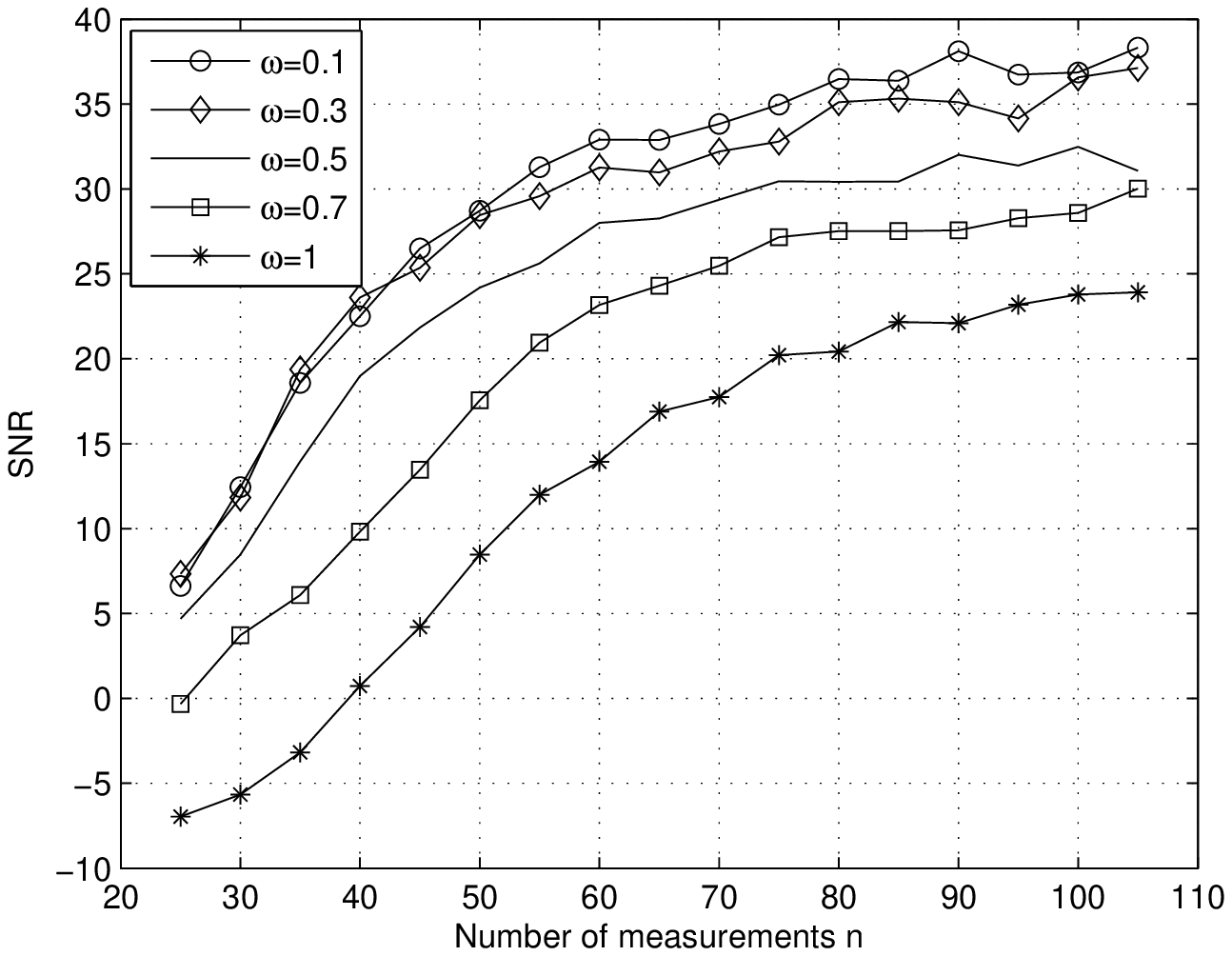}\\[-0.3cm]
\end{minipage}
\centering
\begin{minipage}{4cm}
\centering
{$\alpha=0.5$}
\includegraphics[width=4cm,height=6cm]{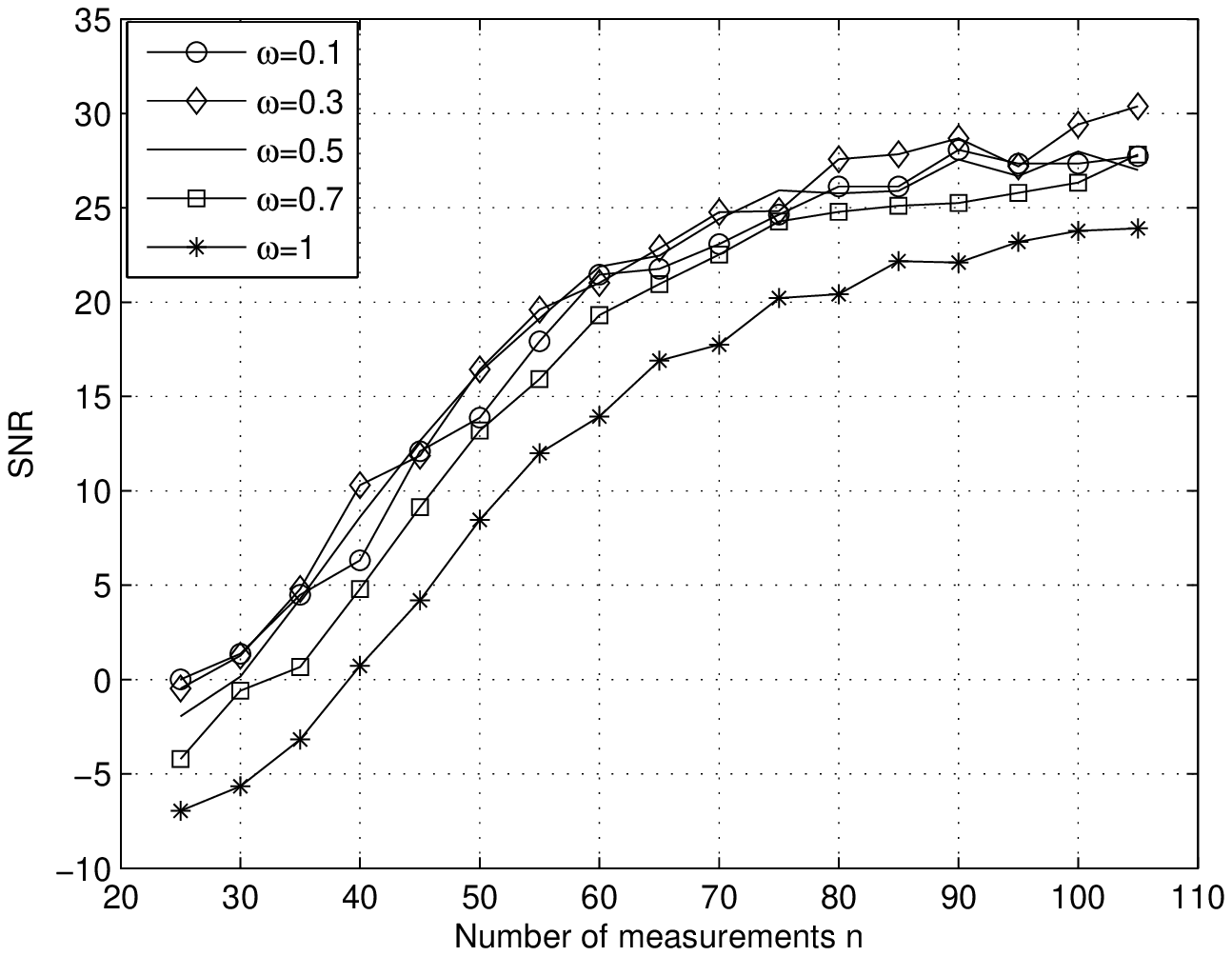}\\[-0.3cm]
\end{minipage}
\centering
\begin{minipage}{4cm}
\centering
{$\alpha=0.2$}
\includegraphics[width=4cm,height=6cm]{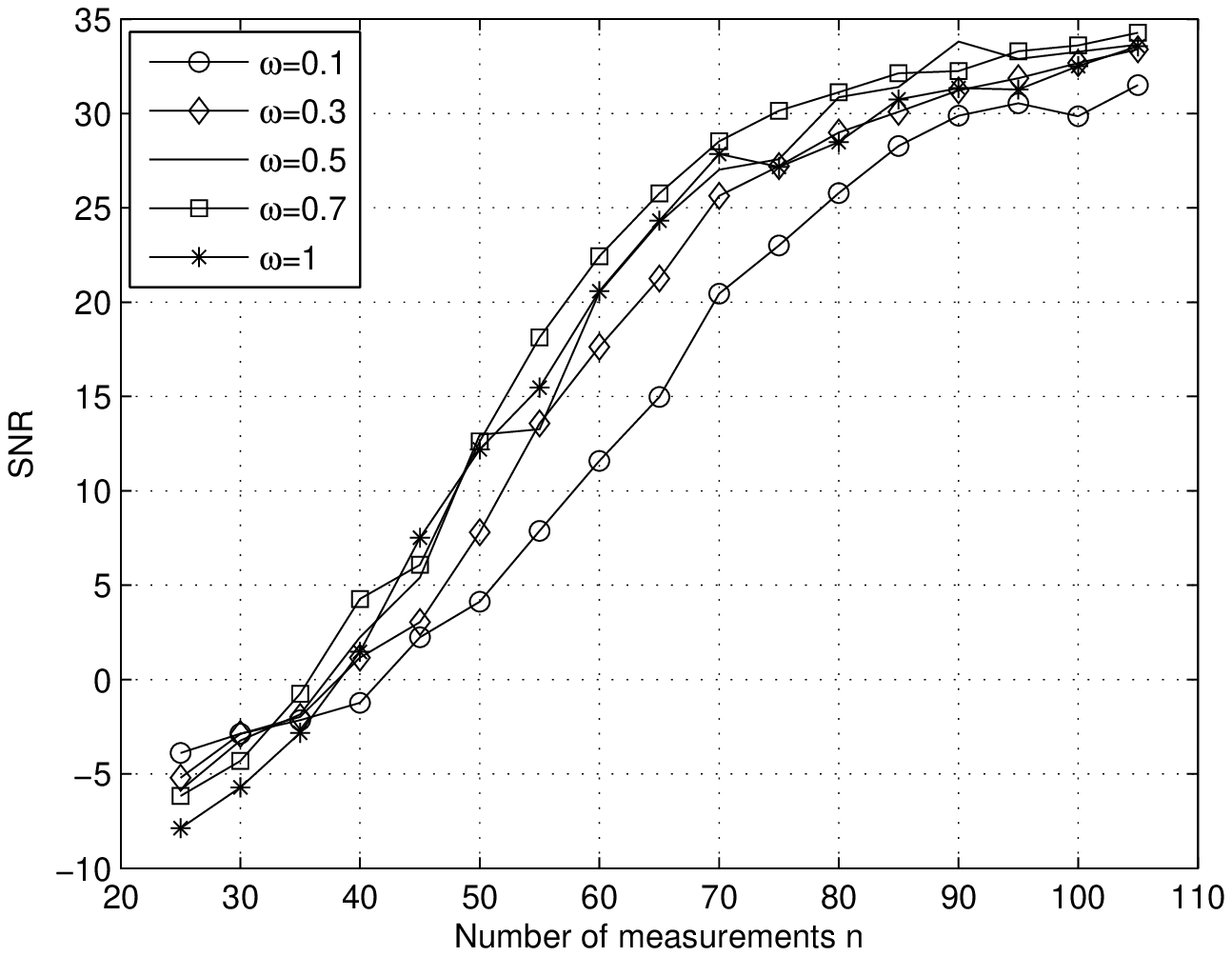}\\[-0.3cm]
\end{minipage}}
\\
\caption{\label{fig5} The recovery performance of the weighted $\ell_2/\ell_1$ minimization is in terms of the exact recovery frequency in noiseless case and the SNR  in
 noisy case. The block sparse signal $x$ with $\hat{d}=2$ has $k=10$ nonzero blocks. }
\end{figure*}

\begin{figure*}[htbp!]
\centering
\subfigure[Noise Free]{
\begin{minipage}{4cm}
\centering
{$\alpha=0.8$}
\includegraphics[width=4cm,height=6cm]{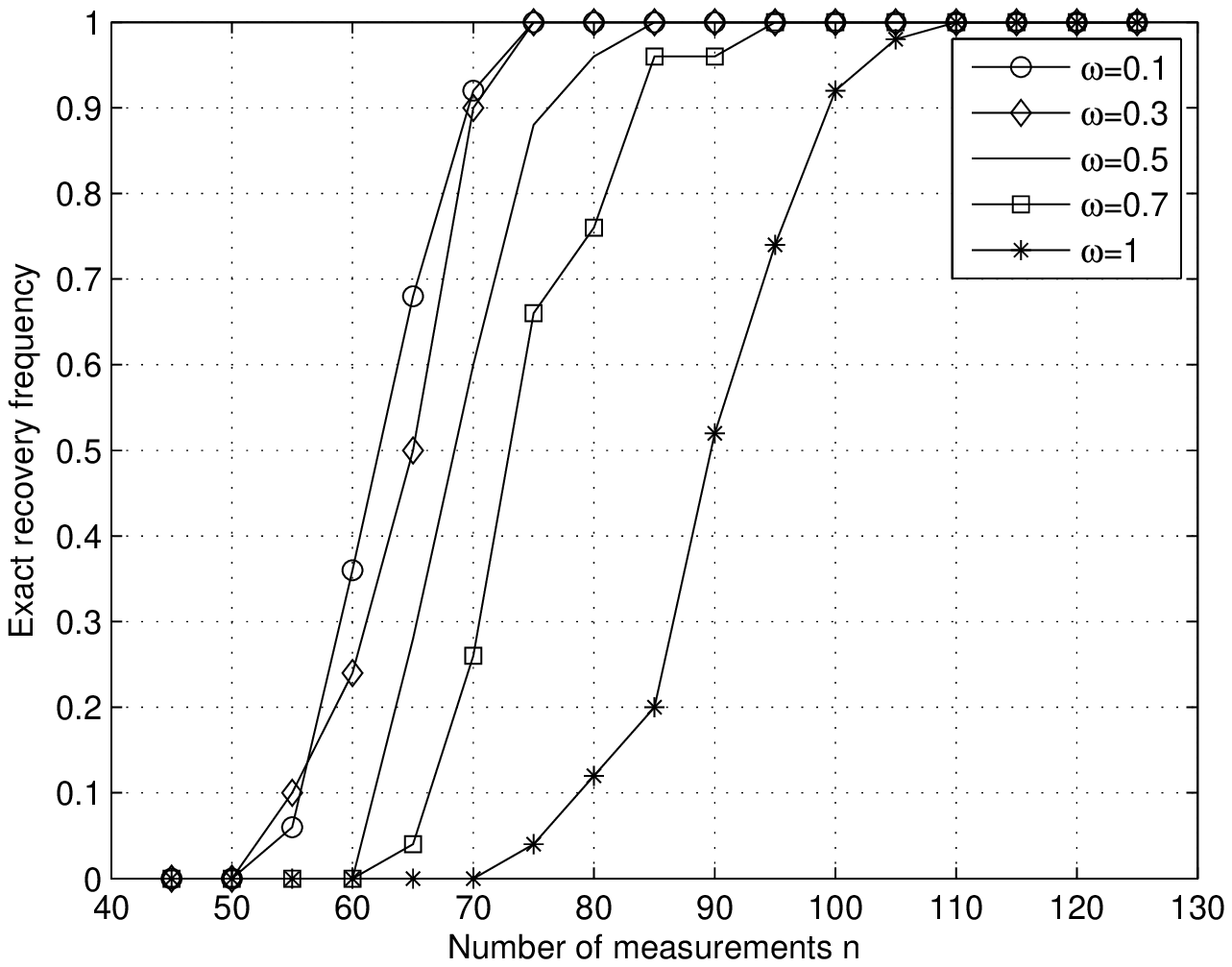}%\\[-0.3cm]
\end{minipage}
\centering
\begin{minipage}{4cm}
\centering
{$\alpha=0.5$}
\includegraphics[width=4cm,height=6cm]{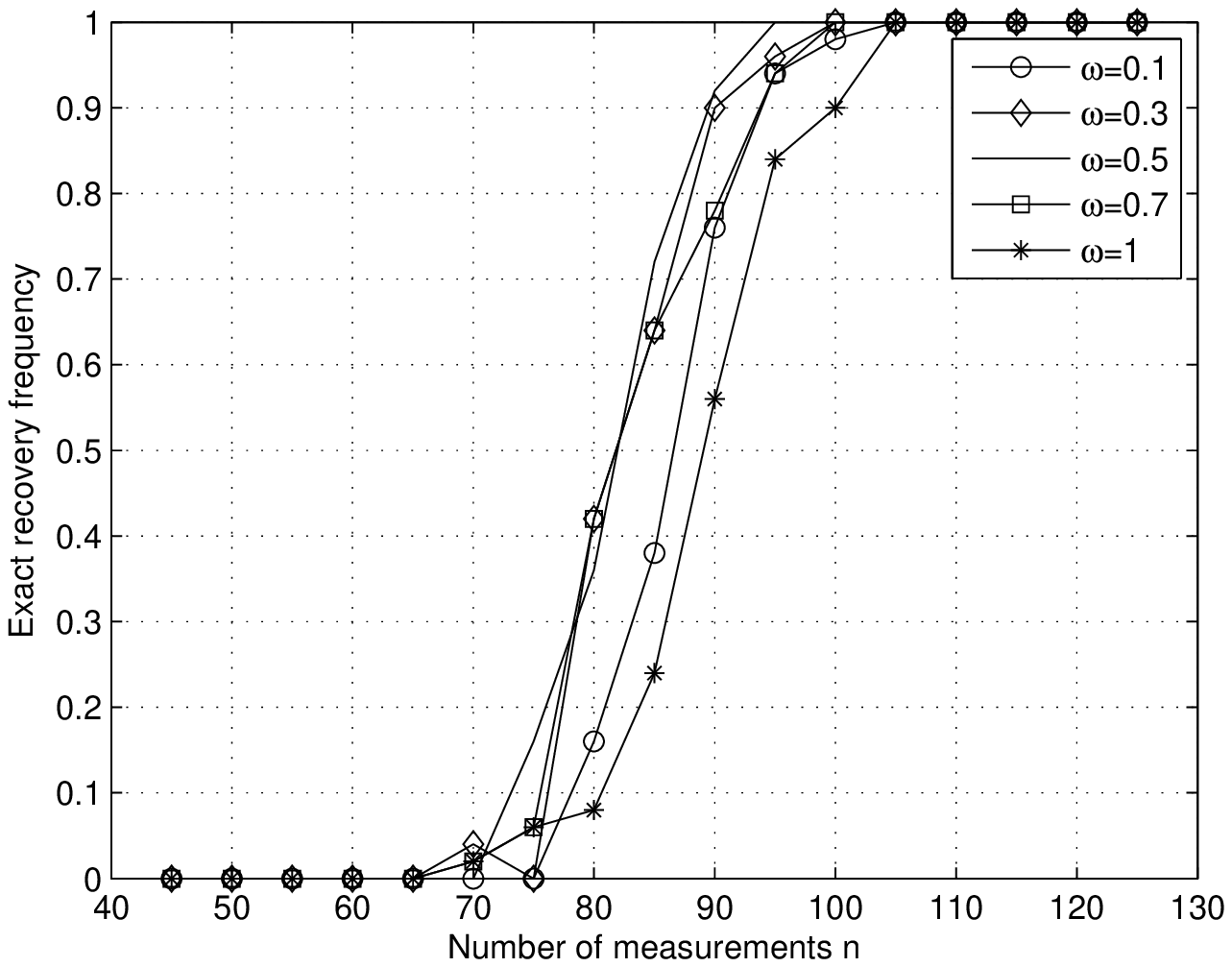}%\\[-0.3cm]
\end{minipage}
\centering
\begin{minipage}{4cm}
\centering
{$\alpha=0.2$}
\includegraphics[width=4cm,height=6cm]{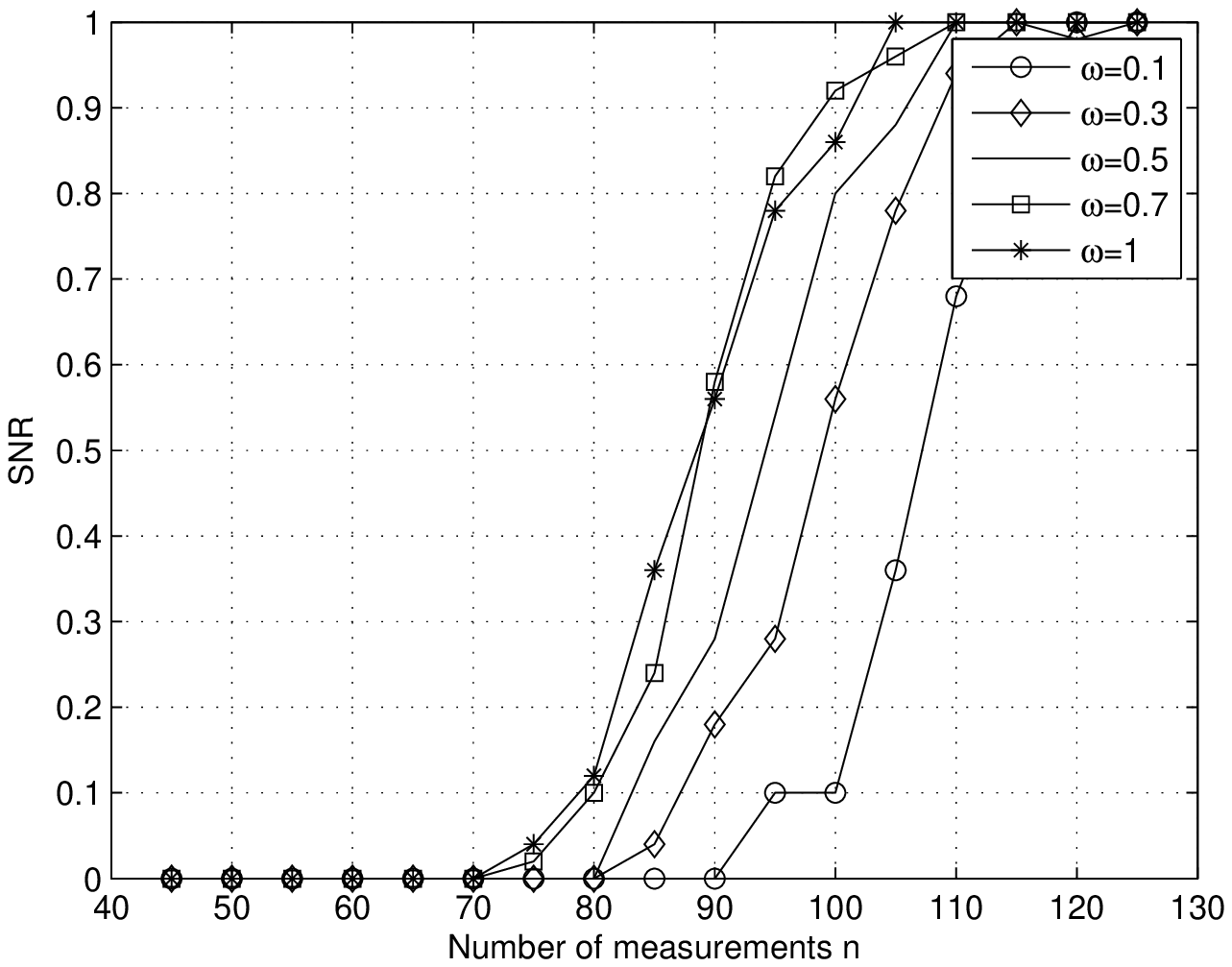}%\\[-0.3cm]
\end{minipage}}

\subfigure[Noisy Case]{
\centering
\begin{minipage}{4cm}
\centering
{$\alpha=0.8$}
\includegraphics[width=4cm,height=6cm]{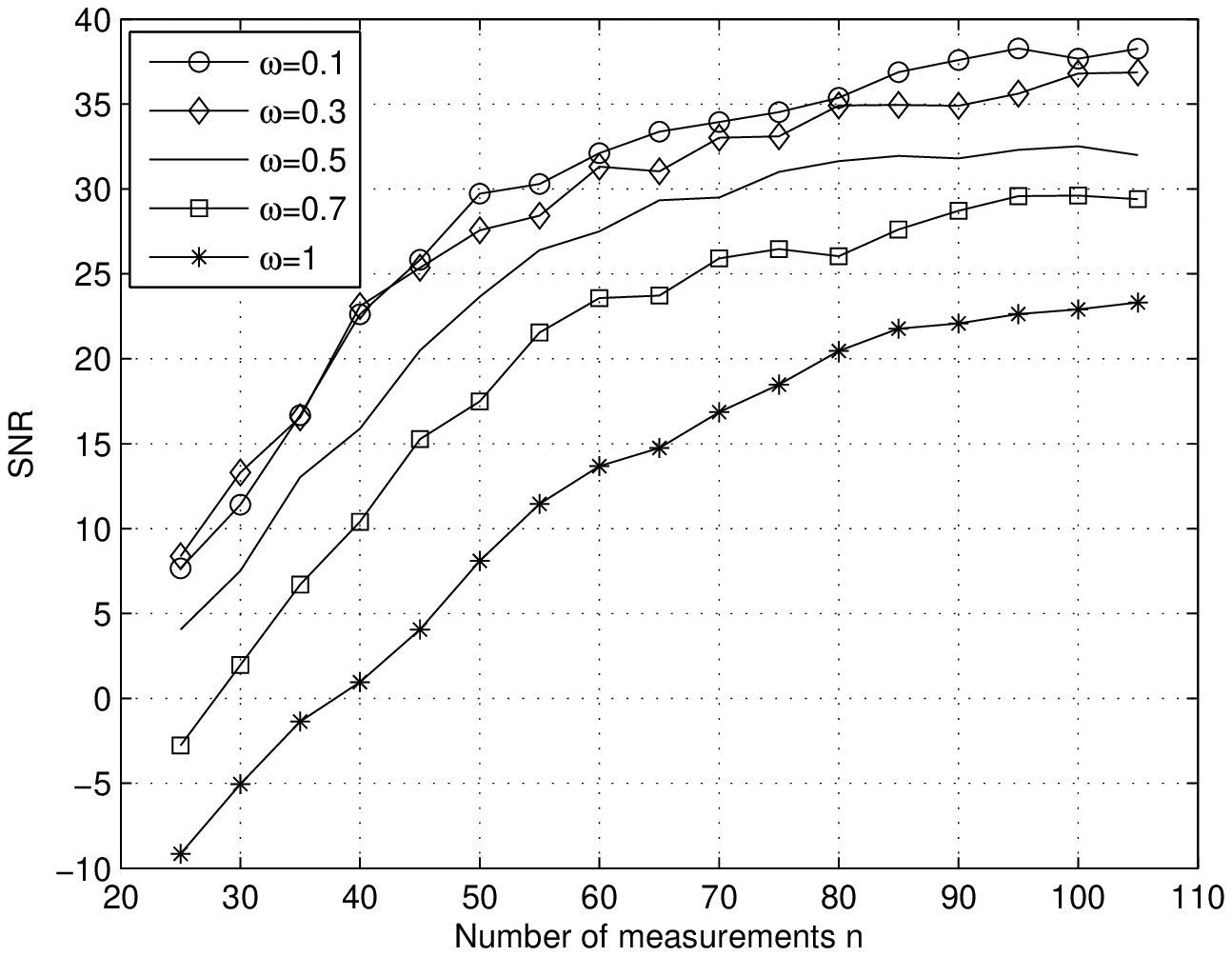}\\[-0.3cm]
\end{minipage}
\centering
\begin{minipage}{4cm}
\centering
{$\alpha=0.5$}
\includegraphics[width=4cm,height=6cm]{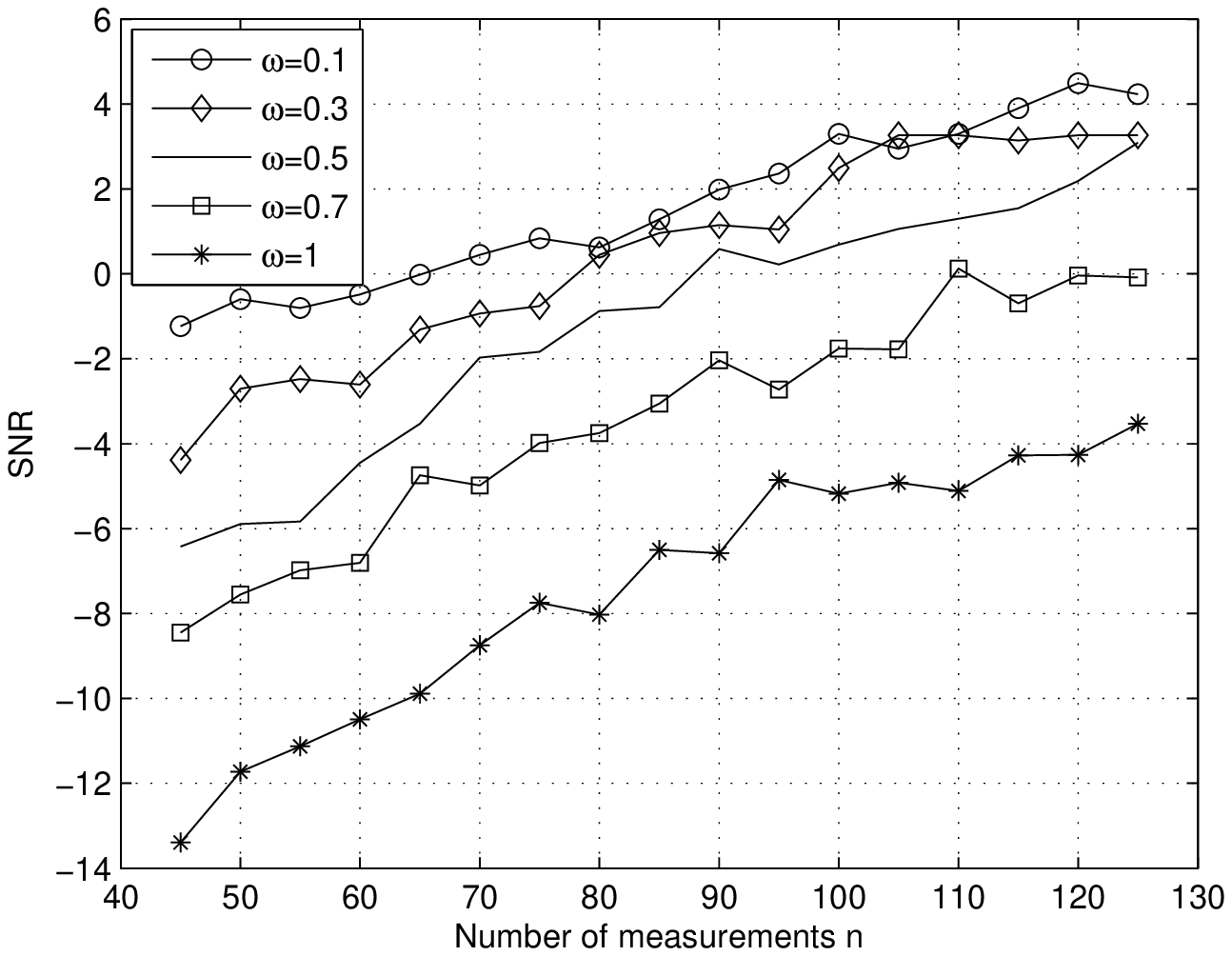}\\[-0.3cm]
\end{minipage}
\centering
\begin{minipage}{4cm}
\centering
{$\alpha=0.2$}
\includegraphics[width=4cm,height=6cm]{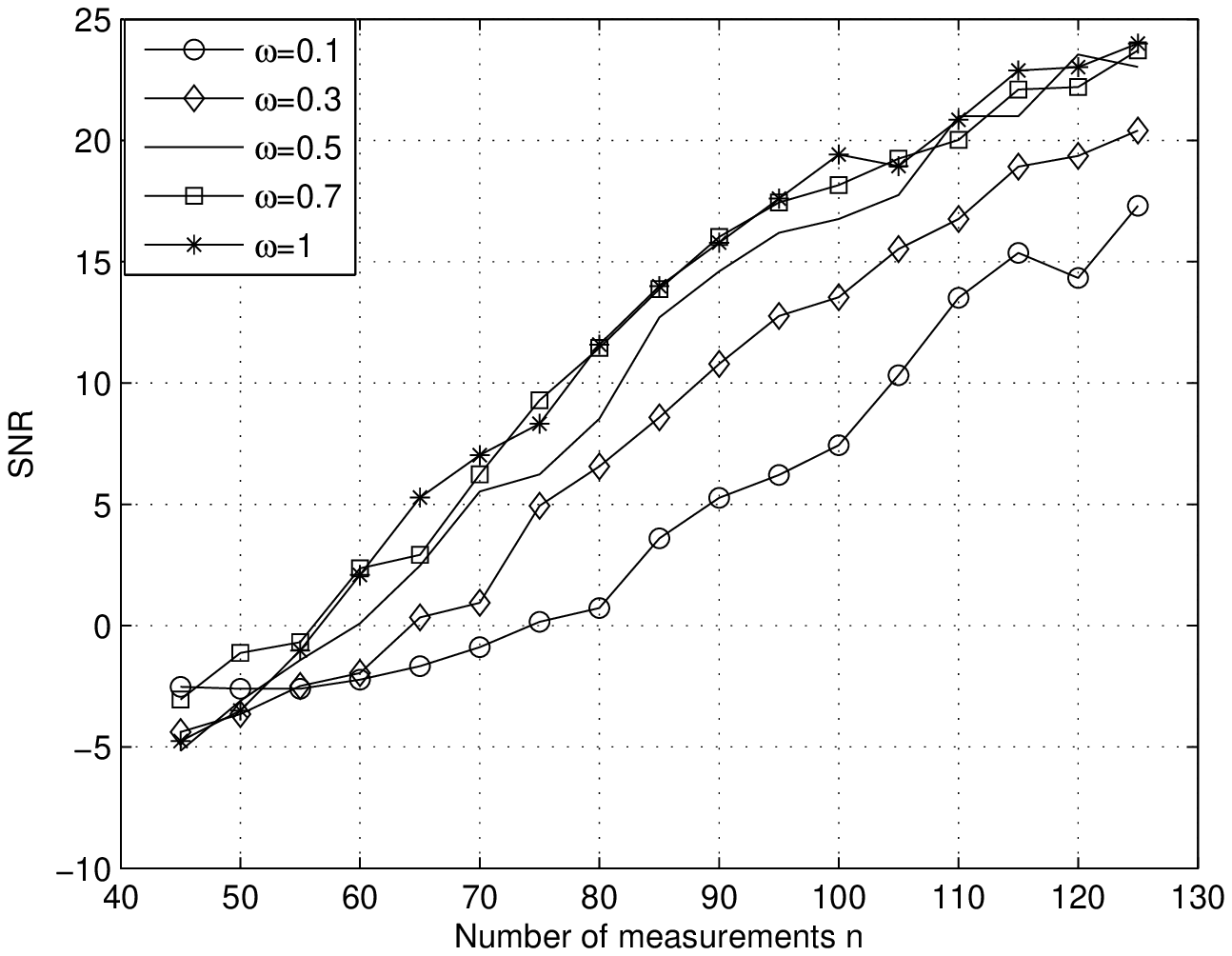}\\[-0.3cm]
\end{minipage}}
\\
\caption{\label{fig6} The recovery performance of the weighted $\ell_2/\ell_1$ minimization is in terms of the exact recovery frequency in noiseless case and the SNR  in
 noisy case. The block sparse signal $x$ with $\hat{d}=4$ has $k=10$ nonzero blocks. }
\end{figure*}

\begin{figure*}[htbp!]
\centering
\subfigure[Noise Free]{
\begin{minipage}{4cm}
\centering
{$\alpha=0.8$}
\includegraphics[width=4cm,height=6cm]{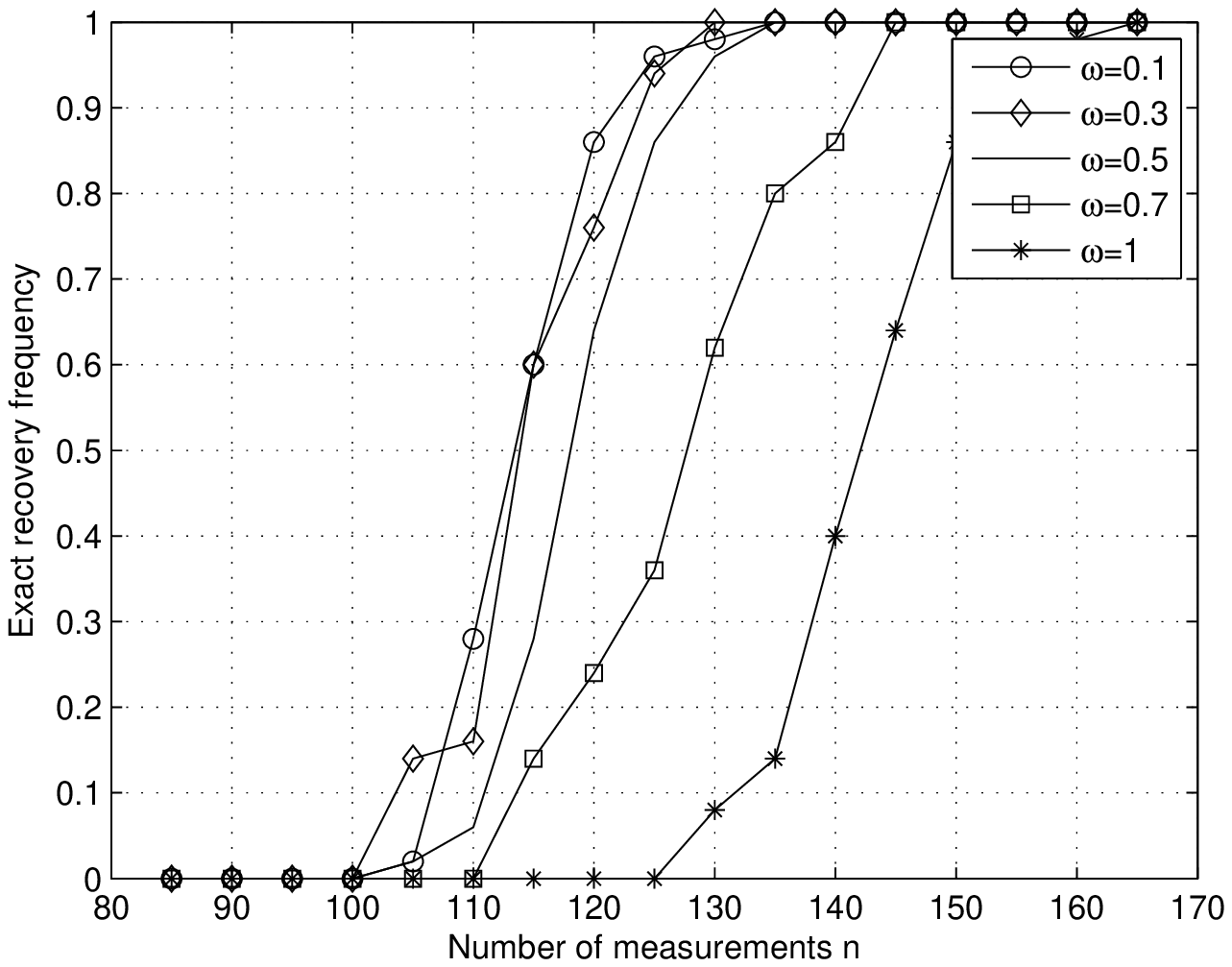}%\\[-0.3cm]
\end{minipage}
\centering
\begin{minipage}{4cm}
\centering
{$\alpha=0.5$}
\includegraphics[width=4cm,height=6cm]{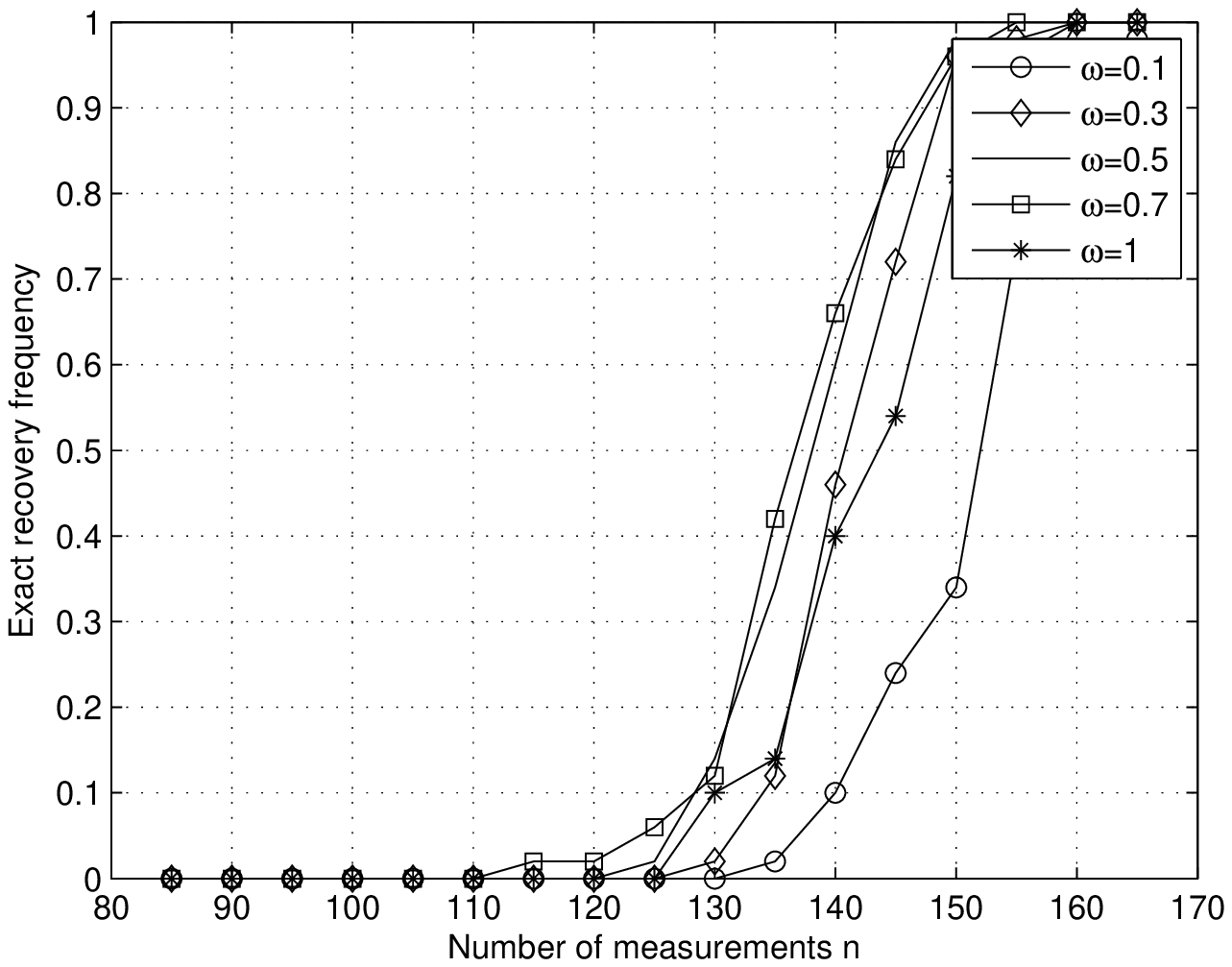}%\\[-0.3cm]
\end{minipage}
\centering
\begin{minipage}{4cm}
\centering
{$\alpha=0.2$}
\includegraphics[width=4cm,height=6cm]{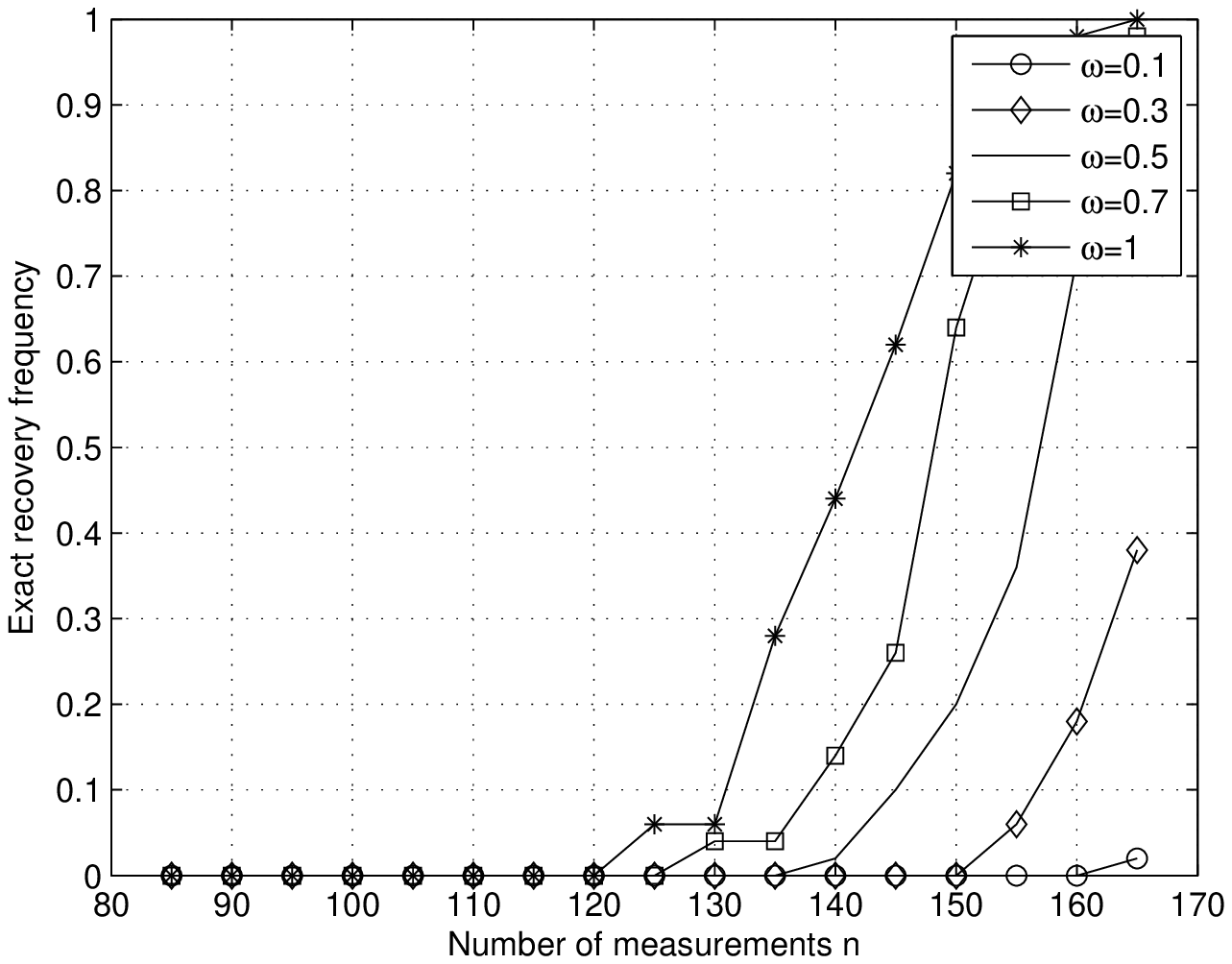}%\\[-0.3cm]
\end{minipage}}

\subfigure[Noisy Case ]{
\centering
\begin{minipage}{4cm}
\centering
{$\alpha=0.8$}
\includegraphics[width=4cm,height=6cm]{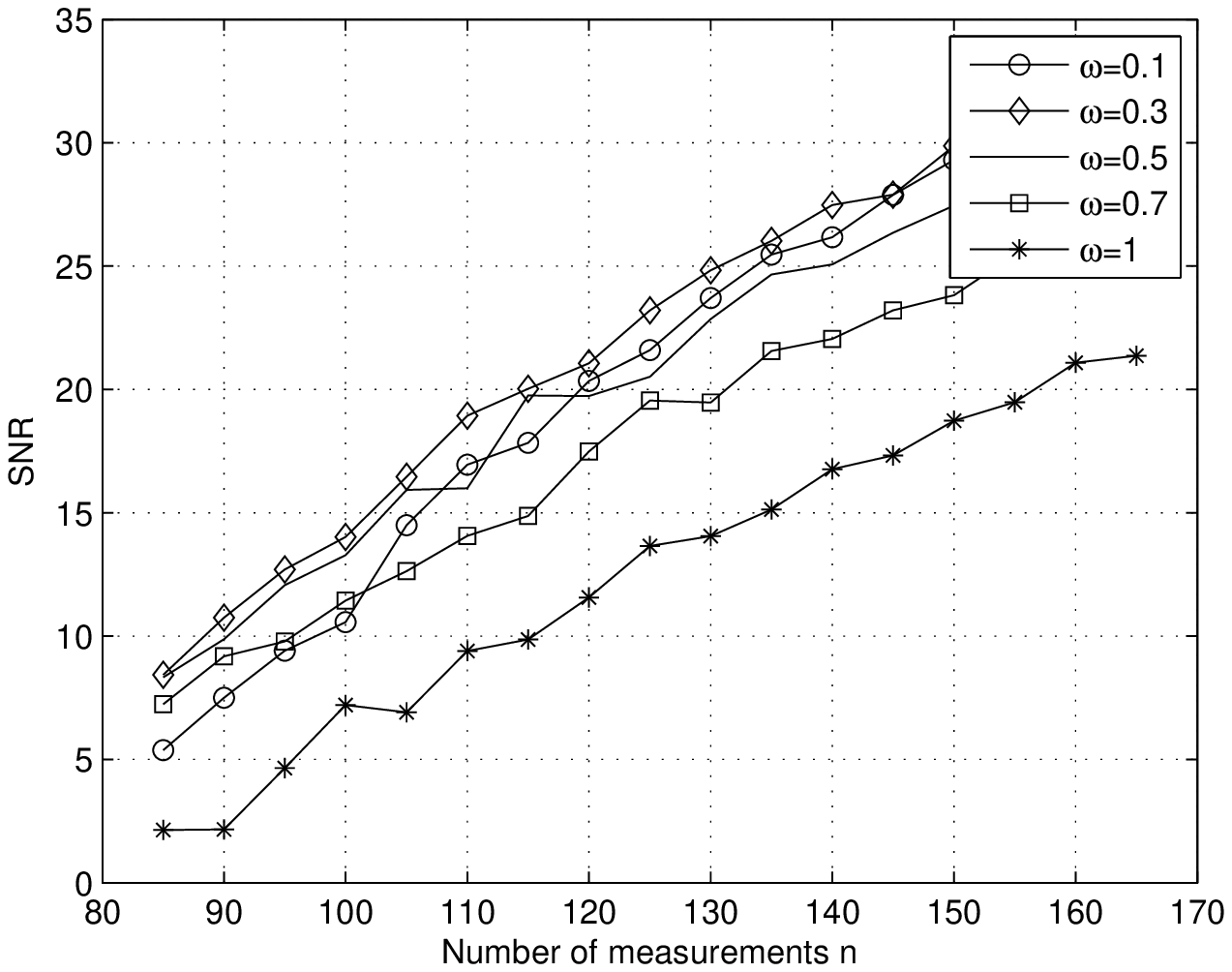}\\[-0.3cm]
\end{minipage}
\centering
\begin{minipage}{4cm}
\centering
{$\alpha=0.5$}
\includegraphics[width=4cm,height=6cm]{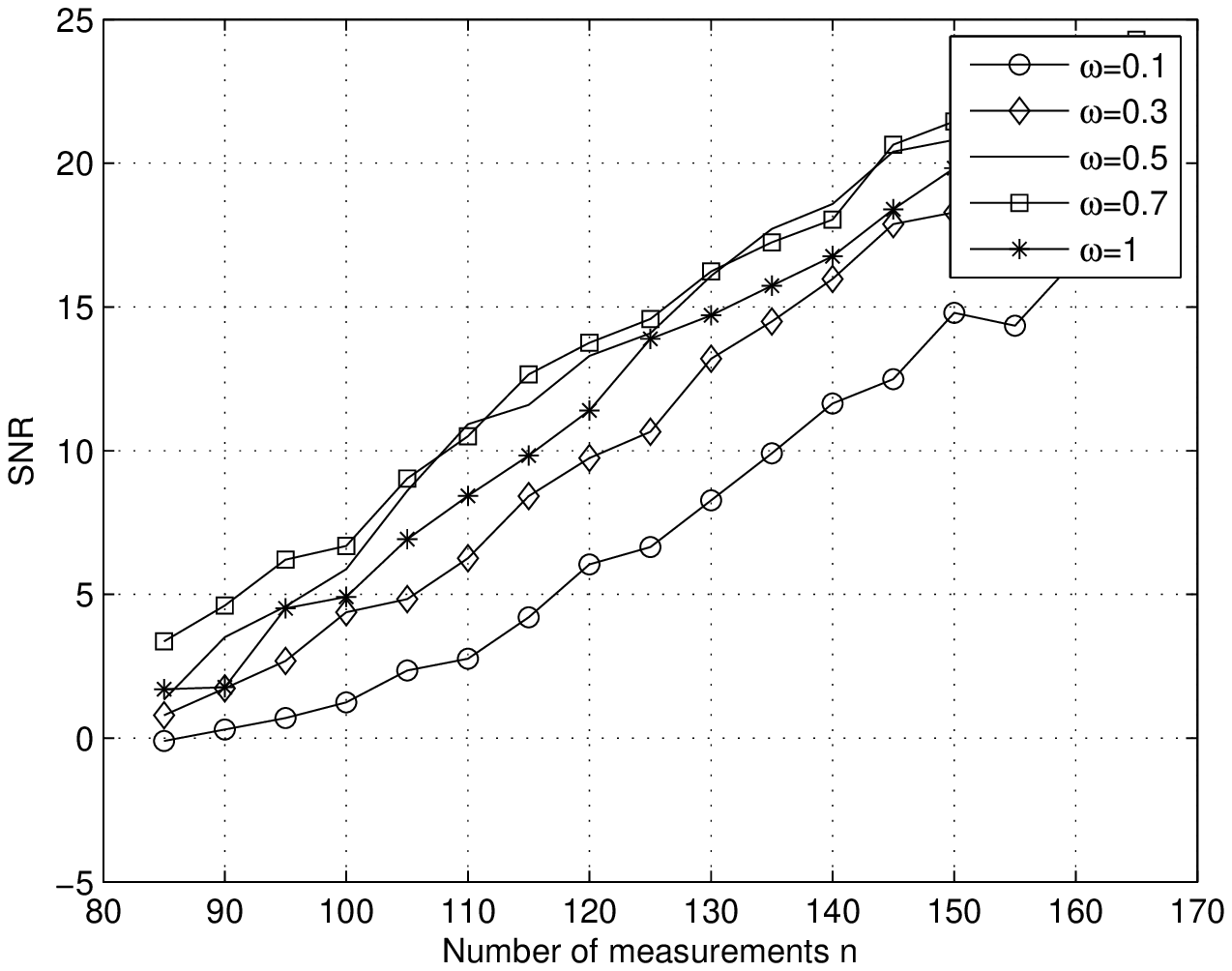}\\[-0.3cm]
\end{minipage}
\centering
\begin{minipage}{4cm}
\centering
{$\alpha=0.2$}
\includegraphics[width=4cm,height=6cm]{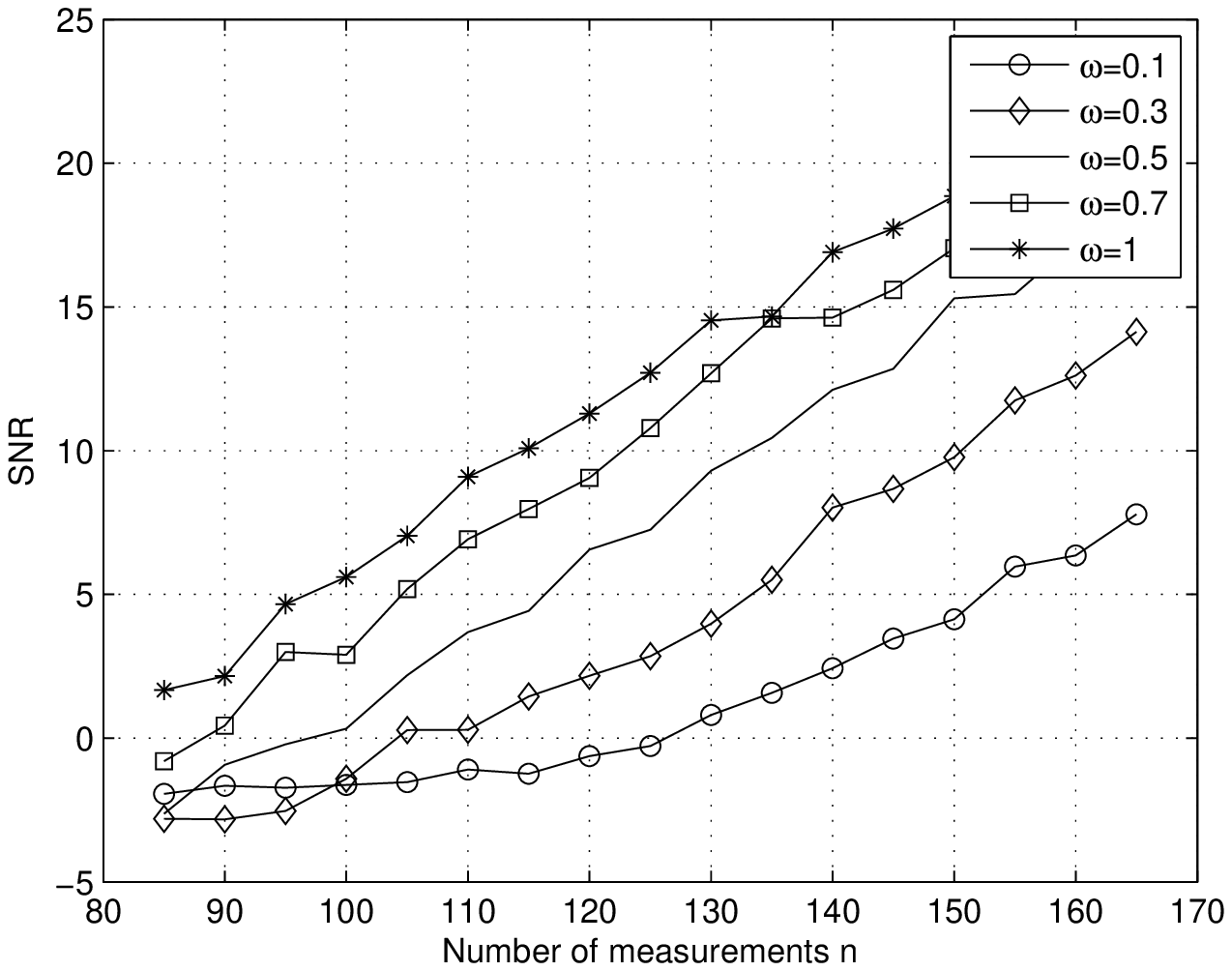}\\[-0.3cm]
\end{minipage}}
\\
\caption{\label{fig7} The recovery performance of the weighted $\ell_2/\ell_1$ minimization is in terms of the exact recovery frequency in noiseless case and the SNR  in
 noisy case. The block sparse signal $x$ with $\hat{d}=8$ has $k=10$ nonzero blocks. }
\end{figure*}

In Fig.\ref{fig5}(a)-\ref{fig7}(a),  the average exact recovery frequency is plotted versus measurement level $n$ for the accuracy of
the prior block support estimate: $\alpha=0.8,\alpha=0.5,\alpha=0.2,$ which illustrates the reconstruction   performance  of  the block $k$-sparse signal $x$
with three different block sizes $\hat{d}=2,\ \hat{d}=4$ and $\hat{d}=8$ in the noiseless case. Fig.\ref{fig5}(b)-\ref{fig7}(b)
depict the case of recovering the block $k$-sparse signal $x$ with three different block sizes $\hat{d}=2,\ \hat{d}=4$ and $\hat{d}=8$ in the presence of noise by the SNR. When the accuracy of the prior block support estimate $\alpha\geq0.5$,
 one can easily see that the best recovered performance is achieved for weight $\omega=0.1$ whereas  weight $\omega=1$
 results in the worst
exact recovery frequency in the noiseless case and  the worst SNR in
the noisy case. In addition, reducing the uniform  weight $\omega$
below $1$ reduces the number of measurements required for the
recovery of $x$. On the other hand, the exact recovery frequency and
the SNR are shifted towards larger weights $\omega$ for small $n$ as
$\alpha<0.5$, which means that the performance of the reconstruction
algorithm is shifted. In a word,  applying a larger prior block
support estimate favors better recovery and  the experimental
results are consistent with our theoretical results in Theorem
\ref{t3}.
 And it is also shown that the curves are very close for different weights $\omega$ when $\alpha=0.5$. Proposition
\ref{pro6} explained  why  the phenomenon happens.
\subsection{The non-uniform weight  case}

 In this subsection, we demonstrate that multiple weights can be preferable to a single weight by designing  serval numerical experiments.
In these experiments, we set $N=256$, $k=|T|=20$ and $\hat{d}=2$.
 We compare the exact recovery frequency
  and SNR when applying either a single prior block support $\widetilde{T}$
  or two disjoint prior block supports $\widetilde{T}_1$ and $\widetilde{T}_2$ satisfying $\widetilde{T}=\widetilde{T}_1+\widetilde{T}_2$, which means $\rho=\rho_1+\rho_2$.

\begin{figure*}[htbp!]
\centering
\begin{minipage}[htbp]{0.46\linewidth}
\centering
\includegraphics[width=3.0in]{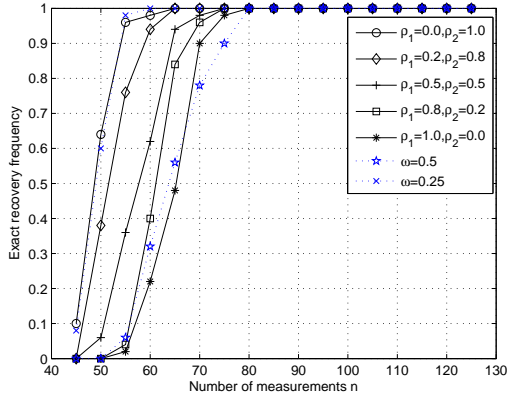}\\[-0.3cm]
{(a) Noise Free}
\end{minipage}
\centering
\begin{minipage}[htbp]{0.46\linewidth}
\centering
\includegraphics[width=3.0in]{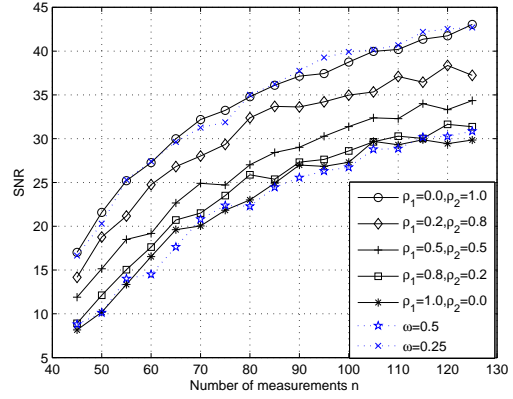}\\[-0.3cm]
{(b) Noisy Case }
\end{minipage}
\\
\caption{\label{fig3} Comparison of the exact recovery frequency and the SNR over 50 trials versus the number
of measurements $n$ while using the weighted $\ell_2/\ell_1$-minimization with  a single weight $\omega$  (blue dotted lines)
  and two distinct weights $\omega_1$ and $\omega_2$ (black solid lines).
Let $\rho_1+\rho_2=\rho=1$ and $\alpha=\alpha_1=\alpha_2=0.5$.}
\end{figure*}
\begin{figure*}[htbp!]
\centering
\begin{minipage}[htbp]{0.46\linewidth}
\centering
\includegraphics[width=3.0in]{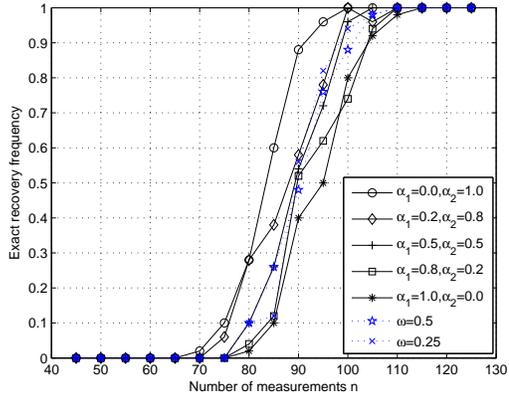}\\[-0.3cm]
{(a) Noise Free}
\end{minipage}
\centering
\begin{minipage}[htbp]{0.46\linewidth}
\centering
\includegraphics[width=3.0in]{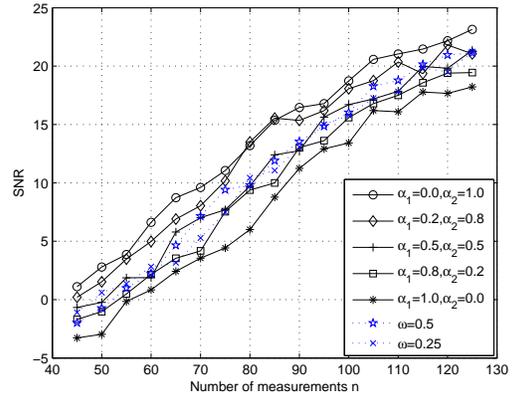}\\[-0.3cm]
{(b) Noisy Case}
\end{minipage}
\caption{\label{fig4}Under $\alpha_1+\alpha_2=1$, $\alpha=0.5$ and $\rho_1=\rho_2=0.5$,  we compare
  the exact recovery frequency in the  noiseless case and the SNR in the noisy case over 50 trials versus  the number
of measurements $n$  while using the weighted $\ell_2/\ell_1$-minimization with  a single weight $\omega$ (blue dotted lines) and two distinct weights $\omega_1$ and $\omega_2$  (black solid lines).}
\end{figure*}

   In Fig.\ref{fig3}, we set $\alpha_1=\alpha_2=\alpha=0.5$, $\rho_1+\rho_2=\rho=1$ and vary the size of
    $\rho_1$ and $\rho_2$. For the two prior block supports $\widetilde{T}_1$ and $\widetilde{T}_2$,
     $\widetilde{T}_1$ applies the larger weight $\omega_1=0.5$ and $\widetilde{T}_2$ applies the smaller weight $\omega_2=0.25$.
     In the single prior block support case, the weight $\omega=0.5$ or $\omega=0.25$ is used on $\widetilde{T}$.
  Fig.\ref{fig3} (a) and (b) plot the exact recovery frequency  and the SNR versus the number of measurements
  $n$, respectively. One can see that using the smaller weight $\omega=0.25$ prefers the best, using the larger weight $\omega=0.5$ prefers the worst and using two different weights $\omega_1=0.5$ and $\omega_2=0.25$
   as $\rho_1$ and $\rho_2$ are varied produces intermediate  performance.

   Fig.\ref{fig4} (a) and (b) depict the exact recovery frequency and the SNR  versus the number of measurements
   $n$ for some different $\alpha_1$ and $\alpha_2$ maintaining $\rho_1\alpha_1+\rho_2\alpha_2=\rho\alpha$, where
   $\alpha=0.5,\ \rho=1$ and $\rho_1=\rho_2=0.5$, which imply $\alpha_1+\alpha_2=1$.  The weights $\omega_1=0.5$
   and $\omega_2=0.25$ are applied on $\widetilde{T}_1$ and $\widetilde{T}_2$, respectively. Note that $\widetilde{T}_1\subseteq T^c$
   and $\widetilde{T}_2\subseteq T$ when $\alpha_1=0.0$ and $\alpha_2=1.0$. As expected, we observe that
   the  exact recovery frequency  and  SNR are largest when $\alpha_1=0.0$ and $\alpha_2=1.0$ implying that
   the recovery performance of the weighted $\ell_2/\ell_1$-minimization in this case is best from Fig.\ref{fig4} (a) and (b).
   As $\alpha_1$ increases from $0$ to 1 and $\alpha_2$ decreases from $1$ to $0$,
   fewer correctly identified block indexes in $T$ receive the smaller weight $\omega_2=0.25$,
   but rather the larger weight $\omega_1=0.5$.  Moreover, we also see that the values of the  exact recovery frequency  and the SNR are very close, as using a single weight
   $\omega=0.5$ or $\omega=0.25$. In fact, the recovery is slightly better applying the single
   weight $\omega=0.25$ than that using $\omega=0.5$, and the recovery falls in between the $\omega=0.25$ and $\omega=0.5$ curves when the two weights $\omega_1=0.5$ and $\omega_2=0.25$ are used with $\alpha_1=\alpha_2=0.5$.
   Here, we turn to make the fact clear in the theory. Based on Theorem \ref{t3} and Proposition \ref{pro6} (1) and (3), for the case of $\omega=0.5$, $\omega_1=0.5$ and $\omega_2=0.25$ we need to compare the terms
   \begin{align*}
   \omega\|x[T^c]\|_{2,1}+(1-\omega)\|x[\widetilde{T}^c\cap T^c]\|_{2,1}
   =0.5\|x[T^c]\|_{2,1}+0.5\|x[\widetilde{T}^c\cap T^c]\|_{2,1}
   \end{align*}
   and
    \begin{align*}
   &(\omega_1+\omega_2)\|x[T^c]\|_{2,1}+(1-(\omega_1+\omega_2))\|x[\widetilde{T}^c\cap T^c]\|_{2,1}-\omega_2\|x[\widetilde{T}_1\cap T^c]\|_{2,1}-\omega_1\|x[\widetilde{T}_2\cap T^c]\|_{2,1}\\
   &=0.75\|x[T^c]\|_{2,1}+0.25\|x[\widetilde{T}^c\cap T^c]\|_{2,1}-0.25\|x[\widetilde{T}_1\cap T^c]\|_{2,1}-0.5\|x[\widetilde{T}_2\cap T^c]\|_{2,1}\\
   &\leq0.5\|x[T^c]\|_{2,1}+0.5\|x[\widetilde{T}^c\cap T^c]\|_{2,1}-0.5\|x[\widetilde{T}_2\cap T^c]\|_{2,1}
   \end{align*}
   where we use  $\|x[\widetilde{T}_1\cap T^c]\|_{2,1}+\|x[\widetilde{T}_2\cap T^c]\|_{2,1}=\|x[\widetilde{T}\cap T^c]\|_{2,1}=\|x[T^c]\|_{2,1}-\|x[\widetilde{T}^c\cap T^c]\|_{2,1}$.
  It is obvious that
  \begin{align*}
  &(\omega_1+\omega_2)\|x[T^c]\|_{2,1}+(1-(\omega_1+\omega_2))\|x[\widetilde{T}^c\cap T^c]\|_{2,1}-\omega_2\|x[\widetilde{T}_1\cap T^c]\|_{2,1}-\omega_1\|x[\widetilde{T}_2\cap T^c]\|_{2,1}\\
  &\leq\omega\|x[T^c]\|_{2,1}+(1-\omega)\|x[\widetilde{T}^c\cap T^c]\|_{2,1}.
  \end{align*}
 Similarly, for $\omega=0.25$, $\omega_1=0.5$ and $\omega_2=0.25$ there is
 \begin{align*}
 &(\omega_1+\omega_2)\|x[T^c]\|_{2,1}+(1-(\omega_1+\omega_2))\|x[\widetilde{T}^c\cap T^c]\|_{2,1}-\omega_2\|x[\widetilde{T}_1\cap T^c]\|_{2,1}-\omega_1\|x[\widetilde{T}_2\cap T^c]\|_{2,1}\\
  &\geq\omega\|x[T^c]\|_{2,1}+(1-\omega)\|x[\widetilde{T}^c\cap T^c]\|_{2,1}.
  \end{align*}

\section{Conclusion} \label{6}

In this paper, the problem of reconstructing unknown block sparse signals under arbitrary prior block support information is studied from incomplete linear measurements.
Firstly, we introduce the weighted $\ell_2/\ell_1$ minimization and obtain a high order block RIP condition to guarantee stable and robust recovery of block signals in bounded $\ell_2$
 noise setting. The condition is weaker than that of block sparse signals by the standard $\ell_2/\ell_1$ minimization when all of  the accuracy of $L$ disjoint prior block support
estimates  are at least  $50\%$.
Secondly, we determine how many random measurements are needed to fulfill  the high order block RIP condition
 with high probability for some random  matrices.
Finally, a series of numerical experiments have been carried out to  illustrate the benefit of using the weighted $\ell_2/\ell_1$ minimization to recover block
sparse signals when  prior block support information is available
and that non-uniform block support information  can be preferable to
 uniform  block support information.
%\section{Acknowledgements}

\end{document}